\PassOptionsToPackage{svgnames,dvipsnames,svgnames}{xcolor}
\newif\ifarxiv
\arxivtrue
\ifarxiv
\documentclass[acmsmall,screen,nonacm]{acmart}
\else
\documentclass[acmsmall,screen]{acmart}
\fi

\ifarxiv
\newcommand{\appendixName}{appendix}
\else
\newcommand{\appendixName}{extended appendix}
\fi

\usepackage{booktabs}   %% For formal tables:
\usepackage{subcaption} %% For complex figures with subfigures/subcaptions
\usepackage{microtype}
\usepackage{mdframed}
\usepackage{colortab}
\usepackage{mathpartir}
\usepackage{enumitem}
\usepackage{bbm}
\usepackage{stmaryrd}
\usepackage{mathtools}
\usepackage{leftidx}
\usepackage{todonotes}
\usepackage{xspace}
\usepackage{wrapfig}
\usepackage{extarrows}

\usepackage{listings}%
\newcommand{\li}[1]{\lstinline[basicstyle=\ttfamily\fontsize{9pt}{1em}\selectfont]{#1}}
\newcommand{\lismall}[1]{\lstinline[basicstyle=\ttfamily\fontsize{9pt}{1em}\selectfont]{#1}}

\def\OPTIONConf{1}%
\usepackage{joshuadunfield}

\usepackage{blindtext}

\usepackage{enumitem}

\usepackage[export]{adjustbox}% http://ctan.org/pkg/adjustbox
\newcommand{\HazelnutLive}{\textsf{Hazelnut Live}\xspace}
\newcommand{\Hazelnut}{\textsf{Hazelnut}\xspace}
\newcommand{\Hazel}{\textsf{Hazel}\xspace}

\newtheoremstyle{slplain}% name
  {.15\baselineskip\@plus.1\baselineskip\@minus.1\baselineskip}% Space above
  {.15\baselineskip\@plus.1\baselineskip\@minus.1\baselineskip}% Space below
  {\slshape}% Body font
  {\parindent}%Indent amount (empty = no indent, \parindent = para indent)
  {\bfseries}%  Thm head font
  {.}%       Punctuation after thm head
  { }%      Space after thm head: " " = normal interword space;
  {}%       Thm head spec
\theoremstyle{slplain}
\newtheorem{thm}{Theorem}  % Numbered with the equation counter
\numberwithin{thm}{section}
\newtheorem{defn}[thm]{Definition}
\newtheorem{lem}[thm]{Lemma}
\newtheorem{prop}[thm]{Proposition}
\newtheorem{corol}[thm]{Corollary}

\ifarxiv
\setcopyright{rightsretained} 
\acmJournal{PACMPL}
\acmYear{2019} \acmVolume{3} \acmNumber{POPL} \ifarxiv \acmArticle{1} \else \acmArticle{14} \fi \acmMonth{1} \acmPrice{}\acmDOI{10.1145/3290327}
\copyrightyear{2019}
\else
\setcopyright{rightsretained}
\acmPrice{}
\acmDOI{10.1145/3290327}
\acmYear{2019}
\copyrightyear{2019}
\acmJournal{PACMPL}
\acmVolume{3}
\acmNumber{POPL}
\acmArticle{14}
\acmMonth{1}
\fi

\newcommand{\mynote}[3]{\textcolor{#3}{\textsf{{#2}}}}
\newcommand{\rkc}[1]{\mynote{rkc}{#1}{blue}}

\newcommand{\mattOmit}[1]{\colorbox{yellow}{(Matt omitted stuff here)}}

\def\parahead#1{\paragraph{\textbf{#1.}}}
\newcommand{\Elm}{\ensuremath{\textsf{Elm}}}
\newcommand{\sns}{\ensuremath{\textrm{Sketch-n-Sketch}}}

\newcommand{\defeq}{\overset{\textrm{def}}{=}}

\newcommand{\Property}[1]{\textrm{#1}}

\newcommand{\codeSize}
  {\small}

\newcommand{\JoinTypes}[2]{\textsf{join}(#1,#2)}

\newcommand{\vsepRuleHeight}{0.08in}
\newcommand{\vsepRule}{\vspace{\vsepRuleHeight}}

\newcommand{\cmttclo}[2]{\mathsf{clo}(#1, #2)}

\newcommand{\CaptionLabel}[2]{
  \caption{#1}
  \label{#2}}

\newcommand{\llparenthesiscolor}{\textcolor{violet}{\llparenthesis}}
\newcommand{\rrparenthesiscolor}{\textcolor{violet}{\rrparenthesis}}
\newcommand{\isComplete}[1]{#1~\mathsf{complete}}

\newcommand{\htau}{\tau}
\newcommand{\tarr}[2]{#1 \rightarrow #2}

\newcommand{\tnum}{\texttt{num}}

\newcommand{\tehole}{\llparenthesiscolor\rrparenthesiscolor}
\newcommand{\tsum}[2]{{#1} + {#2}}

\newcommand{\tconsistent}[2]{#1 \sim #2}

\newcommand{\hexp}{e}
\newcommand{\hlam}[2]{\lambda #1.#2}
\newcommand{\halam}[3]{\lambda #1{:}#2.#3}
\newcommand{\hap}[2]{#1(#2)}
\newcommand{\hnum}[1]{\underline{#1}}
\newcommand{\hadd}[2]{#1 + #2}
\newcommand{\hehole}[1]{\llparenthesiscolor\rrparenthesiscolor^{#1}}
\newcommand{\hhole}[2]{\llparenthesiscolor#1\rrparenthesiscolor^{#2}}
\newcommand{\hinL}[1]{\mathsf{inl}(#1)}
\newcommand{\hinR}[1]{\mathsf{inr}(#1)}
\newcommand{\hcase}[5]{\texttt{case}({#1},{#2}.{#3},{#4}.{#5})}

\newcommand{\hGamma}{\Gamma}
\newcommand{\EmptyDelta}{\cdot} % From hand-written notes, ES-Const rule, page 1
\newcommand{\domof}[1]{\text{dom}(#1)}
\newcommand{\hsyn}[3]{#1 \vdash #2 \Rightarrow #3}
\newcommand{\hana}[3]{#1 \vdash #2 \Leftarrow #3}

\newcommand{\removeSel}[1]{#1^{\diamond}}

\newcommand{\zexp}{\hat{e}}

\newcommand{\performSyn}[6]{#1 \vdash #2 \Rightarrow #3 \xlongrightarrow{#4} #5 \Rightarrow #6}

\newcommand{\arrmatch}[2]{#1 \blacktriangleright_{\rightarrow} #2}
\newcommand{\groundmatch}[2]{#1 \blacktriangleright_{\mathsf{ground}} #2}

\newcommand{\summatch}[2]{#1 \blacktriangleright_{+} #2}

\newcommand{\mvar}[0]{u}
\newcommand{\subst}[0]{\sigma}
\newcommand{\substitute}[3]{[#1/#2]#3}
\newcommand{\fvof}[1]{\mathsf{FV}(#1)}
\newcommand{\dexp}[0]{d}

\newcommand{\dcasttwo}[3]{#1 \langle{#2}\Rightarrow{#3}\rangle}
\newcommand{\dcastthree}[4]
  {#1 \langle{#2}\Rightarrow{#3}\Rightarrow{#4}\rangle} %% sugared version
\newcommand{\dcastfail}[3]{#1 \langle{#2}\Rightarrow{\tehole}\not\Rightarrow{#3}\rangle}
\newcommand{\dlam}[3]{\halam{#1}{#2}{#3}}
\newcommand{\dap}[2]{#1(#2)}
\newcommand{\dapP}[2]{(#1)(#2)} % Extra paren around function term
\newcommand{\dnum}[1]{\underline{#1}}
\newcommand{\dadd}[2]{#1 + #2}
\newcommand{\dehole}[3]{\leftidx{^{#3}}{\llparenthesiscolor\rrparenthesiscolor}{^{#1}_{#2}}}
\newcommand{\dhole}[4]{\leftidx{^{#4}}{\llparenthesiscolor#1\rrparenthesiscolor}{^{#2}_{#3}}}

\newcommand{\dinL}[2]{\mathsf{inl}_{#1}(#2)}
\newcommand{\dinR}[2]{\mathsf{inr}_{#1}(#2)}
\newcommand{\dcase}[5]{\texttt{case}({#1},{#2}.{#3},{#4}.{#5})}

\newcommand{\elabAna}[6]{#1 \vdash #2 \Leftarrow #3 \leadsto #4 : #5 \dashv #6}
\newcommand{\elabSyn}[5]{#1 \vdash #2 \Rightarrow #3 \leadsto #4 \dashv #5}
\newcommand{\hasType}[4]{#1; #2 \vdash #3 : #4}
\newcommand{\isValue}[1]{#1~\mathsf{val}}
\newcommand{\isGround}[1]{#1~\mathsf{ground}}
\newcommand{\isBoxedValue}[1]{#1~\mathsf{boxedval}}
\newcommand{\isIndet}[1]{#1~\mathsf{indet}}
\newcommand{\isFinal}[1]{#1~\mathsf{final}}

\newcommand{\stepsToD}[3]{#2 \mapsto #3}
\newcommand{\multiStepsTo}[2]{#1 \mapsto^* #2}

\newcommand{\hDelta}{\Delta}
\newcommand{\Dunion}[2]{#1 \cup #2}
\newcommand{\idof}[1]{\mathsf{id}(#1)}
\newcommand{\Dbinding}[3]{#1 :: #3[#2]}
\newcommand{\instantiate}[3]{\llbracket#1 / #2\rrbracket #3}

\newcommand{\evalctx}{\mathcal{E}}
\newcommand{\evalhole}{\circ}
\newcommand{\isevalctx}[1]{#1~\mathsf{evalCtx}}
\newcommand{\reducesE}[3]{#2 \longrightarrow #3}
\newcommand{\selectEvalCtxR}[2]{#1\{#2\}}
\newcommand{\selectEvalCtx}[3]{#1=\selectEvalCtxR{#2}{#3}}
\newcommand{\maybePremise}[1]{{\textcolor{red}[}#1{\textcolor{red}]}}

\newcommand{\inhole}[2]{\mathsf{inhole}(#1; #2)}

\newcommand{\DoSubst}[3]{[#1/#2]{#3}}

\begin{document}

%% Title information
\title{Live Functional Programming with Typed Holes}         %% [Short Title] is optional;
\ifarxiv
\subtitle{Extended Version}
\subtitlenote{The original version of this article was published in the POPL 2019 edition of PACMPL \cite{HazelnutLive}. This extended version includes an additional appendix.}
\fi
                                        %% when present, will be used in

                                        %% header instead of Full Title.
% \titlenote{with title note}             %% \titlenote is optional;
                                        %% can be repeated if necessary;
                                        %% contents suppressed with 'anonymous'
% \subtitle{Subtitle}                     %% \subtitle is optional
% \subtitlenote{with subtitle note}       %% \subtitlenote is optional;
                                        %% can be repeated if necessary;
                                        %% contents suppressed with 'anonymous'

%% Author information
%% Contents and number of authors suppressed with 'anonymous'.
%% Each author should be introduced by \author, followed by
%% \authornote (optional), \orcid (optional), \affiliation, and
%% \email.
%% An author may have multiple affiliations and/or emails; repeat the
%% appropriate command.
%% Many elements are not rendered, but should be provided for metadata
%% extraction tools.

%% Author with single affiliation.
\author{Cyrus Omar}
% \authornote{with author1 note}          %% \authornote is optional;
                                        %% can be repeated if necessary
% \orcid{nnnn-nnnn-nnnn-nnnn}             %% \orcid is optional
\affiliation{
  % \position{Position1}
  % \department{Department1}              %% \department is recommended
  \institution{University of Chicago, USA}            %% \institution is required
  % \streetaddress{Street1 Address1}
  % \city{City1}
  % \state{State1}
  % \postcode{Post-Code1}
  % \country{Country1}
}
\email{comar@cs.uchicago.edu}          %% \email is recommended

\author{Ian Voysey}
% \authornote{with author1 note}          %% \authornote is optional;
                                        %% can be repeated if necessary
% \orcid{nnnn-nnnn-nnnn-nnnn}             %% \orcid is optional
\affiliation{
  % \position{Position1}
  % \department{Department1}              %% \department is recommended
  \institution{Carnegie Mellon University, USA}            %% \institution is required
  % \streetaddress{Street1 Address1}
  % \city{City1}
  % \state{State1}
  % \postcode{Post-Code1}
  % \country{Country1}
}
\email{iev@cs.cmu.edu}          %% \email is recommended

\author{Ravi Chugh}
% \authornote{with author1 note}          %% \authornote is optional;
                                        %% can be repeated if necessary
% \orcid{nnnn-nnnn-nnnn-nnnn}             %% \orcid is optional
\affiliation{
  % \position{Position1}
  % \department{Department1}              %% \department is recommended
  \institution{University of Chicago, USA}            %% \institution is required
  % \streetaddress{Street1 Address1}
  % \city{City1}
  % \state{State1}
  % \postcode{Post-Code1}
  % \country{Country1}
}
\email{rchugh@cs.uchicago.edu}          %% \email is recommended

\author{Matthew A. Hammer}
% \authornote{with author1 note}          %% \authornote is optional;
                                        %% can be repeated if necessary
% \orcid{nnnn-nnnn-nnnn-nnnn}             %% \orcid is optional
\affiliation{
  % \position{Position1}
  % \department{Department1}              %% \department is recommended
  \institution{University of Colorado Boulder, USA}            %% \institution is required
  % \streetaddress{Street1 Address1}
  % \city{City1}
  % \state{State1}
  % \postcode{Post-Code1}
  % \country{Country1}
}
\email{matthew.hammer@colorado.edu}          %% \email is recommended

% %% Author with two affiliations and emails.
% \author{First2 Last2}
% \authornote{with author2 note}          %% \authornote is optional;
%                                         %% can be repeated if necessary
% \orcid{nnnn-nnnn-nnnn-nnnn}             %% \orcid is optional
% \affiliation{
%   \position{Position2a}
%   \department{Department2a}             %% \department is recommended
%   \institution{Institution2a}           %% \institution is required
%   \streetaddress{Street2a Address2a}
%   \city{City2a}
%   \state{State2a}
%   \postcode{Post-Code2a}
%   \country{Country2a}
% }
% \email{first2.last2@inst2a.com}         %% \email is recommended
% \affiliation{
%   \position{Position2b}
%   \department{Department2b}             %% \department is recommended
%   \institution{Institution2b}           %% \institution is required
%   \streetaddress{Street3b Address2b}
%   \city{City2b}
%   \state{State2b}
%   \postcode{Post-Code2b}
%   \country{Country2b}
% }
% \email{first2.last2@inst2b.org}         %% \email is recommended

%% Paper note
%% The \thanks command may be used to create a "paper note" ---
%% similar to a title note or an author note, but not explicitly
%% associated with a particular element.  It will appear immediately
%% above the permission/copyright statement.
% \thanks{with paper note}                %% \thanks is optional
                                        %% can be repeated if necesary
                                        %% contents suppressed with 'anonymous'

%% Abstract
%% Note: \begin{abstract}...\end{abstract} environment must come
%% before \maketitle command
% !TEX root = hazelnut-dynamics.tex

\begin{abstract}
Live programming environments aim to provide programmers (and sometimes audiences) 
with continuous feedback about a program's dynamic behavior as it is being edited. 
The problem is that programming languages typically assign dynamic meaning only 
to programs that are \emph{complete}, i.e. syntactically well-formed and free
of type errors. Consequently,    
live feedback presented to the programmer exhibits temporal or perceptive gaps. 

This paper confronts this ``{gap problem}'' from type-theoretic first principles by developing 
\emph{a dynamic semantics for incomplete functional programs}, 
starting from the static semantics for incomplete functional programs developed in recent work on \Hazelnut. 
We model incomplete functional programs as expressions with \emph{holes}, 
with empty holes standing for missing expressions or types, and  non-empty holes 
operating as membranes around static and dynamic type inconsistencies. 
Rather than aborting when evaluation encounters any of these holes as in
some existing systems, evaluation proceeds around holes,
tracking the 
closure around each hole instance as it flows through the remainder of the program. Editor services can use the information in these hole closures 
to help the programmer develop and confirm their mental model of the behavior of the complete portions of the program as they decide how to fill the remaining holes. 
Hole closures also enable a \emph{fill-and-resume} operation that 
avoids the need to restart evaluation after edits that amount to hole filling. 
Formally, the semantics borrows machinery from both gradual type theory (which supplies the basis for handling unfilled 
type holes) and contextual modal type theory (which supplies a
logical basis for hole closures), combining these and developing additional machinery necessary 
to continue evaluation past holes while maintaining type safety. We have mechanized the metatheory of the core calculus, called \HazelnutLive{}, using the Agda proof assistant.

We have also implemented these ideas into the \Hazel programming environment. The implementation inserts holes automatically, following the \Hazelnut edit action calculus, to guarantee that every editor state has some (possibly incomplete)
type.
Taken together with this paper's type safety property, the 
result is  
a proof-of-concept live programming environment where rich dynamic feedback 
is truly available without gaps, i.e. for every reachable editor state.

%In our system, evaluation treats failed casts much like it treats expression holes (rather than immediately failing with a cast error, as in gradual type theory). Prior work on contextual modal type theory did not develop an operational semantics for programs with free metavariables, which correspond to programs with holes in our formulation. 

% These incomplete edit states are sometimes transient, but at other times, they persist for extended periods of time, e.g. when the programmer is filling in the branches of a large case analysis one-by-one, or when an edit causes type errors to appear throughout a program. %This gap can cause live programming services to lag behind the programmer's edits, in many cases for an extended period of time.
\end{abstract}

%% 2012 ACM Computing Classification System (CSS) concepts
%% Generate at 'http://dl.acm.org/ccs/ccs.cfm'.
% <concept_id>10011007.10011006.10011008</concept_id>
% <concept_desc>Software and its engineering~General programming languages</concept_desc>
% <concept_id>10003456.10003457.10003521.10003525</concept_id>
% <concept_desc>Social and professional topics~History of programming languages</concept_desc>
% <concept_significance>300</concept_significance>
\begin{CCSXML}
<ccs2012>
<concept>
<concept_id>10011007.10011006.10011008.10011009.10011012</concept_id>
<concept_desc>Software and its engineering~Functional languages</concept_desc>
<concept_significance>500</concept_significance>
</concept>
% <concept_id>10003752.10010124.10010131.10010134</concept_id>
% <concept_desc>Theory of computation~Operational semantics</concept_desc>
</ccs2012>
\end{CCSXML}

\ccsdesc[500]{Software and its engineering~Functional languages}
% \ccsdesc[500]{Theory of computation~Operational semantics}
%% End of generated code

%% Keywords
%% comma separated list
\keywords{live programming, gradual typing, contextual modal type theory, typed holes, structured editing}  %% \keywords is optional

%% \maketitle
%% Note: \maketitle command must come after title commands, author
%% commands, abstract environment, Computing Classification System
%% environment and commands, and keywords command.
\maketitle
% \thispagestyle{empty} % suppresses the footer

% !TEX root = hazelnut-dynamics.tex
\ifarxiv \clearpage \fi
\newcommand{\introSec}{Introduction}
\section{\introSec} 
\label{sec:intro}

% Programming environments often operate in ``batch mode'', assuming that the programmer will spend a substantial amount 
% of time editing the program text blindly before evaluating (i.e. running) the program. 

Programmers typically shift between program editing and program evaluation many times before converging upon a program that behaves as intended. 
Live programming environments aim to granularly interleave editing and evaluation so as to   
narrow what \citet{burckhardt2013s} call the ``temporal and perceptive gap'' between these activities.
% In other words, the goal \IS to provide continuous feedback about the dynamic behavior of the program, in whole or in part,
% directly alongside the program text itself.

For example, read-evaluate-print loops (REPLs) and derivatives thereof, like the IPython/Jupyter lab notebooks popular in data science~\cite{PER-GRA:2007}, allow the programmer to edit and immediately execute program fragments organized into a sequence of cells. 
Spreadsheets are live functional dataflow environments, with cells organized into a grid \cite{DBLP:journals/jfp/Wakeling07}. 
More specialized examples include live direct manipulation programming environments like SuperGlue
\cite{McDirmid:2007}, \sns{}~\cite{sns-pldi,sns-uist}, and the tools
demonstrated by \citet{victor2012inventing} in his lectures;
live GUI frameworks \cite{burckhardt2013s};
live image processing languages~\cite{DBLP:journals/vlc/Tanimoto90};
and live visual and auditory dataflow languages \cite{DBLP:conf/vl/BurnettAW98}, which can support live coding as performance art \cite{DBLP:journals/programming/ReinRLHP19}.
Editor-integrated debuggers \cite{mccauley2008debugging} and editors that support inspecting run-time state, like Smalltalk environments \cite{Goldberg:1983cn}, are also live programming environments. 
% Live programming, in its various incarnations \cite{DBLP:journals/vlc/Tanimoto90,DBLP:conf/icse/Tanimoto13},
%has been and continues to be an active area of research and development.
% has been, and continues to be, an active area of research and development.

% \matt{ ``The problem that specifically motivates this paper'' is a
% really long noun phrase; the sentence is more interesting, IMO, when
% the subject is 'programming languages' and not that noun phrase.  }
%
The problem at the heart of this paper is that
programming languages typically assign meaning only to {complete programs}, i.e. programs that are syntactically well-formed and free of type and binding errors. Programming environments, however, often encounter incomplete, and therefore conventionally meaningless, editor states. As a result, live feedback either ``flickers out'', creating a temporal gap, or it ``goes stale'', i.e. it relies on the most recent complete editor state, creating a perceptive gap because the feedback does not accurately reflect the editor state.

In some cases these gaps are brief, like while the programmer
is in the process of entering  
a short expression. In other cases, these gaps can persist over substantial lengths of time, such as when there are many branches of a case analysis whose bodies are initially left blank or when the programmer is puzzling over a mistake.
This is particularly problematic for novice programmers, who make more mistakes \cite{mccauley2008debugging,fitzgerald2008debugging}.
%
%\matt{we are not really modeling a language with user type definitions -- what does the following sentence really add here? also, I just added this scenerio to the list above}
The problem is also particularly pronounced for languages with rich static type systems where certain program changes, such as a change to a type definition, can cause type errors to propagate throughout the program. Throughout the process of addressing these errors, which can sometimes span many days, the program text remains formally meaningless. 
About 40\% of edits performed by Java programmers using Eclipse left the program text malformed \cite{popl-paper} and some additional number, which could not be determined from the data gathered by \citet{6883030}, were well-formed but ill-typed.

In recognition of this ``{gap problem}''---that incomplete programs are formally meaningless---\citet{popl-paper} describe a static semantics (i.e. a type system) for incomplete 
functional programs, modeling them formally as typed expressions with \emph{holes} in 
both expression and type position. 
Empty holes stand for missing expressions or types,
and non-empty holes operate as ``membranes'' around static type inconsistencies 
(i.e. they internalize the ``red underline'' that editors commonly display under a type inconsistency).
% For editor states into this language of incomplete programs, 
% editor services can reason about types and binding in many more situations than previously possible.
\citet{HazelnutSNAPL} discuss several ways to determine an incomplete expression from the editor state. When the editor state is a text buffer, error recovery mechanisms might insert holes implicitly \cite{DBLP:journals/siamcomp/AhoP72,charles1991practical,graham1979practical,DBLP:conf/oopsla/KatsJNV09}. Alternatively, the language might provide explicit syntax for holes, so that the programmer can insert them either manually  
or semi-automatically via a code completion service \cite{Amorim2016}. For example, GHC Haskell supports the notation \li{_u} for empty holes, where \li{u} is an optional hole name \cite{GHCHoles}. When the editor state is instead a tree or graph structure, i.e. in a structure editor, the editor inserts explicitly represented holes fully automatically \cite{popl-paper}; we say more about structure editors in Sec.~\ref{sec:implementation}.

A static semantics for incomplete programs is useful, but for the purposes of live programming, it does not suffice---%
we also need a corresponding dynamic semantics that specifies how to evaluate expressions with holes. That is the focus of this paper. %This paper addresses this need by developing
% a {dynamic semantics} for incomplete functional programs, 
% starting from the static semantics developed by \citet{popl-paper}.
%
%Our goal in this paper \IS to develop

The simplest approach would be to define a dynamic semantics that aborts with an error when evaluation reaches a hole. 
%%%%%%%%%%%% New par:
%
This mirrors a workaround that programmers commonly deploy: 
raising an exception as a placeholder, e.g. \lismall{raise Unimplemented}. 
GHC Haskell supports this ``exceptional approach'' using the \lismall{-fdefer-typed-holes} flag.\footnote{
Without this flag, holes cause compilation to fail with an error message that reports information about each hole's type and typing context. 
Proof assistants like Agda \cite{norell:thesis,norell2009dependently} and Idris \cite{brady2013idris} also respond to holes in this way.
} 
Although better than nothing, the exceptional approach to expression holes has limitations 
within a live programming environment because 
(1)~it provides no information about the behavior of the remainder of the program, 
parts of which may not depend on the missing or erroneous expression (e.g. subsequent cells in a lab notebook, or tests involving other components of the program);  
(2)~it provides limited information about the dynamic state of the program where the hole appears 
(typically only a stack trace);  and
(3)~it provides no means by which to resume evaluation after hole filling.

Furthermore, exceptions can appear only in expressions, but we might also like to be able to evaluate programs that have type holes. Again, existing approaches do not support this situation well---GHC supports type holes, but compilation fails if type inference cannot automatically fill them (such as when a variable is used at multiple types) \cite{GHCHoles}. 
The static semantics developed by \citet{popl-paper} provides more hope because it derives the machinery for 
reasoning about type holes from gradual type theory, 
identifying the type hole with the unknown type \cite{DBLP:conf/snapl/SiekVCB15,Siek06a}.
As such, we might look to the dynamic semantics from
gradual type theory, 
which inserts dynamic casts as necessary. 
%However, this \IS again somewhat dissatisfying from the perspective of live programming because when a cast fails, evaluation again stops with 
The only problem is that when a cast fails, evaluation stops with 
an exception, again leaving the live programming environment unable to provide rich, continuous feedback about the behavior of the remainder of the program.

\parahead{Contributions}

This paper develops a dynamic semantics for incomplete functional programs, starting from the static semantics developed by \citet{popl-paper},  that addresses the three limitations of the exceptional approach enumerated above.

In particular, rather than stopping when evaluation encounters an expression hole instance, evaluation continues ``around'' it.   
The system tracks the closure around each expression hole instance as evaluation proceeds. The live programming environment can feed the incomplete result and relevant information from these {hole closures} to the programmer to help them develop and confirm their mental model of the portions of the program that are complete as they work to fill the remaining holes. 
Then, when the programmer performs an edit that fills an empty expression hole or that replaces a non-empty hole with a type-correct expression, evaluation can resume from the previous evaluation state. We call this operation \emph{fill-and-resume}. For programs with unfilled type holes, casts are inserted as in gradual type theory (GTT) \cite{DBLP:conf/snapl/SiekVCB15} and, uniquely, evaluation proceeds around failed casts as it does around expression holes. 

The primary contribution of this paper is a simple type-theoretic account of this approach in the form of a core calculus, \HazelnutLive. We observe that expression hole closures are closely related to metavariable closures from contextual modal type theory~(CMTT)~\cite{Nanevski2008}, which, by its Curry-Howard correspondence with contextual modal logic, provides a logical basis for reasoning about and operating on hole closures. These connections to well-established systems (GTT and CMTT), together with mechanized proofs of the metatheory of \HazelnutLive, serve to support our main claim: that this approach to live programming is theoretically well-grounded.

There are many possible ways to present incomplete results and hole closure information to programmers. A secondary contribution of this paper is one  proof-of-concept user interface, which has been implemented into the \Hazel programming environment being developed by \citet{HazelnutSNAPL}. In particular, we describe \Hazel's novel live context inspector, which  combines static type information with hole closure information and interactively presents nested hole closures. We make no strong claims about this particular user interface; we evaluate it only with several suggestive example programming tasks where this interface presents information that would not otherwise be available and that we conjecture would be useful when teaching functional programming.

The editor component of \Hazel is organized around a language of structured edit actions, 
based on the \Hazelnut structure editor calculus developed by \citet{popl-paper}, that insert holes automatically to guarantee that
every editor state has some, possibly incomplete, type. 
The type safety invariant that we establish then guarantees that every editor state has dynamic meaning. Taken together, the result is an end-to-end solution to the gap problem, i.e. a proof-of-concept
live functional programming environment that continuously provides rich static and dynamic feedback.

\vspace{-2px}
\parahead{Paper Outline}

\newcommand{\contribution}[2]{\paragraph{#1. #2}}

We begin in Sec.~\ref{sec:examples} by detailing the approach informally, with several example programming tasks, in the setting of the \Hazel design. 
% We also qualitatively describe some user interface features that may help manage visual complexity when working with larger programs.

Sec.~\ref{sec:calculus} then abstracts away the inessential details of the language and user interface and makes the  intuitions developed in Sec.~\ref{sec:examples} formally precise by detailing the primary contribution of this paper: a core calculus, \HazelnutLive, that supports evaluating incomplete expressions and tracking hole closures. 
% The dynamic semantics of \HazelnutLive is a small-step reduction semantics equipped with clean type safety theorems (including, notably, a Progress theorem for incomplete expressions). 
Sec.~\ref{sec:agda-mechanization} outlines our Agda-based mechanization of \HazelnutLive, which is included in the archived artifact and is also available from the following URL: 

\begin{center}
\url{https://github.com/hazelgrove/hazelnut-dynamics-agda}
\end{center}

\noindent
Sec.~\ref{sec:implementation} states the continuity invariant, which formally solves the gap problem, as a corollary of the primary theorems of \Hazelnut and \HazelnutLive. It also provides some additional details on the implementation of \Hazel. A snapshot of the implementation is included in the archived artifact. An online version of the ongoing implementation of \Hazel is available from \url{hazel.org}, and the source code is available from the following URL:

\begin{center}
\url{https://github.com/hazelgrove/hazel}
\end{center}

Sec.~\ref{sec:resumption} defines the fill-and-resume operation, which is rooted in the contextual substitution operation from CMTT. We establish the correctness of fill-and-resume with a commutativity theorem. We also discuss how the fill-and-resume operation allows us to semantically interpret the act of editing and evaluating cells in a REPL or Jupyter-like live lab notebook environment.

Sec.~\ref{sec:relatedWork} describes related work in detail and simultaneously discusses limitations and a number of directions for future work. Sec.~\ref{sec:discussion} briefly concludes. 

\ifarxiv
The appendix 
\else
The extended version of the paper, which is available in the ArXiV \cite{2018arXiv180500155O}, includes an appendix that 
\fi
(1) provides some straightforward auxiliary definitions and proofs that were omitted from the \ifarxiv original \fi paper  for the sake of space; and (2) defines some simple extensions to the core calculus (namely, numbers, and sum types), together with a brief discussion on defining other extensions (in part by by following the ``gradualization'' approach of \citet{DBLP:conf/popl/CiminiS16}). 
% These are not necessary to communicate the fundamental ideas of the paper, but may be of interest to some readers.
% !TEX root = hazelnut-dynamics.tex
\newcommand{\examplesSec}{Live Programming in Hazel}
\section{\examplesSec} 
\label{sec:examples}

% !TEX root = hazelnut-dynamics.tex

\begin{figure}[t]
\begin{subfigure}[t]{\textwidth}
\centering
\includegraphics[width=\textwidth,interpolate=false]{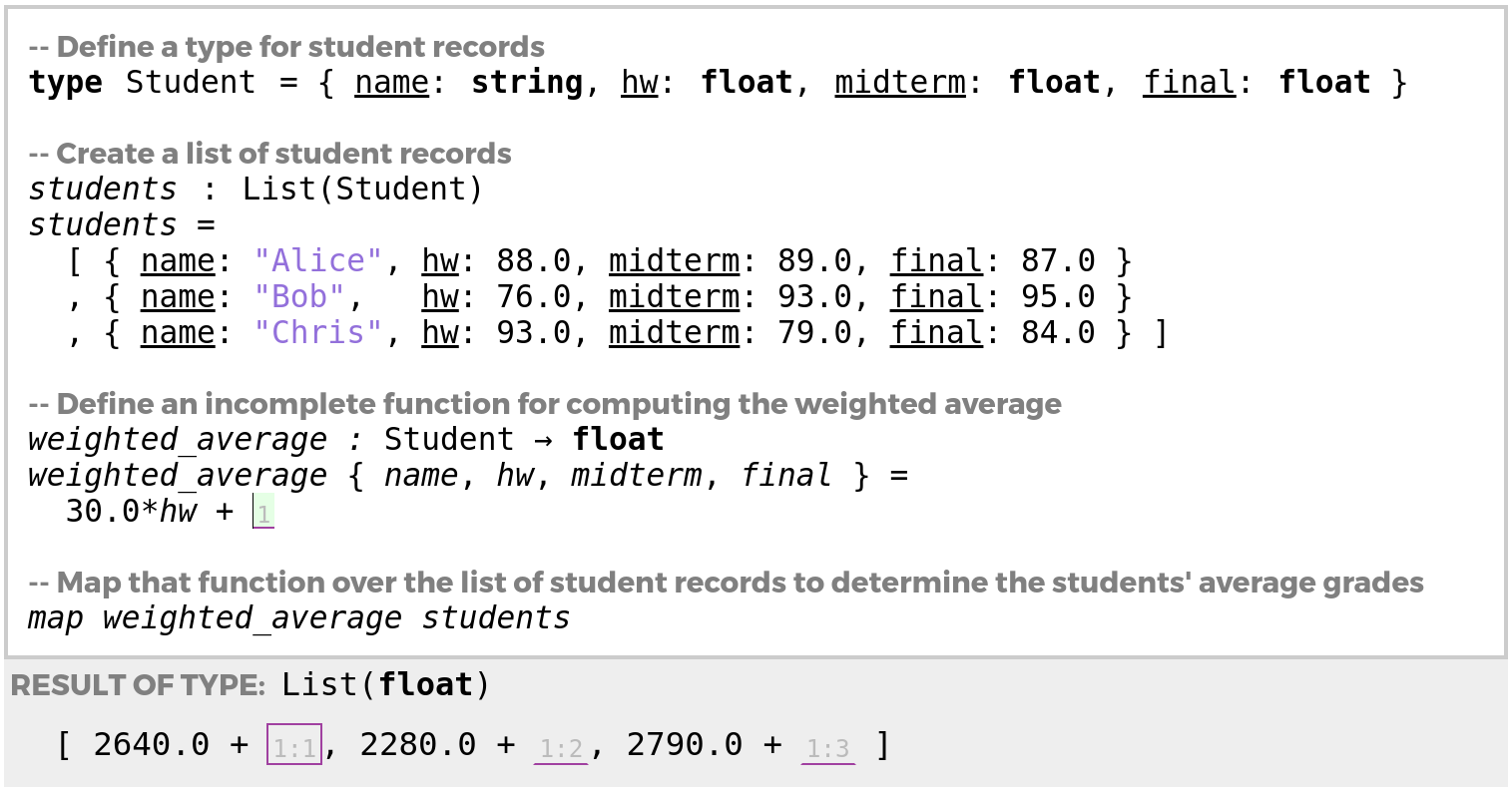}
\vspace{-10px}
\caption{Evaluating an incomplete functional program past the first hole}
\label{fig:grades-cell-mockup}
\end{subfigure}

\vspace{10px}

\begin{subfigure}[t]{\textwidth}
\centering
\includegraphics[width=0.31\textwidth,interpolate=false,valign=c]{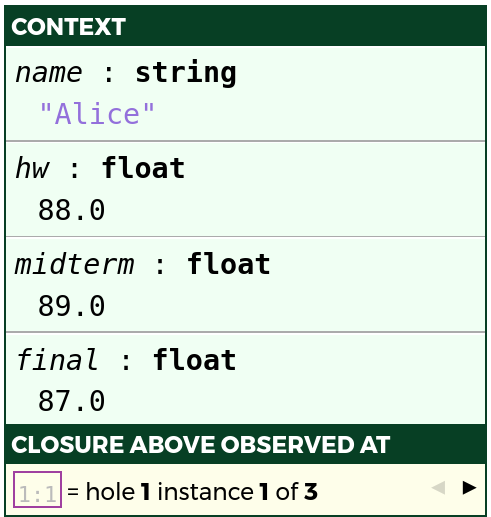}
~${}^\blacktriangleright$
\includegraphics[width=0.31\textwidth,interpolate=false,valign=c]{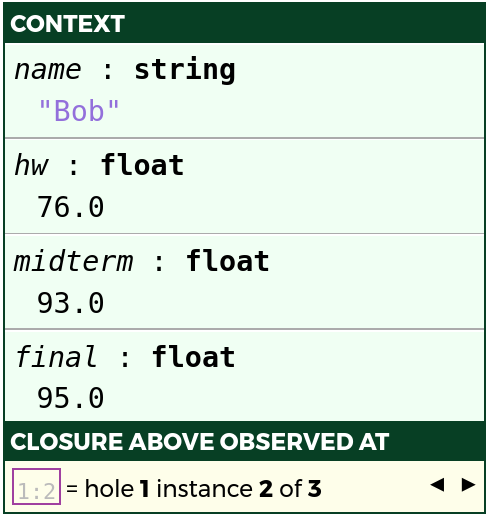}
~${}^\blacktriangleright$
\includegraphics[width=0.31\textwidth,interpolate=false,valign=c]{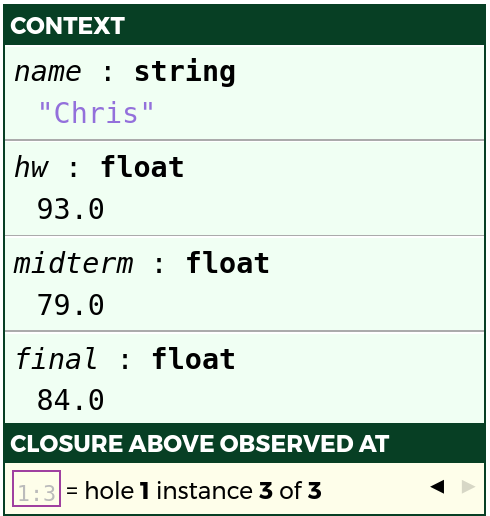}
\caption{The live context inspector communicates relevant static \emph{and} dynamic information about variables in scope.}
\label{fig:grades-sidebar}
\end{subfigure}
% %% TODO once the code above is removed, scale up the screenshots
% \includegraphics[scale=0.20]{images/hazel-placeholder-0.png}

% \rkc{Draw arrows and captions on the top figure to show how to get
% to the bottom figure.
% ser navigates to hole a, types + to create a plus, types * to create a
% multiplication, types \#10 to create 10, types vh1 to create variable use.}

% \includegraphics[scale=0.20]{images/hazel-placeholder-1a.png}
\vspace{3px}

\caption{Example 1: Grades}
\label{fig:grades-example}
\vspace{-5px}
\end{figure}

\newcommand{\overviewExample}[2]{\paragraph{Example {#1}: {#2}}}

Let us start with an example-driven overview of this paper's approach in \Hazel, a live programming environment being developed by \citet{HazelnutSNAPL}. The \Hazel user interface is based roughly on IPython/Jupyter \cite{PER-GRA:2007}, with a result appearing below each cell that contains an expression, and the \Hazel language is tracking toward feature parity with \Elm~(\url{elm-lang.org}) \cite{czaplicki2012elm,Elm}, a popular pure functional programming language similar to ``core ML'', with which we assume familiarity. \Hazel is intended initially for use by students and instructors in introductory functional programming courses (where \Elm~ has been successful \cite{DBLP:journals/corr/abs-1805-05125,zhang2018graphics}). 

For the sake of 
exposition, we have post-processed the screenshots in this section after generating them in \Hazel to make use of  ``syntactic and semantic sugar'' from \Elm~that was not available in \Hazel (which, as of this writing, implements little more than the language features described in Sec.~\ref{sec:calculus} and \ifarxiv Appendix \ref{sec:extensions}\else the appendix\fi). These conveniences are orthogonal to the contributions of this paper; all of the user interface features demonstrated in this section have been implemented and all of the computations can be expressed using standard ``encoding tricks''.

% !TEX root = hazelnut-dynamics.tex

\subsection{Example 1: Evaluating Past Holes and Hole Closures}

Consider the perspective of a teacher in the midst of developing a \Hazel{} notebook
to compute final student grades at the end of a course.
Fig.~\ref{fig:grades-cell-mockup} depicts the cell containing the incomplete program that 
the teacher has written so far (we omit irrelevant parts of the UI). 

At the top of this program, the teacher defines a
record type, \lismall{Student}, for recording a student's course data---here,
the student's name, of type
\lismall{string}, and, for simplicity, three grades, each of type \lismall{float}. Next, the teacher constructs a list of student records, binding it to the variable \lismall{students}. For simplicity, we include only three example students.
At the bottom of the program, the teacher maps a function \li{weighted_average} over this student data (\li{map} is the standard map function over lists, not shown), intending to compute a final weighted average for each student.
However, the program is incomplete because the teacher has not yet completed the body of the \lismall{weighted_average} function. This pattern is quite common: programmers often consume a function before implementing it.

Thus far in the body of \lismall{weighted_average}, the teacher has decomposed the function argument into variables by record destructuring, then multiplied the homework grade, \li{hw}, by \li{30.0} and finally inserted the \li{+} operator. 
In a conventional ``batch'' programming system, writing \li{30.0*hw +} by itself would simply cause parsing to fail and there would be no static or dynamic feedback.
In response, the programmer might temporarily raise an exception. This would cause typechecking to succeed. 
However, 
evaluation would proceed only as far as the first \li{map} iteration, 
which would call into \li{weighted_average} and then fail when the exception is raised.%
\footnote{In a lazy language, like Haskell, the result would be much the same because the environment forces the result for printing.}
% While the stack trace would confirm that \li{weighted_average} has been called, this is not entirely satisfying.
Hazel instead inserts an \emph{empty hole} at the cursor, as indicated by the vertical bar and the green background. Each hole has a unique name (generated automatically in \Hazel), here simply \li{1}. 

When running the program, \Hazel{} does not take an ``exceptional'' approach to holes. Instead, evaluation continues past the hole, treating it as an opaque expression of the appropriate type. The result, shown at the bottom of Fig.~\ref{fig:grades-cell-mockup}, is a list of length $3$, confirming that \li{map} does indeed behave as expected in this regard despite the teacher having provided an incomplete argument. Furthermore, each element of the resulting list has been evaluated as far as possible, i.e. the arithmetic expression \li{30.0*hw} has been evaluated for each corresponding value of \li{hw}. Evaluation cannot proceed any further because holes appear as addends. We say that each of these addition expressions is an \emph{indeterminate} sub-expression, and the result as a whole is therefore also indeterminate, because it is not yet a value, nor can it take a step due to holes in elimination positions.

At this point, the teacher might take notice of the magnitude of the numbers being computed, e.g. \li{2640.0} and \li{2280.0}, and
realize immediately that a mistake has been made: the teacher wants to compute a 
weighted average between \li{0.0} and \li{100.0}, and so the correct
constant is \li{0.30}, not \li{30.0}. %(There are of course other possible fixes.)

Although these observations might 
save only a small amount of time in this case, 
% relative to the 
% situation where the teacher notices this mistake only once  
% the \li{weighted_average} function has been completed, 
it demonstrates
the broader motivations of live programming: continuous feedback about the dynamic behavior of the 
program can help confirm the mental model 
that the programmer has developed (in this case, regarding the behavior of \li{map}), and also help quickly dispel 
misconceptions about the 
actual behavior of the program (in this case, the magnitude of the arithmetic expressions being computed).

% The result itself is not the only live feedback that \Hazel reports to the programmer.
% Programmers also keep track of the names and types of the variables that are in scope at the cursor. 
Going further, \Hazel helps programmers reason about the dynamic behavior of expressions bound to variables in scope at a hole via the \emph{live context inspector}, normally displayed  as a sidebar but shown disembodied in three states in Fig.~\ref{fig:grades-sidebar}. 
In all three states, the live context inspector displays the names and types of the variables that are in scope. 
New to our approach are the values associated with bindings, which come from the environment associated with the selected hole instance in the result. 
We call a hole instance paired with an  environment a \emph{hole closure}, by analogy to function closures (and also due to the logical connection detailed in Sec.~\ref{sec:calculus}). 
In this case, there are three instances of hole \li{1} in the result, numbered sequentially \li{1:1}, \li{1:2} and \li{1:3}, arising from the three calls to \li{weighted_averages} by \li{map}. 
% An environment is a partial mapping of variables in scope to values. 
The programmer can directly select different closures by clicking on a hole instance in the result (the first instance, \li{1:1}, is selected by default, as indicated by the purple outline in Fig.~\ref{fig:grades-cell-mockup}), or cycle through all of the closures for the hole at the cursor by clicking the arrows at the bottom of the context inspector, e.g. $\blacktriangleright$, or using the corresponding keyboard shortcuts.

\vspace{-4px}
\subsection{Example 2: Recursive Functions}\label{sec:qsort1}\label{sec:paths}
\vspace{-2px}

% !TEX root = hazelnut-dynamics.tex

\begin{figure}[t]
\begin{subfigure}[t]{0.70\textwidth}
\centering
\includegraphics[width=\textwidth,interpolate=false,valign=t]{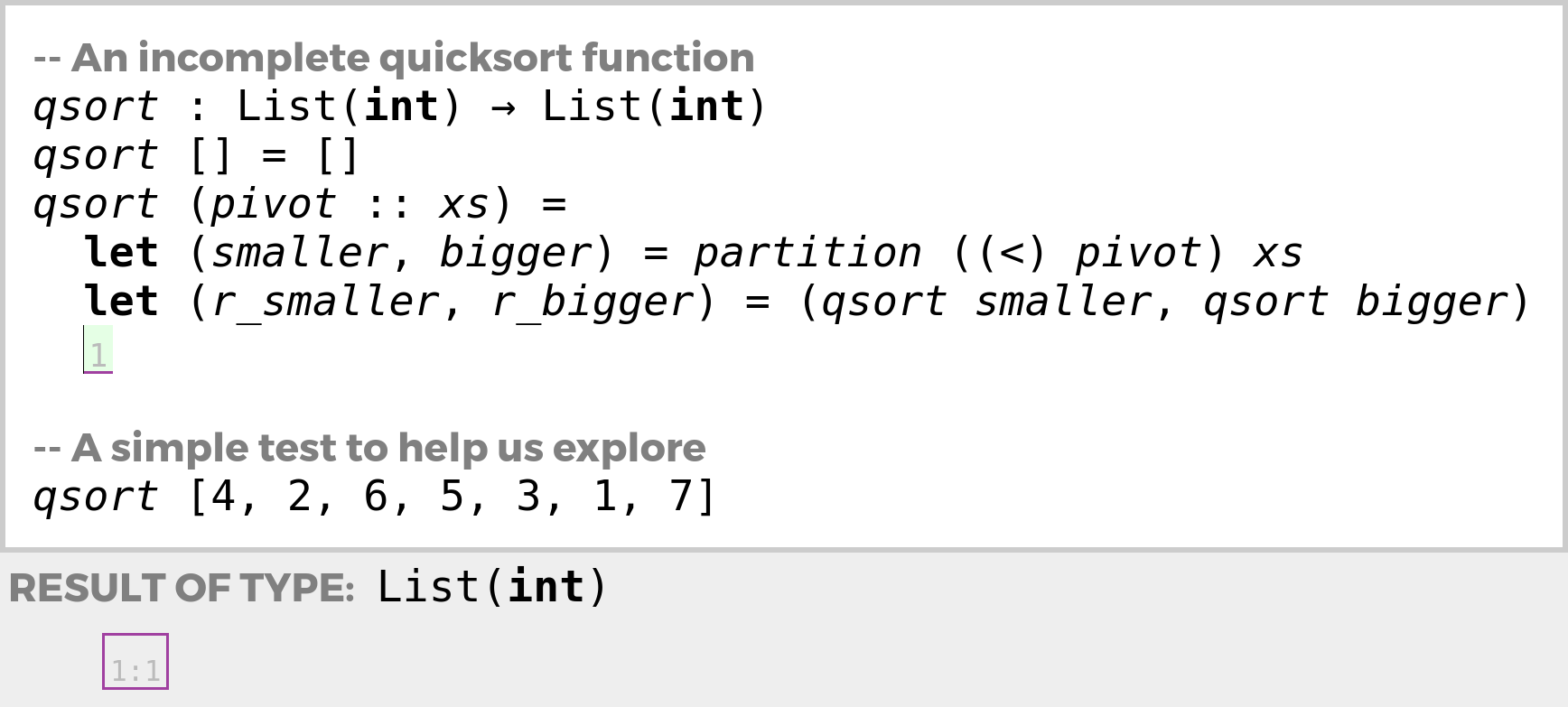}
\vspace{-3px}
\caption{The result of evaluation is a hole closure.}
\label{fig:qsort-example-code}
\end{subfigure}
~
\begin{subfigure}[t]{0.29\textwidth}
\centering
\includegraphics[width=\textwidth,interpolate=false,valign=t]{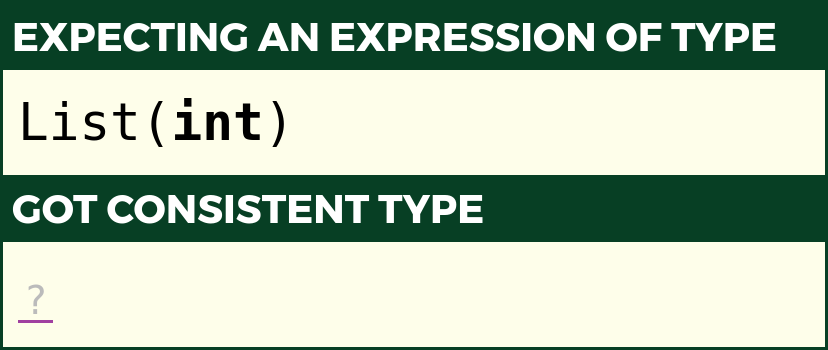}
\vspace{-3px}
\caption{The type inspector provides static feedback about the term at the cursor. Currently, the hole should be filled by a list expression. Holes have the hole (i.e. unknown) type, which  is universally consistent (see Sec.~\ref{sec:calculus} and \cite{popl-paper}).
}
\label{fig:qsort-type-inspector}
\end{subfigure}

\vspace{8px}

\begin{subfigure}[t]{\textwidth}
\centering
\includegraphics[width=0.29\textwidth,interpolate=false,valign=c]{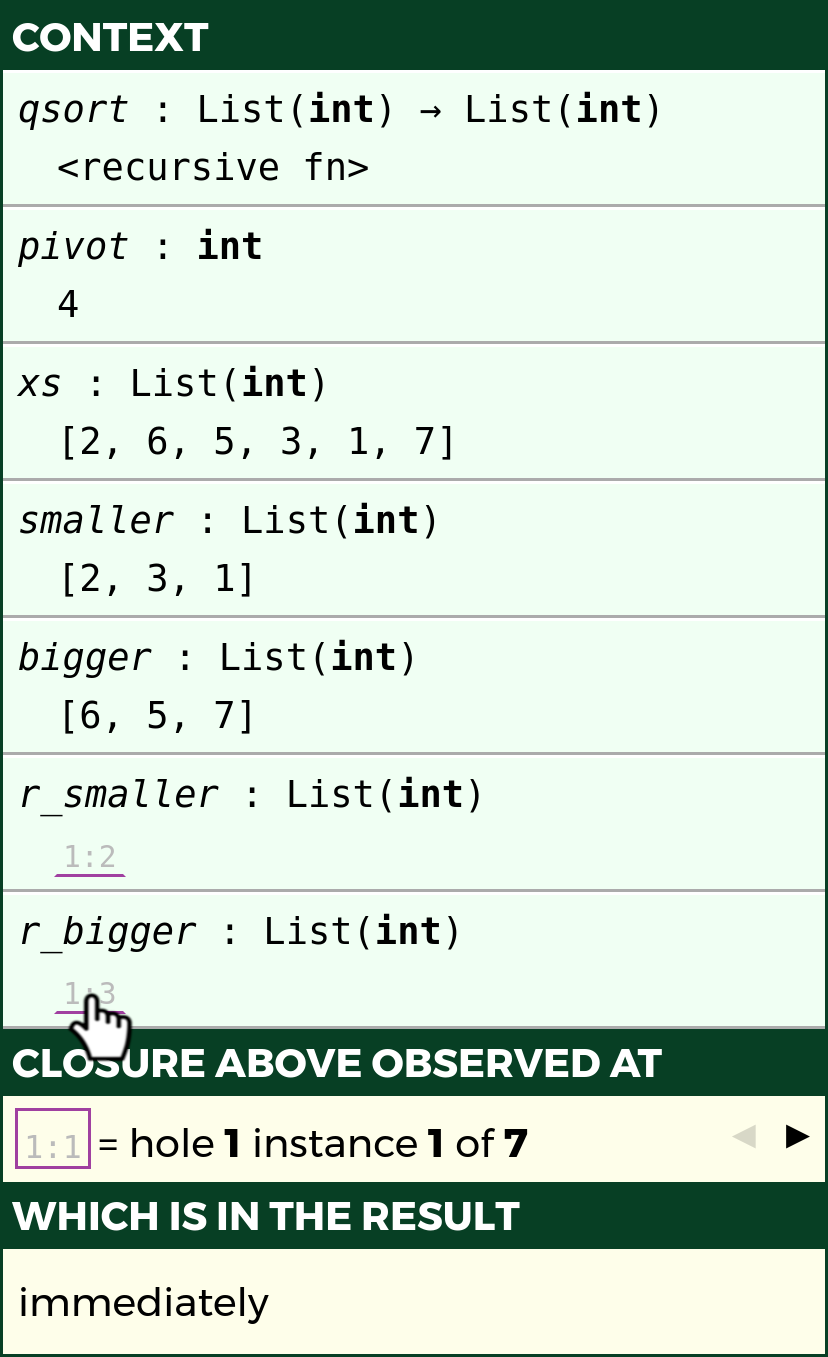}
~$\xrightarrow[\text{click}]{}$
\includegraphics[width=0.29\textwidth,interpolate=false,valign=c]{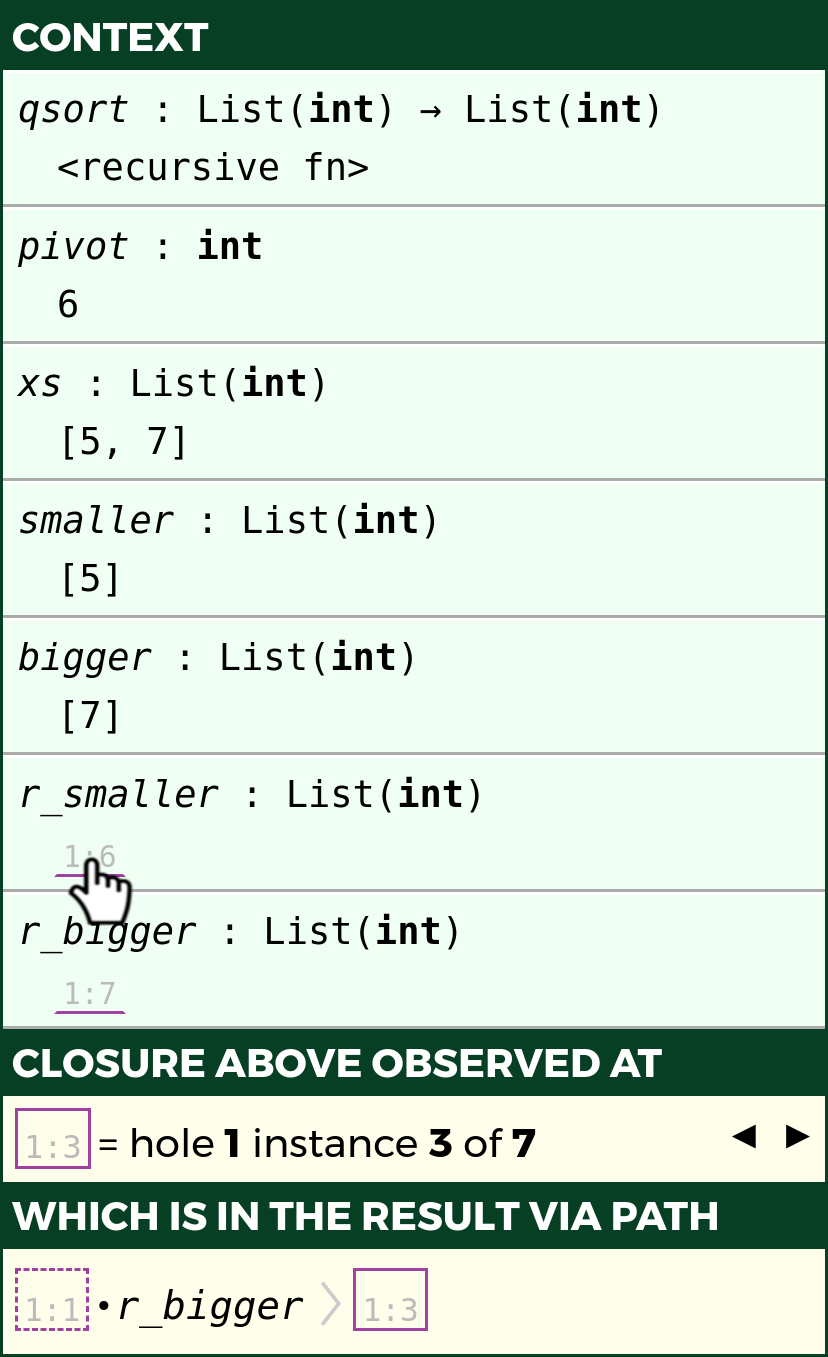}
~$\xrightarrow[\text{click}]{}$
\includegraphics[width=0.29\textwidth,interpolate=false,valign=c]{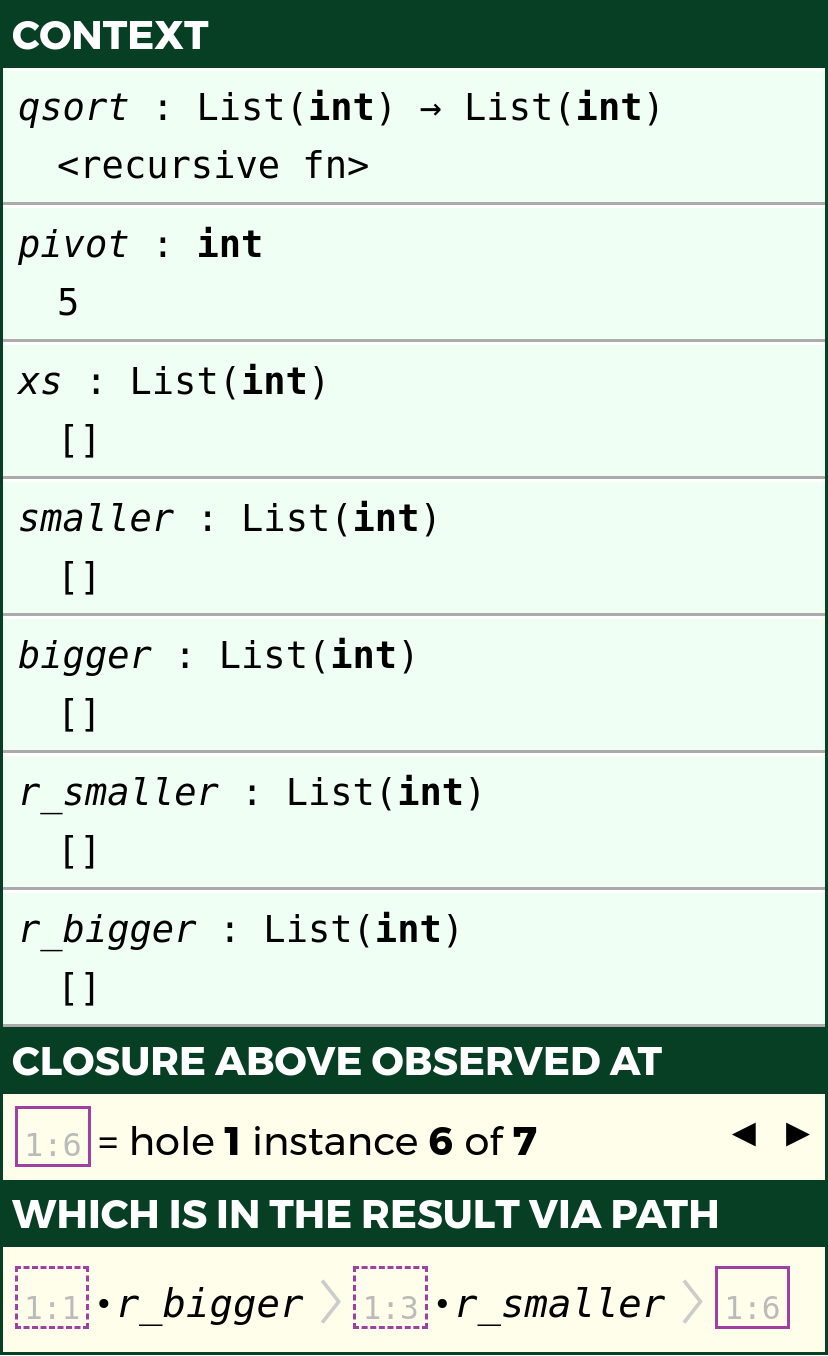}
\caption{The programmer can explore the recursive structure of the computation by clicking on hole instances.}
\label{fig:qsort-sidebars}
\end{subfigure}

\vspace{3px}

\caption{Example 2: Incomplete Quicksort}
\label{fig:qsort-cell-mockup}
% \end{subfigure}

\end{figure}

% \begin{subfigure}[t]{\textwidth}
% \centering
% \includegraphics[width=0.3\textwidth,interpolate=false]{images/grades-sidebar-1.png}
% ~${}^\blacktriangleright$
% \includegraphics[width=0.3\textwidth,interpolate=false]{images/grades-sidebar-2.png}
% ~${}^\blacktriangleright$
% \includegraphics[width=0.3\textwidth,interpolate=false]{images/grades-sidebar-3.png}
% \caption{Typing context view with live hole closure information}
% \label{sec:grades-sidebar}
% \end{subfigure}
% %% TODO once the code above is removed, scale up the screenshots
% \includegraphics[scale=0.20]{images/hazel-placeholder-0.png}

% \rkc{Draw arrows and captions on the top figure to show how to get
% to the bottom figure.
% ser navigates to hole a, types + to create a plus, types * to create a
% multiplication, types \#10 to create 10, types vh1 to create variable use.}

% \includegraphics[scale=0.20]{images/hazel-placeholder-1a.png}

Let us now consider a second more sophisticated example: an incomplete implementation of the recursive quicksort function, shown in Fig.~\ref{fig:qsort-example-code}. So far, the programmer (perhaps a student, or a lecturer using \Hazel as a presentational aid) has filled in the base case, and in the recursive case, partitioned the remainder of the list relative to the head, and made the two recursive calls. A hole appears in return position as the programmer contemplates how to fill the hole with an appropriate expression of list type, as indicated by the \emph{type inspector} in Fig.~\ref{fig:qsort-type-inspector}.

At the bottom of the cell in Fig.~\ref{fig:qsort-example-code}, the programmer has applied \li{qsort} to an example list. However, 
the indeterminate result of this function application is simply an instance of hole \li{1}, which serves only to confirm that evaluation went through the recursive case of \li{qsort}. 
More interesting is the live context inspector, shown in three states in Fig.~\ref{fig:qsort-sidebars}, which provides feedback about the values of the variables in scope at hole \li{1} from the the various instances of hole \li{1} that appear in the result, either immediately or within an outer closure. For example, in its initial state (Fig.~\ref{fig:qsort-sidebars}, left) it shows the closure at the instance of hole \li{1} that appears immediately in the result due to the outermost application of \li{qsort}. From this, the programmer can confirm (or the lecturer can visually point out) that 
the lists \li{smaller} and \li{bigger} computed by the call to \li{partition} are appropriately named, and observe that they are not yet themselves sorted.

The results from the subsequent recursive calls, \li{r_smaller} and \li{r_bigger}, are again hole instances, \li{1:2} and \li{1:3}. 
The programmer can click on either of these hole instances to reveal the associated closures from the corresponding recursive calls. 
For example, clicking on \li{1:3} reveals the hole closure from the \li{r_bigger} recursive call as shown in Fig.~\ref{fig:qsort-sidebars} (middle). 
From there, the programmer can click another hole closure, e.g. \li{1:6} to reveal the hole closure from the subsequent \li{r_smaller} recursive call as shown in Fig.~\ref{fig:qsort-sidebars} (right). 
Notice in each case that the path from the result to the selected hole closure is reported as shown at the bottom of the context inspector in Fig.~\ref{fig:qsort-sidebars}. 
In exploring these paths rooted at the result, the programmer can develop concrete intuitions about the recursive structure of the computation (e.g. by following through to the base case as shown in Fig.~\ref{fig:qsort-sidebars}) even before the program is complete. 

\subsection{Example 3: Live Programming with Static Type Errors}
\label{sec:static-errors}

% \begin{subfigure}[t]{\textwidth}
\begin{figure}
\centering
\includegraphics[width=\textwidth,interpolate=false,valign=t]{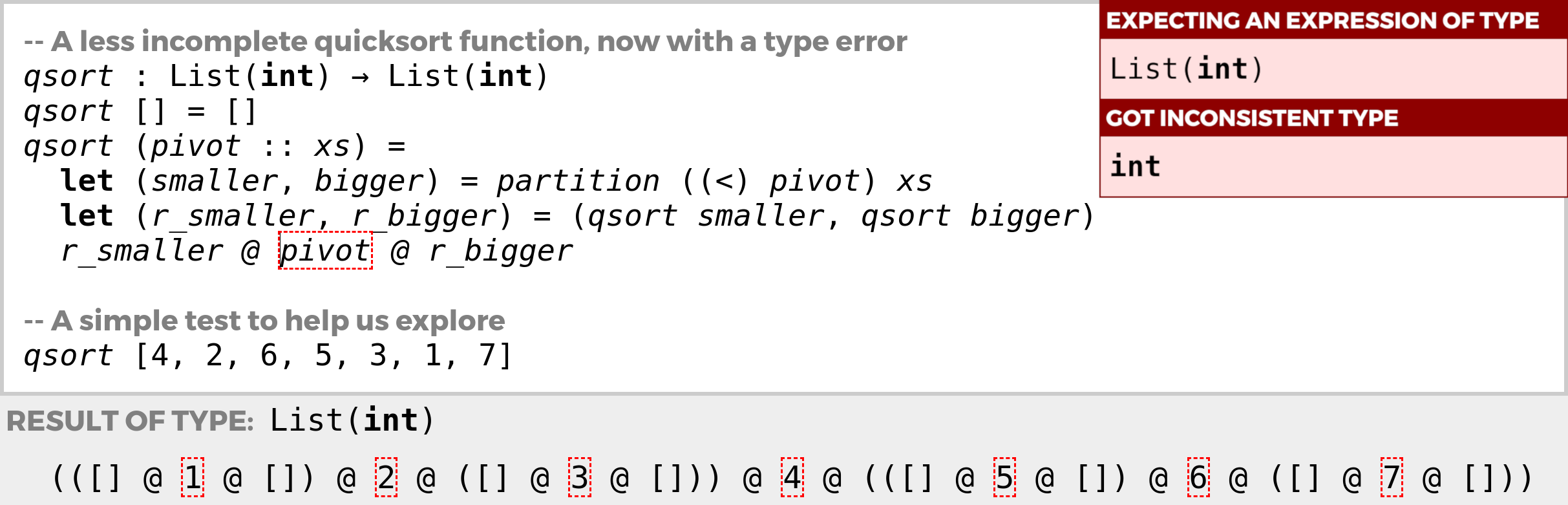}
\vspace{2px}
\caption{Example 3: Ill-Typed Quicksort}
\label{fig:qsort-type-error}
\vspace{-3px}
\end{figure}
% \end{subfigure}

The previous examples were incomplete 
because of \emph{missing} expressions.
Now, we discuss programs that are incomplete, 
and therefore conventionally meaningless, because of
\emph{type inconsistencies}. 
Let us return to the quicksort example just described, 
but assume that the programmer has filled in the previous hole
as shown in Fig.~\ref{fig:qsort-type-error}. In Sec.~\ref{sec:resumption}, we discuss
how the programming environment could avoid restarting evaluation after such edits by using the values stored in the hole closures.

The programmer appears to be on the right track conceptually
in recognizing that the pivot needs to appear between the 
smaller and bigger elements. 
However, the types do not quite work out: the \li{@} operator here
performs list concatenation, but the pivot is an integer. 
Most compilers and editors will report a static error message
to the programmer in this case, and \Hazel 
follows suit in the type inspector (shown inset in Fig.~\ref{fig:qsort-type-error}). 
However, our argument is that the presence of a static type error should not cause all feedback about 
the dynamic behavior of the program to ``flicker out'' or ``go stale'' --
after all, there are perfectly meaningful parts of the program (both nearby
and far away from the error) 
whose dynamic behavior may be of interest. Concrete results can also help the programmer understand the implications of the type error \cite{Seidel2016}.
% After all,
% the error is localized and there is perfectly good code elsewhere 
% in the program (if not nearby, then perhaps far away).

Our approach, following the prior work of \citet{popl-paper}, 
is to semantically internalize the ``red outline'' around
a type-inconsistent expression as a \emph{non-empty hole} around that expression.
Evaluation safely proceeds past a non-empty hole just as if it were an empty hole.
The semantics also associates an environment with each instance of a non-empty hole,
so we can use the live context inspector essentially as in Fig.~\ref{fig:qsort-sidebars} (not shown). 
Evaluation proceeds inside the hole, so that 
feedback about the type-inconsistent expression, which might ``almost'' be correct, is available. 
In this case, the result at the bottom of Fig.~\ref{fig:qsort-type-error}
reveals that the programmer is on the right track: the list elements 
appear in the correct order.
They simply have not been combined correctly.

We are treating \li{@} as a primitive operator in this example. If it were defined as a function, then
it would be possible to further reduce the example by proceeding into the function body. This is likely unhelpful for
functions other than those that programmer is actively working on. Our semantics \emph{allows} beta reduction when the argument is a hole, but it does not \emph{require} it.

Although \Hazel inserts non-empty holes automatically, the earlier work on \Hazelnut allowed the programmer to explicitly insert non-empty holes. It may be that these are useful even when there is not a type inconsistency because the non-empty hole defers elimination of the enclosed expression and causes hole closure information to be tracked. For example, if we repaired the program in Fig.~\ref{fig:qsort-type-error} by replacing the erroneous variable \li{pivot} with \li{[pivot]} and then inserted a non-empty hole around this expression, the result would show all of the list concatenation operations that would be performed without actually performing them, producing a result much like that in Fig.~\ref{fig:qsort-type-error}. This provides another way to explore the recursive structure of the \li{qsort} function.

\subsection{Example 4: Type Holes and Dynamic Type Errors}
\label{sec:dynamic-type-errors}

% So far, we have only discussed incomplete programs where a hole appears within an expression.
In \Hazel, the program can also be incomplete because holes appear in types. 
\citet{popl-paper} confirmed that the literature on \emph{gradual type systems} \cite{Siek06a,DBLP:conf/snapl/SiekVCB15} is directly relevant to the problem of reasoning with type holes, by identifying the type hole with the unknown type. 
% Unsurprisingly, then, it is also relevant to the problem of running 
% programs with type holes. 
Gradual type systems can run programs that are not yet sufficiently annotated with types by inserting \emph{casts}  where necessary. We take the same well-studied approach in \Hazel. As such, let us consider only a small synthetic example to demonstrate what is unique to our approach.

Fig.~\ref{fig:cast-errors} defines a simple function, \li{f}, of two arguments. 
The type annotation on the first line leaves the type of those arguments unknown. 
As such, the \Hazel type system, following the gradual typing approach,
allows the body of the function to use those two
arguments at any type (that is, the hole type is universally consistent). 
Here, the first argument, \li{simple}, is used at one type, \li{bool}, 
and the second argument, \li{x}, is used at two different types in the two branches (perhaps because the programmer made a mistake), 
 \li{int} and  \li{string}  ( \li{^} is string concatenation).
Although \Hazel supports only local type inference as of this writing, 
a system that uses ML-style type reconstruction to fill type holes statically, like GHC Haskell, would only be able to fill the first hole. Leaving the second hole unfilled is 
a parsimonious alternative to arbitrarily or heuristically choosing one of the possibilities and marking the
other uses of \li{x} as ill-typed (see \cite{DBLP:journals/jfp/ChenE18}).

% \begin{subfigure}[t]{\textwidth}
\begin{figure}
\centering
\includegraphics[width=0.75\textwidth,interpolate=false,valign=t]{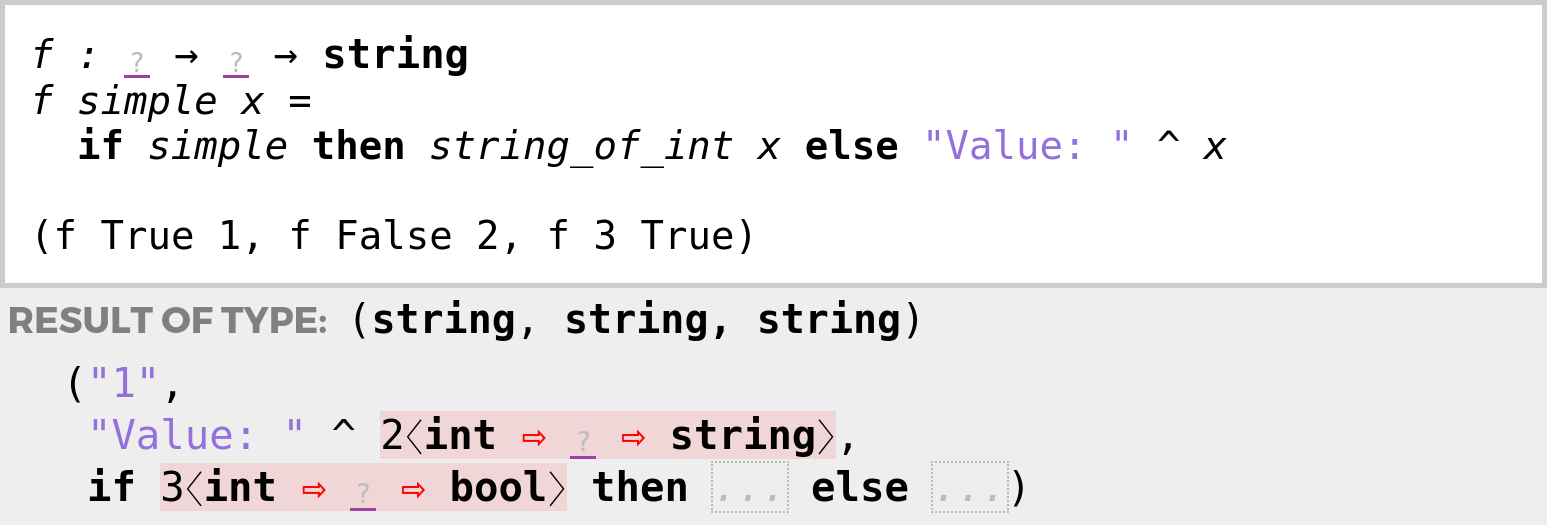}
\caption{Example 4: Type Holes and Dynamic Type Errors}
\label{fig:cast-errors}
\vspace{-6px}
\end{figure}
% \end{subfigure}

At the bottom of the cell in Fig.~\ref{fig:cast-errors}, we have three 
example applications of \li{f}, tupled together for concision. All three are statically
well-typed, again because the hole type is universally consistent. 
The result at the bottom of Fig.~\ref{fig:cast-errors} demonstrates that the first application
of \li{f} is dynamically unproblematic. This allows the programmer to confirm that 
the first branch operates as intended without the need to 
address the typing problems in the other branch \cite{Bayne:2011:ASD:1985793.1985864}. 
% This flexibility is a common motivation for dynamic and gradual languages .

The second application of \li{f}, in contrast, causes a dynamic type error because the second argument, \li{2}, is an \li{int} but evaluation takes the branch where it is used as a \li{string}. 
Rather than aborting evaluation when this occurs, as in existing gradual type systems, the problematic term becomes a \emph{failed cast} term, shown shaded in red, which can be read ``\li{2} is an \li{int} that was used through a variable of hole type (\li{?}) as a \li{string}''. 
A failed cast acts much like a non-empty hole as a membrane around a problematic term. The surrounding concatenation operation becomes indeterminate, but  
evaluation can continue on to the third application of \li{f}, which is also problematic, this time because the first argument is not a \li{bool} (perhaps because the programmer had an incorrect understanding of the argument order). 
Again, this causes a failed cast to appear, this time in guard position. Like a hole in guard position, evaluation cannot determine which branch to take so the whole conditional becomes indeterminate. 
The pretty printer
hides the two branches behind ellipses for concision.

\ifarxiv

\else
In this small example, it might have only been a small burden for the programmer to provide the intended types in the signature for \li{f}, but there are situations (e.g. during rapid prototyping or a live performance) where the programmer might consider the burden more substantial. This approach ensures that dynamic feedback does not exhibit gaps even when there is a dynamic type error.
\fi

\newcommand{\calculusSec}{Hazelnut Live}
\section{\calculusSec}
\label{sec:calculus}
% \def \TirNameStyle \Vtexttt{#1}{{#1}}

% !TEX root = hazelnut-dynamics.tex
\begin{figure}[t]
$\arraycolsep=4pt\begin{array}{rllllll}
\mathsf{HTyp} & \htau & ::= &
  b ~\vert~
  \tarr{\htau}{\htau} ~\vert~
  % \tprod{\htau}{\htau} ~\vert~
  % \tsum{\htau}{\htau} ~\vert~
  \tehole\\
\mathsf{HExp} & \hexp & ::= &
  c ~\vert~
  x ~\vert~
  \halam{x}{\htau}{\hexp} ~\vert~
      {\hlam{x}{\hexp}} ~\vert~
  \hap{\hexp}{\hexp} ~\vert~
  % \hpair{\hexp}{\hexp} ~\vert~
  % \hprj{i}{\hexp} ~\vert~
  % \hinj{i}{\hexp} ~\vert~
  % \hcase{\hexp}{x}{\hexp}{x}{\hexp} ~\vert~
  % \hadd{\hexp}{\hexp} ~\vert~
  \hehole{u} ~\vert~
  \hhole{\hexp}{u} ~\vert~
  \hexp : \htau\\
% \mathsf{Mark} & \markname{} & ::= &
%   \evaled{} ~\vert~  \unevaled{}\\
 \mathsf{IHExp} & \dexp  & ::= &
  c ~\vert~
  x ~\vert~
  {\halam{x}{\htau}{\dexp}} ~\vert~
  \hap{\dexp}{\dexp} ~\vert~
  % \hpair{\dexp}{\dexp} ~\vert~
  % \hprj{i}{\dexp} ~\vert~
  % \hinj{i}{\dexp} ~\vert~
  % \hcase{\dexp}{x}{\dexp}{x}{\dexp} ~\vert~
  % \hadd{\dexp}{\dexp} ~\vert~
  \dehole{\mvar}{\subst}{} ~\vert~
  \dhole{\dexp}{\mvar}{\subst}{} ~\vert~
  \dcasttwo{\dexp}{\htau}{\htau} ~\vert~
  \dcastfail{\dexp}{\htau}{\htau}\\
\end{array}$
$$
\dcastthree{\dexp}{\htau_1}{\htau_2}{\htau_3} \defeq
  \dcasttwo{\dcasttwo{\dexp}{\htau_1}{\htau_2}}{\htau_2}{\htau_3}
$$
\vspace{-12px}
\CaptionLabel{Syntax of types, $\htau$, external expressions, $\hexp$, and internal expressions, $\dexp$.
We write $x$ to range over variables,
$u$ over hole names, and
$\sigma$ over finite substitutions (i.e., environments) 
which map variables to internal expressions, written $d_1/x_1, ~\cdots, d_n/x_n$ for $n \geq 0$.}{fig:hazelnut-live-syntax}
\label{fig:HTyp}
\label{fig:HExp}
\end{figure}

We will now make the intuitions developed in the previous section formally
precise by specifying a core calculus, which we call \HazelnutLive, and
characterizing its metatheory.

\noindent
\parahead{Overview} The syntax of the core calculus given in
Fig.~\ref{fig:hazelnut-live-syntax} consists of types and expressions with
holes.  We distinguish between {external} expressions, $e$, and {internal}
expressions, $d$.  External expressions correspond to programs as entered
by the programmer (see Sec.~\ref{sec:intro} for discussion of implicit,
manual, semi-automated and fully automated hole entry methods).  Each
well-typed external expression (see Sec.~\ref{sec:external-statics} below)
elaborates to a well-typed internal expression (see
Sec.~\ref{sec:elaboration}) before it is evaluated (see
Sec.~\ref{sec:evaluation}).  We take this approach, notably also taken in
the ``redefinition'' of Standard ML by \citet{Harper00atype-theoretic},
because (1) the external language supports type inference and explicit type
ascriptions, $\hexp : \htau$, but it is formally simpler to eliminate
ascriptions and specify a type assignment system when defining the dynamic
semantics; and (2) we need additional syntactic machinery during evaluation
for tracking hole closures and dynamic type casts.  This machinery is
inserted by elaboration, rather than entered explicitly by the programmer.
In this regard, the internal language is analogous to the cast calculus in
the gradually typed lambda calculus
\cite{DBLP:conf/snapl/SiekVCB15,Siek06a}, though as we will see the
\HazelnutLive internal language goes beyond the cast calculus in several
respects. We have mechanized these formal developments using the Agda proof
assistant \cite{norell:thesis,norell2009dependently} (see
Sec.~\ref{sec:agda-mechanization}). Rule names in this section,
e.g. \rulename{SVar}, correspond to variables from the
mechanization. The \Hazel implementation substantially follows the formal
specification of \Hazelnut (for the editor) and \HazelnutLive
(for the evaluator). We can formally state a continuity invariant for a
putative combined calculus (see Sec.~\ref{sec:implementation}).

% \rkc{this syntactic sugar is used in four places: ITCastSucceed, ITCastFail,
% ITGround, and ITExpand. that's not many, and those rules don't look much more
% cluttered without the sugar, so consider eliminating it. if so, just toggle the
% definition of the dcastthree macro to the unsugared option.}

\subsection{Static Semantics of the External Language}
\label{sec:external-statics}

% !TEX root = hazelnut-dynamics.tex

\begin{figure}[t]
\judgbox{\hsyn{\hGamma}{\hexp}{\htau}}{$\hexp$ synthesizes type $\htau$}
\vspace{-10px}
\begin{mathpar}
\inferrule[SConst]{ }{
  \hsyn{\hGamma}{c}{b}
}

\inferrule[SVar]{
  x : \htau \in \hGamma
}{
  \hsyn{\hGamma}{x}{\htau}
}

\inferrule[SLam]{
  \hsyn{\hGamma, x : \htau_1}{\hexp}{\htau_2}
}{
  \hsyn{\hGamma}{\halam{x}{\htau_1}{\hexp}}{\tarr{\htau_1}{\htau_2}}
}

\inferrule[SAp]{
    \hsyn{\hGamma}{\hexp_1}{\htau_1}    \\
    \arrmatch{\htau_1}{\tarr{\htau_2}{\htau}}\\\\
        \hana{\hGamma}{\hexp_2}{\htau_2}
}{
  \hsyn{\hGamma}{\hap{\hexp_1}{\hexp_2}}{\htau}
}

\inferrule[SEHole]{ }{
  \hsyn{\hGamma}{\hehole{u}}{\tehole}
}

\inferrule[SNEHole]{
  \hsyn{\hGamma}{\hexp}{\htau}
}{
  \hsyn{\hGamma}{\hhole{\hexp}{u}}{\tehole}
}

\inferrule[SAsc]{
  \hana{\hGamma}{\hexp}{\htau}
}{
  \hsyn{\hGamma}{\hexp : \htau}{\htau}
}
\end{mathpar}

\vsepRule

\judgbox{\hana{\hGamma}{\hexp}{\htau}}{$\hexp$ analyzes against type $\htau$}
\vspace{-4px}
\begin{mathpar}
\inferrule[ALam]{
  \arrmatch{\htau}{\tarr{\htau_1}{\htau_2}}\\
  \hana{\hGamma, x : \htau_1}{\hexp}{\htau_2}
}{
  \hana{\hGamma}{\hlam{x}{\hexp}}{\htau}
}

\inferrule[ASubsume]{
  \hsyn{\hGamma}{\hexp}{\htau}\\
  \tconsistent{\htau}{\htau'}
}{
  \hana{\hGamma}{\hexp}{\htau'}
}
\end{mathpar}
\vspace{-2px}
\CaptionLabel{Bidirectional Typing of External Expressions}{fig:bidirectional-typing}
\vspace{-2px}
\end{figure}

We start with the type system of the \HazelnutLive external language,
which closely follows the \Hazelnut type system \cite{popl-paper}; we summarize the minor differences as they come up.

% except that (1) holes in \HazelnutLive each have a \emph{unique name};
% (2) for brevity of exposition, we removed numbers in favor of a simpler base type, $b$, with one value, $c$ (i.e. $b$ is the unit type); and
% (3) we include both unannotated lambdas, $\hlam{x}{\hexp}$, and half-annotated lambdas, $\halam{x}{\htau}{\hexp}$.  which we discuss
% (along with other systematic extensions) in
% Appendix~\ref{sec:extensions}.

\Figref{fig:bidirectional-typing} defines the type system in the \emph{bidirectional} style
with two mutually defined judgements \cite{Pierce:2000ve,bidi-tutorial,DBLP:conf/icfp/DunfieldK13,Chlipala:2005da}. The type synthesis
judgement~$\hsyn{\hGamma}{\hexp}{\htau}$ synthesizes a type~$\htau$
for external expression~$\hexp$ under typing context $\hGamma$, which tracks typing
assumptions of the form $x : \htau$ in the usual
manner \cite{pfpl,tapl}.
The type analysis judgement~$\hana{\hGamma}{\hexp}{\htau}$ checks
expression~$\hexp$ against a given type~$\htau$.
Algorithmically, analysis accepts a type as input, and synthesis gives
a type as output.
We start with synthesis for the programmer's top level external
expression.

% Algorithmically, the type is an output of type synthesis but an input of type analysis.

The primary benefit of specifying the \HazelnutLive external language
bidirectionally is that the programmer need not annotate each hole with a type.
An empty hole is
written simply $\hehole{u}$, where $u$ is the hole name, which we tacitly assume is unique
(holes in \Hazelnut were not named).
Rule \rulename{SEHole} specifies that an empty hole synthesizes hole type, written $\tehole$.
If an empty hole appears where an expression of some other type is
expected, e.g. under an explicit ascription (governed by Rule \rulename{SAsc})
or in the argument position of a function application (governed by
Rule \rulename{SAp}, discussed below), we apply the \emph{subsumption rule},
Rule \rulename{ASubsume}, which specifies that if an expression $e$ synthesizes
type $\htau$, then it may be checked against any \emph{consistent}
type, $\htau'$.

% !TEX root = hazelnut-dynamics.tex

\begin{figure}[t]
\judgbox{\tconsistent{\htau_1}{\htau_2}}{$\htau_1$ is consistent with $\htau_2$}
\vspace{-5px}
\begin{mathpar}
\inferrule[TCHole1]{ }{
  \tconsistent{\tehole}{\htau}
}

\inferrule[TCHole2]{ }{
  \tconsistent{\htau}{\tehole}
}

\inferrule[TCRefl]{ }{
  \tconsistent{\htau}{\htau}
}

\inferrule[TCArr]{
  \tconsistent{\htau_1}{\htau_1'}\\
  \tconsistent{\htau_2}{\htau_2'}
}{
  \tconsistent{\tarr{\htau_1}{\htau_2}}{\tarr{\htau_1'}{\htau_2'}}
}
%
% \inferrule{
%   \tconsistent{\htau_1}{\htau_1'}\\
%   \tconsistent{\htau_2}{\htau_2'}
% }{
%   \tconsistent{\tprod{\htau_1}{\htau_2}}{\tprod{\htau_1'}{\htau_2'}}
% }
%
% \inferrule{
%   \tconsistent{\htau_1}{\htau_1'}\\
%   \tconsistent{\htau_2}{\htau_2'}
% }{
%   \tconsistent{\tsum{\htau_1}{\htau_2}}{\tsum{\htau_1'}{\htau_2'}}
% }
\end{mathpar}

% \vsepRule

% \judgbox{\tinconsistent{\htau_1}{\htau_2}}{$\htau_1$ is inconsistent with $\htau_2$}
% \begin{mathpar}
%     \inferrule[ICBaseArr1]{ }{
%       \tinconsistent{\tb}{\tarr{\htau_1}{\htau_2}}
%     }

%     \inferrule[ICBaseArr2]{ }{
%       \tinconsistent{\tarr{\htau_1}{\htau_2}}{\tb}
%     }

%     \inferrule[ICArr1]{
%       \tinconsistent{\htau_1}{\htau_3}
%     }{
%       \tinconsistent{\tarr{\htau_1}{\htau_2}}{\tarr{\htau_3}{\htau_4}}
%     }

%     \inferrule[ICArr2]{
%       \tinconsistent{\htau_2}{\htau_4}
%     }{
%       \tinconsistent{\tarr{\htau_1}{\htau_2}}{\tarr{\htau_3}{\htau_4}}
%     }
% \end{mathpar}

\vsepRule

\judgbox{\arrmatch{\htau}{\tarr{\htau_1}{\htau_2}}}{$\htau$ has matched arrow type $\tarr{\htau_1}{\htau_2}$}
\vspace{-4px}
\begin{mathpar}
\inferrule[MAHole]{ }{
  \arrmatch{\tehole}{\tarr{\tehole}{\tehole}}
}

\inferrule[MAArr]{ }{
  \arrmatch{\tarr{\htau_1}{\htau_2}}{\tarr{\htau_1}{\htau_2}}
}
\end{mathpar}

% \judgbox{\prodmatch{\htau}{\tprod{\htau_1}{\htau_2}}}{$\htau$ has matched product type $\tprod{\htau_1}{\htau_2}$}
% \begin{mathpar}
% \inferrule{ }{
%   \prodmatch{\tehole}{\tprod{\tehole}{\tehole}}
% }

% \inferrule{ }{
%   \prodmatch{\tprod{\htau_1}{\htau_2}}{\tprod{\htau_1}{\htau_2}}
% }
% \end{mathpar}

% \judgbox{\summatch{\htau}{\tsum{\htau_1}{\htau_2}}}{$\htau$ has matched sum type $\tsum{\htau_1}{\htau_2}$}
% \begin{mathpar}
% \inferrule{ }{
%   \summatch{\tehole}{\tsum{\tehole}{\tehole}}
% }

% \inferrule{ }{
%   \summatch{\tsum{\htau_1}{\htau_2}}{\tsum{\htau_1}{\htau_2}}
% }
% \end{mathpar}
\CaptionLabel{Type Consistency and Matching}{fig:tconsistent}
\label{fig:arrmatch}
\vspace{-4px}
\end{figure}

Fig.~\ref{fig:tconsistent} specifies the type consistency relation, written $\tconsistent{\htau}{\htau'}$, which specifies that two types are consistent if they differ only up to type holes in corresponding positions.
The hole type is consistent with every type, and so, by the subsumption rule, expression holes may appear where an expression of any type is expected. The type consistency relation here coincides with the type consistency relation from gradual type theory by identifying the hole type with the unknown type~\cite{Siek06a}.
Type consistency is reflexive and symmetric, but it is \emph{not} transitive.
This stands in contrast to subtyping, which is anti-symmetric and transitive; subtyping may be integrated into a gradual type system following \citet{Siek:2007qy}.

Non-empty expression holes, written $\hhole{\hexp}{u}$, behave similarly to empty holes.
Rule \rulename{SNEHole} specifies that a non-empty expression hole also synthesizes hole type as long as the expression inside the hole, $\hexp$, synthesizes some (arbitrary) type.
Non-empty expression holes therefore internalize the ``red underline/outline'' that many editors display around type inconsistencies in a program.

For the familiar forms of the lambda calculus, the rules again follow prior work.
For simplicity, the core calculus includes only a single base type~$b$ with a single constant~$c$, governed by Rule \rulename{SConst} (i.e. $b$ is the unit type).
%
%\matt{Extraneous and uninteresting:}
By contrast, \citet{popl-paper} instead defined a number type with a single operation. That paper also defined sum types as an extension to the core calculus. We follow suit on both counts in \ifarxiv Appendix \ref{sec:extensions}\else the \appendixName\fi.%Appendix~\ref{sec:extensions}.

Rule \rulename{SVar} synthesizes the corresponding type from $\hGamma$.
For the sake of exposition, \HazelnutLive includes ``half-annotated'' lambdas, $\halam{x}{\htau}{\hexp}$, in addition to the unannotated lambdas, $\hlam{x}{\hexp}$, from \Hazelnut.
Half-annotated lambdas may appear in synthetic position according to Rule \rulename{SLam}, which is standard \cite{Chlipala:2005da}.
Unannotated lambdas may only appear where the expected type is known to be either an arrow type or the hole type, which is treated as if it were $\tarr{\tehole}{\tehole}$.\footnote{A system supporting ML-style type reconstruction \cite{damas1982principal} might also include a synthetic rule for unannotated lambdas, e.g. as outlined by \citet{DBLP:conf/icfp/DunfieldK13}, but we stick to this simpler ``Scala-style'' local type inference scheme in this paper \cite{Pierce:2000ve,Odersky:2001lb}.} 
To avoid the need for separate rules for these two cases, Rule \rulename{ALam} uses the matching relation $\arrmatch{\htau}{\tarr{\htau_1}{\htau_2}}$ defined in \Figref{fig:arrmatch}, which produces the matched arrow type $\tarr{\tehole}{\tehole}$ given the hole type, and operates as the identity on arrow types \cite{DBLP:conf/snapl/SiekVCB15,DBLP:conf/popl/GarciaC15}. The rule governing function application, Rule \rulename{SAp}, similarly treats an expression of hole type in function position as if it were of type $\tarr{\tehole}{\tehole}$ using the same matching relation.
%
% \Secref{sec:related} dicusses how \HazelnutLive might be enriched with
% with ML-style type reconstruction~\cite{damas1982principal}, perhaps via
% the approach outlined by~\citet{DBLP:conf/icfp/DunfieldK13}.
%
%%%%%%%%%%%%%%%%%%%%%%%%%%%%%%%%%%%%%%%%%%%%%%%%%%%%%%%%%%%%%%%%%%%%%%%%%%%%%%%%%%%%%%%%%%%%%%%%%%%%%%%%%%%%%%%%%%%%%
% To Cyrus from Matt:
%
% Why say the following here? --- It's discussing a different design that we didn't pursue here.  We should move such discussion to related work.
%
%Note that
%%

We do not formally need an explicit fixpoint operator because this calculus supports general recursion due to type holes, e.g. we can express the Y combinator as $(\halam{x}{\tehole}{x(x)}) (\halam{x}{\tehole}{x(x)})$. More generally, the untyped lambda calculus can be embedded as described by \citet{Siek06a}.
%As such, we omit an explicit fixpoint operator for concision.\todo{fixpoint?}{}

\vspace{-4px}
\subsection{Elaboration}
\label{sec:elaboration}
\vspace{-1px}

% !TEX root = hazelnut-dynamics.tex

\begin{figure}[p]
\judgbox
  {\elabSyn{\hGamma}{\hexp}{\htau}{\dexp}{\Delta}}
  {$\hexp$ synthesizes type $\htau$ and elaborates to $\dexp$}
\begin{mathpar}
\inferrule[ESConst]{ }{
  \elabSyn{\hGamma}{c}{b}{c}{\emptyset}
}
\and
\inferrule[ESVar]{
  x : \htau \in \hGamma
}{
  \elabSyn{\hGamma}{x}{\htau}{x}{\emptyset}
}
\and
\inferrule[ESLam]{
  \elabSyn{\hGamma, x : \htau_1}{\hexp}{\htau_2}{\dexp}{\Delta}
}{
  \elabSyn{\hGamma}{\halam{x}{\htau_1}{\hexp}}{\tarr{\htau_1}{\htau_2}}{\halam{x}{\htau_1}{\dexp}}{\Delta}
}
\and
\inferrule[ESAp]{
  \hsyn{\hGamma}{\hexp_1}{\htau_1}\\
  \arrmatch{\htau_1}{\tarr{\htau_2}{\htau}}
  \\\\
  \elabAna{\hGamma}{\hexp_1}{\tarr{\htau_2}{\htau}}{\dexp_1}{\htau_1'}{\Delta_1}\\
  \elabAna{\hGamma}{\hexp_2}{\htau_2}{\dexp_2}{\htau_2'}{\Delta_2}
}{
  \elabSyn
    {\hGamma}
    {\hap{\hexp_1}{\hexp_2}}
    {\htau}
    {\hap{(\dcasttwo{\dexp_1}{\htau_1'}{\tarr{\htau_2}{\htau}})}
         {\dcasttwo{\dexp_2}{\htau_2'}{\htau_2}}}
    {\Dunion{\Delta_1}{\Delta_2}}
}
%
%% \inferrule[ESAp1]{
%%   \hsyn{\hGamma}{\hexp_1}{\tehole}\\
%%   \elabAna{\hGamma}{\hexp_1}{\tarr{\htau_2}{\tehole}}{\dexp_1}{\htau_1}{\Delta_1}\\
%%   \elabAna{\hGamma}{\hexp_2}{\tehole}{\dexp_2}{\htau_2}{\Delta_2}
%% }{
%%   \elabSyn{\hGamma}{\hap{\hexp_1}{\hexp_2}}{\tehole}{\hap{(\dcast{\tarr{\htau_2}{\tehole}}{\dexp_1})}{\dexp_2}}{\Dunion{\Delta_1}{\Delta_2}}
%% }
%%
%% \inferrule[ESAp2]{
%%   \elabSyn{\hGamma}{\hexp_1}{\tarr{\htau_2}{\htau}}{\dexp_1}{\Delta_1}\\
%%   \elabAna{\hGamma}{\hexp_2}{\htau_2}{\dexp_2}{\htau'_2}{\Delta_2}\\
%%   \htau_2 \neq \htau'_2
%% }{
%%   \elabSyn{\hGamma}{\hap{\hexp_1}{\hexp_2}}{\htau}{\hap{\dexp_1}{\dcast{\htau_2}{\dexp_2}}}{\Dunion{\Delta_1}{\Delta_2}}
%% }
%%
%% \inferrule[ESAp3]{
%%   \elabSyn{\hGamma}{\hexp_1}{\tarr{\htau_2}{\htau}}{\dexp_1}{\Delta_1}\\
%%   \elabAna{\hGamma}{\hexp_2}{\htau_2}{\dexp_2}{\htau_2}{\Delta_2}
%% }{
%%   \elabSyn{\hGamma}{\hap{\hexp_1}{\hexp_2}}{\htau}{\hap{\dexp_1}{\dexp_2}}{\Dunion{\Delta_1}{\Delta_2}}
%% }\\
%
%
% \inferrule[expand-pair]{
%   \elabSyn{\hGamma}{\hexp_1}{\htau_1}{\dexp_1}{\Delta_1}\\
%   \elabSyn{\hGamma}{\hexp_2}{\htau_2}{\dexp_2}{\Delta_2}
% }{
%   \elabSyn{\hGamma}{\hpair{\hexp_1}{\hexp_2}}{\tprod{\htau_1}{\htau_2}}{\hpair{\dexp_1}{\dexp_2}}{\Dunion{\Delta_1}{\Delta_2}}
% }
%
% \inferrule[expand-prj]{
%   a
% }{
%   b
% }
%
% (inj)
%
%
% \inferrule[expand-plus]{ }{
%   \elabSyn{\hGamma}{\hadd{\hexp_1}{\hexp_2}}{\tnum}{\hadd{\dexp_1}{\dexp_2}}{\Dunion{\Delta_1}{\Delta_2}}
% }
\and
\inferrule[ESEHole]{ }{
  \elabSyn{\hGamma}{\hehole{u}}{\tehole}{\dehole{u}{\idof{\hGamma}}{}}{\Dbinding{u}{\hGamma}{\tehole}}
}
\and
\inferrule[ESNEHole]{
  \elabSyn{\hGamma}{\hexp}{\htau}{\dexp}{\Delta}
}{
  \elabSyn{\hGamma}{\hhole{\hexp}{u}}{\tehole}{\dhole{\dexp}{u}{\idof{\hGamma}}{}}{\Delta, \Dbinding{u}{\hGamma}{\tehole}}
}
\and
\inferrule[ESAsc]{
  \elabAna{\hGamma}{\hexp}{\htau}{\dexp}{\htau'}{\Delta}
}{
  \elabSyn{\hGamma}{\hexp : \htau}{\htau}{\dcasttwo{\dexp}{\htau'}{\htau}}{\Delta}
}
%% \inferrule[ESAsc1]{
%%   \elabAna{\hGamma}{\hexp}{\htau}{\dexp}{\htau'}{\Delta}\\
%%   \htau \neq \htau'
%% }{
%%   \elabSyn{\hGamma}{\hexp : \htau}{\htau}{\dcast{\htau}{\dexp}}{\Delta}
%% }
%%
%% \inferrule[ESAsc2]{
%%   \elabAna{\hGamma}{\hexp}{\htau}{\dexp}{\htau}{\Delta}
%% }{
%%   \elabSyn{\hGamma}{\hexp : \htau}{\htau}{\dexp}{\Delta}
%% }
\end{mathpar}

\vsepRule

\judgbox
  {\elabAna{\hGamma}{\hexp}{\htau_1}{\dexp}{\htau_2}{\Delta}}
  {$\hexp$ analyzes against type $\htau_1$ and
   elaborates to $\dexp$ of consistent type $\htau_2$}
\begin{mathpar}
\inferrule[EALam]{
  \arrmatch{\htau}{\tarr{\htau_1}{\htau_2}}\\
  \elabAna{\hGamma, x : \htau_1}{\hexp}{\htau_2}{\dexp}{\htau'_2}{\Delta}
}{
  \elabAna{\hGamma}{\hlam{x}{\hexp}}{\htau}{\halam{x}{\htau_1}{\dexp}}{\tarr{\htau_1}{\htau_2'}}{\Delta}
}

%% \inferrule[EALam]{
%%   \elabAna{\hGamma, x : \htau_1}{\hexp}{\htau_2}{\dexp}{\htau'_2}{\Delta}
%% }{
%%   \elabAna{\hGamma}{\hlam{x}{\hexp}}{\tarr{\htau_1}{\htau_2}}{\halam{x}{\htau_1}{\dexp}}{\tarr{\htau_1}{\htau_2'}}{\Delta}
%% }
%%
%% \inferrule[EALamHole]{
%%   \elabAna{\hGamma, x : \tehole}{\hexp}{\tehole}{\dexp}{\htau}{\Delta}
%% }{
%%   \elabAna{\hGamma}{\hlam{x}{\hexp}}{\tehole}{\halam{x}{\tehole}{\dexp}}{\tarr{\tehole}{\htau}}{\Delta}
%% }
%%
\inferrule[EASubsume]{
  \hexp \neq \hehole{u}\\
  \hexp \neq \hhole{\hexp'}{u}\\\\
  \elabSyn{\hGamma}{\hexp}{\htau'}{\dexp}{\Delta}\\
  \tconsistent{\htau}{\htau'}
}{
  \elabAna{\hGamma}{\hexp}{\htau}{\dexp}{\htau'}{\Delta}
}

\inferrule[EAEHole]{ }{
  \elabAna{\hGamma}{\hehole{u}}{\htau}{\dehole{u}{\idof{\hGamma}}{}}{\htau}{\Dbinding{u}{\hGamma}{\htau}}
}

\inferrule[EANEHole]{
  \elabSyn{\hGamma}{\hexp}{\htau'}{\dexp}{\Delta}\\
}{
  \elabAna{\hGamma}{\hhole{\hexp}{u}}{\htau}{\dhole{\dexp}{u}{\idof{\hGamma}}{}}{\htau}{\Delta, \Dbinding{u}{\hGamma}{\htau}}
}
\end{mathpar}
\CaptionLabel{Elaboration}{fig:elaboration}
\label{fig:expandSyn}
\label{fig:expandAna}
\end{figure}

% !TEX root = hazelnut-dynamics.tex

\begin{figure}[p]
\judgbox{\hasType{\Delta}{\hGamma}{\dexp}{\htau}}{$\dexp$ is assigned type $\htau$}
\begin{mathpar}
\inferrule[TAConst]{ }{
  \hasType{\Delta}{\hGamma}{c}{b}
}

\inferrule[TAVar]{
  x : \htau \in \hGamma
}{
	\hasType{\Delta}{\hGamma}{x}{\htau}
}

\inferrule[TALam]{
  \hasType{\Delta}{\hGamma, x : \htau_1}{\dexp}{\htau_2}
}{
  \hasType{\Delta}{\hGamma}{\halam{x}{\htau_1}{\dexp}}{\tarr{\htau_1}{\htau_2}}
}

\inferrule[TAAp]{
  \hasType{\Delta}{\hGamma}{\dexp_1}{\tarr{\htau_2}{\htau}}\\
  \hasType{\Delta}{\hGamma}{\dexp_2}{\htau_2}
}{
  \hasType{\Delta}{\hGamma}{\hap{\dexp_1}{\dexp_2}}{\htau}
}

\inferrule[TAEHole]{
  \Dbinding{u}{\hGamma'}{\htau} \in \Delta\\
  \hasType{\Delta}{\hGamma}{\sigma}{\hGamma'}
}{
  \hasType{\Delta}{\hGamma}{\dehole{u}{\sigma}{}}{\htau}
}

\inferrule[TANEHole]{
  \hasType{\Delta}{\hGamma}{\dexp}{\htau'}\\\\
  \Dbinding{u}{\hGamma'}{\htau} \in \Delta\\
  \hasType{\Delta}{\hGamma}{\sigma}{\hGamma'}
}{
  \hasType{\Delta}{\hGamma}{\dhole{\dexp}{u}{\sigma}{}}{\htau}
}

\inferrule[TACast]{
  \hasType{\Delta}{\Gamma}{\dexp}{\htau_1}\\
  \tconsistent{\htau_1}{\htau_2}
}{
  \hasType{\Delta}{\hGamma}{\dcasttwo{\dexp}{\htau_1}{\htau_2}}{\htau_2}
}

\inferrule[TAFailedCast]{
  \hasType{\Delta}{\Gamma}{\dexp}{\htau_1}\\
  \isGround{\htau_1}\\
  \isGround{\htau_2}\\
  \htau_1\neq\htau_2
}{
  \hasType{\Delta}{\hGamma}{\dcastfail{\dexp}{\htau_1}{\htau_2}}{\htau_2}
}
\end{mathpar}
\CaptionLabel{Type Assignment for Internal Expressions}{fig:hasType}
\end{figure}

Each well-typed external expression~$e$ elaborates to a well-typed internal
expression~$d$, for evaluation.
\Figref{fig:elaboration} specifies elaboration, and \Figref{fig:hasType}
specifies type assignment for internal expressions.
%
% \Secref{sec:evaluation} discusses internal expression evaluation.

As with the type system for the external language (above), we specify
elaboration bidirectionally \cite{DBLP:conf/ppdp/FerreiraP14}.
The synthetic elaboration
judgement~$\elabSyn{\hGamma}{\hexp}{\htau}{\dexp}{\Delta}$
produces an elaboration~$d$ and a hole context~$\hDelta$ when synthesizing
type $\htau$ for $\hexp$.
%
%We say more about hole contexts, which are used in the type assignment judgement, $\hasType{\Delta}{\hGamma}{d}{\htau}$, below.
We describe hole contexts, which serve as ``inputs'' to the type assignment
judgement~$\hasType{\Delta}{\hGamma}{d}{\htau}$, further below.
The analytic elaboration
judgement~$\elabAna{\hGamma}{\hexp}{\htau}{\dexp}{\htau'}{\Delta}$,
produces an elaboration~$d$ of type~$\htau'$, and a hole context~$\hDelta$,
when checking~$\hexp$ against~$\htau$.
The following theorem establishes that elaborations are well-typed and in
the analytic case that the assigned type, $\htau'$, is consistent with
provided type, $\htau$.
\begin{thm}[Typed Elaboration]\label{thm:typed-elaboration} ~
  \begin{enumerate}[nolistsep]
    \item
      If $\elabSyn{\hGamma}{\hexp}{\htau}{\dexp}{\Delta}$
      then $\hasType{\Delta}{\hGamma}{\dexp}{\htau}$.
    \item
      If $\elabAna{\hGamma}{\hexp}{\htau}{\dexp}{\htau'}{\Delta}$ then
      $\tconsistent{\htau}{\htau'}$ and
      $\hasType{\Delta}{\hGamma}{\dexp}{\htau'}$.
  \end{enumerate}
\end{thm}
\noindent
%
%The reason analytic expansion produces an expansion of consistent
%type is because the subsumption rule, as previously discussed, allows
%us to check an external expression against any type consistent with
%the type the expression actually synthesizes, whereas every internal
%expression can be assigned at most one type, i.e. the following
%standard unicity property holds of the type assignment system.
%
The reason that $\htau'$ is only consistent with the provided type $\htau$ is because
the subsumption rule permits us to check an external expression against any
type consistent with the type that the expression \emph{actually}
synthesizes, whereas every internal expression can be assigned at most one
type, i.e. the following standard unicity property holds of the type
assignment system.
\begin{thm}[Type Assignment Unicity]
  If $\hasType{\Delta}{\hGamma}{\dexp}{\htau}$
  and $\hasType{\Delta}{\hGamma}{\dexp}{\htau'}$
  then $\htau=\htau'$.
\end{thm}
\noindent
Consequently, analytic elaboration reports the type actually assigned to
the elaboration it produces.
For example, we can derive that
$\elabAna{\hGamma}{c}{\tehole}{c}{b}{\emptyset}$.
% where $\emptyset$ is the empty hole context.

Before describing the rules in detail, let us state two other guiding
theorems.  The following theorem establishes that every well-typed external
expression can be elaborated.
 \begin{thm}[Elaborability] \label{thm:elaborability}~
  \begin{enumerate}[nolistsep]
    \item
      If $\hsyn{\hGamma}{\hexp}{\htau}$
      then $\elabSyn{\hGamma}{\hexp}{\htau}{\dexp}{\Delta}$
      for some $\dexp$ and $\Delta$.
    \item
      If $\hana{\hGamma}{\hexp}{\htau}$
      then $\elabAna{\hGamma}{\hexp}{\htau}{\dexp}{\htau'}{\Delta}$
      for some $\dexp$ and $\htau'$ and $\Delta$.
  \end{enumerate}
 \end{thm}

% \noindent
% The following theorem establishes that when an expansion exists, it is unique.
% \begin{thm}[Expansion Unicity] \label{thm:expansion-unicity}~
%   % \begin{enumerate}[nolistsep]
%   %   \item
%       If $\elabSyn{\hGamma}{\hexp}{\htau}{\dexp}{\Delta}$
%       and $\elabSyn{\hGamma}{\hexp}{\htau'}{\dexp'}{\Delta'}$
%       then $\htau=\htau'$ and $\dexp=\dexp'$ and $\Delta=\Delta'$.
%   %   \item
%   %     If $\elabAna{\hGamma}{\hexp}{\htau_1}{\dexp}{\htau_2}{\Delta}$
%   %     and $\elabAna{\hGamma}{\hexp}{\htau_1}{\dexp'}{\htau_2'}{\Delta'}$
%   %     then $\dexp=\dexp'$ and $\htau_2=\htau_2'$ and $\Delta=\Delta'$.
%   % \end{enumerate}
% \end{thm}
\noindent
The following theorem establishes that elaboration generalizes external
typing.
\begin{thm}[Elaboration Generality] \label{thm:elaboration-generality}~
  \begin{enumerate}[nolistsep]
    \item
      If $\elabSyn{\hGamma}{\hexp}{\htau}{\dexp}{\Delta}$
      then $\hsyn{\hGamma}{\hexp}{\htau}$.
    \item
      If $\elabAna{\hGamma}{\hexp}{\htau}{\dexp}{\htau'}{\Delta}$
      then $\hana{\hGamma}{\hexp}{\htau}$.
  \end{enumerate}
\end{thm}

We also establish that the elaboration produces unique results.
\begin{thm}[Elaboration Unicity] \label{thm:expansion-unicity}~
  \begin{enumerate}[nolistsep]
  \item If $\elabSyn{\hGamma}{\hexp}{\htau_1}{\dexp{}_1}{\Delta_1}$ and
    $\elabSyn{\hGamma}{\hexp}{\htau_2}{\dexp{}_2}{\Delta_2}$, then
    $\htau_1 = \htau_2$ and $\dexp{}_1 = \dexp{}_2$ and $\Delta_1 =
    \Delta_2$.
  \item If
    $\elabAna{\hGamma}{\hexp}{\htau}{\dexp{}_1}{\htau_1}{\Delta_1}$ and
    $\elabAna{\hGamma}{\hexp}{\htau}{\dexp{}_2}{\htau_2}{\Delta_2}$, then
    $\htau_1 = \htau_2$ and $\dexp{}_1 = \dexp{}_2$ and $\Delta_1 =
    \Delta_2$
  \end{enumerate}
\end{thm}

The rules governing elaboration of constants, variables and lambda
expressions---Rules \rulename{ESConst}, \rulename{ESVar}, \rulename{ESLam}
and \rulename{EALam}---mirror the corresponding type assignment rules---
Rules \rulename{TAConst}, \rulename{TAVar} and \rulename{TALam}---and in
turn, the corresponding bidirectional typing rules from
Fig.~\ref{fig:bidirectional-typing}.
%
% Consequently, the corresponding cases of
% Theorem~\ref{thm:typed-expansion}, Theorem~\ref{thm:expandability} and
% Theorem~\ref{thm:expansion-generality} are straightforward.
%
To support type assignment, all lambdas in the internal language are
half-annotated---Rule \rulename{EALam} inserts the annotation when
elaborating an unannotated external lambda based on the given type.
The rules governing hole elaboration, and the rules that perform \emph{cast
  insertion}---those governing function application and type
ascription---are more interesting. Let us consider each of these two groups
of rules in turn in Sec.~\ref{sec:hole-elaboration} and
Sec.~\ref{sec:cast-insertion}, respectively.

\subsubsection{Hole Elaboration}\label{sec:hole-elaboration}
Rules \rulename{ESEHole}, \rulename{ESNEHole}, \rulename{EAEHole} and
\rulename{EANEHole} govern the elaboration of empty and non-empty
expression holes to empty and non-empty \emph{hole closures},
$\dehole{u}{\sigma}{}$ and $\dhole{\dexp}{u}{\sigma}{}$.
The hole name~$u$ on a hole closure identifies the external hole to which
the hole closure corresponds.
While we assume each hole name to be unique in the external language, once
evaluation begins, there may be multiple hole closures with the same name
due to substitution.
For example, the result from Fig.~\ref{fig:grades-example} shows three
closures for the hole named 1.
There, we numbered each hole closure for a given hole sequentially,
\li{1:1}, \li{1:2} and \li{1:3}, but this is strictly for the sake of
presentation, so we omit hole closure numbers from the core calculus.

For each hole, $u$, in an external expression, the hole context generated
by elaboration, $\Delta$, contains a hypothesis of the
form~$\Dbinding{u}{\hGamma}{\htau}$, which records the hole's type, $\tau$,
and the typing context, $\Gamma$, from where it appears in the original
expression.\footnote{ We use a hole context, rather than recording the
  typing context and type directly on each hole closure, to ensure that all
  closures for a hole name have the same typing context and type.}
We borrow this hole context notation from contextual modal type theory
(CMTT) \cite{Nanevski2008}, identifying hole names with metavariables and
hole contexts with modal contexts (we say more about the connection with
CMTT below).
%
% Each hole expansion rule records the ``current'' typing context under which the hole is expanded.
%
In the synthetic hole elaboration rules~\rulename{ESEHole}
and~\rulename{ESNEHole}, the generated hole context assigns the hole
type~$\tehole$ to hole name~$u$, as in the external typing rules.
However, the first two premises of the elaboration subsumption
rule~\rulename{EASubsume} disallow the use of subsumption for holes in
analytic position.
Instead, we employ separate analytic rules~\rulename{EAEHole}
and~\rulename{EANEHole}, which each record the checked type~$\tau$ in the
hole context.
Consequently, we can use type assignment for the internal language --- the
type assignment rules \rulename{TAEHole} and \rulename{TANEHole} in
Fig.~\ref{fig:hasType} assign a hole closure for hole name~$u$ the
corresponding type from the hole context.

Each hole closure also has an associated environment~$\sigma$ which
consists of a finite substitution of the form $[d_1/x_1, ~\cdots, d_n/x_n]$
for $n \geq 0$.
The closure environment keeps a record of the substitutions that occur around the hole as evaluation occurs.
Initially, when no evaluation has yet occurred, the hole elaboration rules
generate the identity substitution for the typing context associated with
hole name~$u$ in hole context~$\Delta$, which we notate $\idof{\hGamma}$,
and define as follows.
\begin{defn}[Identity Substitution] $\idof{x_1 : \tau_1, ~\cdots, x_n : \tau_n} = [x_1/x_1, ~\cdots, x_n/x_n]$
\end{defn}
\noindent
The type assignment rules for hole closures,~\rulename{TAEHole} and \rulename{TANEHole}, each require that the hole closure environment~$\sigma$ be consistent with the corresponding typing context, written as $\hasType{\Delta}{\hGamma}{\sigma}{\hGamma'}$.
Formally, we define this relation in terms of type assignment as follows:
\begin{defn}[Substitution Typing]
$\hasType{\Delta}{\hGamma}{\sigma}{\hGamma'}$ iff $\domof{\sigma} = \domof{\hGamma'}$ and for each $x : \htau \in \hGamma'$ we have that $d/x \in \sigma$ and $\hasType{\Delta}{\hGamma}{d}{\tau}$.
\end{defn}
\noindent
It is easy to verify that the identity substitution satisfies this requirement, i.e. that $\hasType{\Delta}{\hGamma}{\idof{\hGamma}}{\hGamma}$.

Empty hole closures, $\dehole{u}{\sigma}{}$,  correspond to the metavariable closures (a.k.a. deferred substitutions) from CMTT, $\cmttclo{u}{\sigma}$.
\Secref{sec:evaluation} defines how these closure environments evolve during evaluation.
Non-empty hole closures~$\dhole{d}{u}{\sigma}{}$ have no direct correspondence with a notion from CMTT (see Sec.~\ref{sec:resumption}).

\subsubsection{Cast Insertion}\label{sec:cast-insertion}
%
% USE A TOPIC SENTENCE SO THAT THE READER IS GUIDED A LITTLE MORE, e.g.,
%
Holes in types require us to defer certain structural checks to run time.
%-----------------------------------------------------------------------------------------------
%
To see why this is necessary, consider the following
example: $\hap{(\halam{x}{\tehole}{\hap{x}{c}})}{c}$.
Viewed as an external expression, this example synthesizes type
$\tehole$, since the hole type annotation on variable~$x$ permits
applying~$x$ as a function of type~$\tarr{\tehole}{\tehole}$, and base
constant~$c$ may be checked against type~$\tehole$, by subsumption.
However, viewed as an internal expression, this example is not
well-typed---the type assignment system defined in
\Figref{fig:hasType} lacks subsumption.
Indeed, it would violate type safety if we could assign a type to this
example in the internal language, because beta reduction of this
example viewed as an internal expression would result in $c(c)$, which
is clearly not well-typed.
The difficulty arises because leaving the argument type unknown also leaves unknown how
the argument is being used (in this case, as a function).\footnote{In a system where type reconstruction is first used
to try to fill in type holes, we could express a similar example by
using $x$ at two or more different types, thereby causing type
reconstruction to fail.
%
% On the other hand, if it is acceptable to arbitrarily choose one of
%the possible types, and type reconstruction is complete, then type
%holes will never appear in the internal language and the cast
%insertion machinery described in this section can be omitted
%entirely, leaving only the hole closure machinery described
%previously.
}
By our interpretation of hole types as unknown types from gradual type
theory, we can address the problem by performing cast insertion.

The cast form in \HazelnutLive is $\dcasttwo{\dexp}{\htau_1}{\htau_2}$.
This form serves to ``box'' an expression of type $\htau_1$ for
treatment as an expression of a consistent type $\htau_2$
(Rule~\rulename{TACast} in \Figref{fig:hasType}).%
\footnote{
In the earliest work on gradual type theory, the cast form only gave
the target type~$\htau_2$ \cite{Siek06a}, but it simplifies the dynamic semantics substantially
to include the assigned type~$\htau_1$ in the syntax \cite{DBLP:conf/snapl/SiekVCB15}.
}

Elaboration inserts casts at function applications and ascriptions.
The latter is more straightforward: Rule~\rulename{ESAsc}
in \Figref{fig:expandSyn} inserts a cast from the assigned type to the
ascribed type.
Theorem~\ref{thm:typed-elaboration} inductively ensures that the two types
are consistent.
We include ascription for expository purposes---this form is derivable
by using application together with the half-annotated identity, $e
: \tau = \hap{(\halam{x}{\htau}{x})}{e}$; as such, application
elaboration, discussed below, is more general.

Rule~\rulename{ESAp} elaborates function applications.
To understand the rule, consider the elaboration of the example
discussed above, $\hap{(\halam{x}{\tehole}{\hap{x}{c}})}{c}$:
\[
        \hap{\dcasttwo{
        (\halam{x}{\tehole}{\underbrace{
                \hap{\dcasttwo{x}{\tehole}{\tarr{\tehole}{\tehole}}}
                {\dcasttwo{c}{b}{\tehole}}
                }_{\textrm{elaboration of function body}~x(c)}
        }
        )}{\tarr{\tehole}{\tehole}}{
           \tarr{\tehole}{\tehole}}
           }
           {\dcasttwo{c}{b}{\tehole}}
\]
Consider the (indicated) function body,
where elaboration inserts a cast on both the function expression~$x$ and
its argument~$c$.
Together, these casts for~$x$ and~$c$ permit assigning a type to the
function body according to the rules in \Figref{fig:hasType}, where we
could not do so under the same context without casts.
We separately consider the elaborations of~$x$ and of~$c$.

First, consider the function position of this application, here variable~$x$.
Without any cast, the type of variable~$x$ is the hole type~$\tehole$;
however, the inserted cast on~$x$ permits treating it as though it has
arrow type $\tarr{\tehole}{\tehole}$.
The first three premises of Rule~\rulename{ESAp} accomplish this
by first synthesizing a type for the function expression, here
$\tehole$, then
by determining the matched arrow type~$\tarr{\tehole}{\tehole}$, and
finally,
by performing analytic elaboration on the function expression with this
matched arrow type.
The resulting elaboration has some type~$\tau_1'$ consistent with the
matched arrow type.
In this case, because the subexpression~$x$ is a variable, analytic
elaboration goes through subsumption so that type~$\tau_1'$ is
simply~$\tehole$.
The conclusion of the rule inserts the corresponding cast.
We go through type synthesis, \emph{then} analytic elaboration, so that the
hole context records the matched arrow type for holes in function position,
rather than the type~$\tehole$ for all such holes, as would be the case in
a variant of this rule using synthetic elaboration for the function
expression.

Next, consider the application's argument, here constant~$c$.
The conclusion of Rule~\rulename{ESAp} inserts the cast on the argument's
elaboration, from the type it is assigned by the final premise of the
rule~(type~$b$), to the argument type of the matched arrow type of the
function expression~(type~$\tehole$).

The example's second, outermost application goes through the same
application elaboration rule.
In this case, the cast on the function is the identity cast for
$\tarr{\tehole}{\tehole}$.
For simplicity, we do not attempt to avoid the insertion of identity
casts in the core calculus; these will simply never fail during
evaluation.
However, it is safe in practice to eliminate such identity casts during
elaboration, and some formal accounts of gradual typing do so by defining
three application elaboration rules, including the original account of
\citet{Siek06a}.

\subsection{Dynamic Semantics}
\label{sec:evaluation}

To recap, the result of elaboration is a well-typed internal expression
with appropriately initialized hole closures and casts.  This section
specifies the dynamic semantics of \HazelnutLive as a ``small-step''
transition system over internal expressions equipped with a meaningful
notion of type safety even for incomplete programs, i.e. expressions typed
under a non-empty hole context, $\Delta$.  We establish that evaluation
does not stop immediately when it encounters a hole, nor when a cast fails,
by precisely characterizing when evaluation \emph{does} stop. In the case
of complete programs, we also recover the familiar statements of preservation
and progress for the simply typed lambda calculus.

% It is perhaps worth stating at the outset that a dynamic semantics equipped
% with these properties does not simply ``fall out'' from the observations
% made above that (1) empty hole closures correspond to metavariable closures
% from CMTT \cite{Nanevski2008} and (2) casts also arise in gradual type
% theory \cite{DBLP:conf/snapl/SiekVCB15}.
%
% We say more in Sec.~\ref{sec:relatedWork}.
%

%  specify the dynamic semantics of \HazelnutLive.
% %
% The dynamic semantics is capable of running incomplete programs, i.e. those with hole closures, without aborting at holes, or when a cast fails.
% %

% !TEX root = hazelnut-dynamics.tex

\begin{figure}
\begin{subfigure}[t]{0.5\textwidth}
\judgbox{\isGround{\htau}}{$\htau$ is a ground type}
\begin{mathpar}
\inferrule[GBase]{ }{
  \isGround{b}
}

\inferrule[GHole]{ }{
  \isGround{\tarr{\tehole}{\tehole}}
}
\end{mathpar}
\end{subfigure}
\hfill
\begin{subfigure}[t]{0.46\textwidth}
\judgbox{\groundmatch{\htau}{\htau'}}{$\htau$ has matched ground type $\htau'$}
\begin{mathpar}
\inferrule[MGArr]{
  \tarr{\htau_1}{\htau_2}\neq\tarr{\tehole}{\tehole}
}{
  \groundmatch{\tarr{\htau_1}{\htau_2}}{\tarr{\tehole}{\tehole}}
}
\end{mathpar}
\end{subfigure}
\CaptionLabel{Ground Types}{fig:isGround}
\label{fig:groundmatch}
\end{figure}

% !TEX root = hazelnut-dynamics.tex
\begin{figure}

\begin{tabular}[t]{cc}

\begin{minipage}{0.5\textwidth}
\judgbox{\isFinal{\dexp}}{$\dexp$ is final}
\begin{mathpar}
%% \inferrule[FVal]
%% {\isValue{\dexp}}{\isFinal{\dexp}}
\inferrule[FBoxedVal]
{\isBoxedValue{\dexp}}{\isFinal{\dexp}}
\and
\inferrule[FIndet]
{\isIndet{\dexp}}{\isFinal{\dexp}}
\end{mathpar}
\end{minipage}

&

\begin{minipage}{0.5\textwidth}

\judgbox{\isValue{\dexp}}{$\dexp$ is a value}
\begin{mathpar}
\inferrule[VConst]{ }{
  \isValue{c}
}

\inferrule[VLam]{ }{
  \isValue{\halam{x}{\htau}{\dexp}}
}
\end{mathpar}
\end{minipage}

\end{tabular}

\vsepRule

\judgbox{\isBoxedValue{\dexp}}{$\dexp$ is a boxed value}
\begin{mathpar}
\inferrule[BVVal]{
  \isValue{\dexp}
}{
  \isBoxedValue{\dexp}
}

\inferrule[BVArrCast]{
  \tarr{\htau_1}{\htau_2} \neq \tarr{\htau_3}{\htau_4}\\
  \isBoxedValue{\dexp}
}{
  \isBoxedValue{\dcasttwo{\dexp}{\tarr{\htau_1}{\htau_2}}{\tarr{\htau_3}{\htau_4}}}
}

\inferrule[BVHoleCast]{
  \isBoxedValue{\dexp}\\
  \isGround{\htau}
}{
  \isBoxedValue{\dcasttwo{\dexp}{\htau}{\tehole}}
}
\end{mathpar}

\vsepRule

\judgbox{\isIndet{\dexp}}{$\dexp$ is indeterminate}
\begin{mathpar}
\inferrule[IEHole]
{ }
{\isIndet{\dehole{\mvar}{\subst}{}}}

\inferrule[INEHole]
{\isFinal{\dexp}}
{\isIndet{\dhole{\dexp}{\mvar}{\subst}{}}}

\inferrule[IAp]
{\dexp_1\neq
   \dcasttwo{\dexp_1'}
            {\tarr{\htau_1}{\htau_2}}
            {\tarr{\htau_3}{\htau_4}}\\
 \isIndet{\dexp_1}\\
% \isFinal{\dexp_2}~\text{\cy{??}}}
 \isFinal{\dexp_2}}
{\isIndet{\dap{\dexp_1}{\dexp_2}}}

\inferrule[ICastGroundHole] {
  \isIndet{\dexp}\\
  \isGround{\htau}
}{
  \isIndet{\dcasttwo{\dexp}{\htau}{\tehole}}
}

\inferrule[ICastHoleGround] {
  \dexp\neq\dcasttwo{\dexp'}{\htau'}{\tehole}\\
  \isIndet{\dexp}\\
  \isGround{\htau}
}{
  \isIndet{\dcasttwo{\dexp}{\tehole}{\htau}}
}

\inferrule[ICastArr]{
  \tarr{\htau_1}{\htau_2} \neq \tarr{\htau_3}{\htau_4}\\
  \isIndet{\dexp}
}{
  \isIndet{\dcasttwo{\dexp}{\tarr{\htau_1}{\htau_2}}{\tarr{\htau_3}{\htau_4}}}
}

\inferrule[IFailedCast] {
  \isFinal{\dexp}\\
  \isGround{\htau_1}\\
  \isGround{\htau_2}\\
  \htau_1\neq\htau_2
}{
  \isIndet{\dcastfail{\dexp}{\htau_1}{\htau_2}}
}

%% \inferrule[ICast]
%% {\isIndet{\dexp}}
%% {\isIndet{\dcast{\htau}{\dexp}}}

\end{mathpar}

%\vsepRule

\CaptionLabel{Final Forms}{fig:isFinal}
\label{fig:isValue}
\label{fig:isIndet}
\end{figure}

Figures~\ref{fig:isGround}-\ref{fig:step} define the dynamic semantics.
Most of the cast-related machinery closely follows the cast calculus from
the ``refined'' account of the gradually typed lambda calculus
by \citet{DBLP:conf/snapl/SiekVCB15}, which is known to be
theoretically well-behaved.
In particular, \Figref{fig:isGround} defines the judgement
$\isGround{\htau}$, which distinguishes the base type~$b$ and the
least specific arrow type~$\tarr{\tehole}{\tehole}$ as \emph{ground
types}; this judgement helps simplify the treatment of function casts, discussed below.

\Figref{fig:isFinal} defines the judgement $\isFinal{d}$, which
distinguishes the final, i.e. irreducible, forms of the transition system.
The two rules distinguish two classes of final forms: (possibly-)boxed values and
indeterminate forms.%
\footnote{
        Most accounts of the cast calculus distinguish ground types and values
        with separate grammars together with an
        implicit identification convention.
        Our judgemental formulation is more faithful to the mechanization and
        cleaner for our purposes, because we are distinguishing several
        classes of final forms.
}
The judgement $\isBoxedValue{d}$ defines (possibly-)boxed values as either
ordinary values~(Rule~\rulename{BVVal}), or one of two cast forms: casts
between unequal function types and casts from a ground type to the hole
type. In each case, the cast must appear inductively on a boxed value.
These forms are irreducible because they represent values that have been
boxed but have never flowed into a corresponding ``unboxing'' cast,
discussed below.
%
%, which correspond to the values from the cast calculus and include
%the classic values from the lambda calculus, distingui%shed by
%$\isValue{d}$,
%
The judgement $\isIndet{d}$ defines \emph{indeterminate} forms, so named
because they are rooted at expression holes and failed casts, and so,
conceptually, their ultimate value awaits programmer action (see
Sec.~\ref{sec:resumption}). Note that no term is both complete, i.e. has no
holes, and indeterminate.
The first two rules specify that {empty} hole closures are always
indeterminate, and that {non}-empty hole closures are indeterminate when
they consist of a {final} inner expression.
Below, we describe failed casts and the remaining indeterminate forms simultaneously
with the corresponding transition rules.

Figures~\ref{fig:instruction-transitions}-\ref{fig:step} define the transition rules.
Top-level transitions are \emph{steps}, $\stepsToD{}{d}{d'}$, governed by Rule~\rulename{Step} in \Figref{fig:step}, which
(1) decomposes $d$ into an evaluation context, $\evalctx$, and a selected sub-term, $d_0$;
(2) takes an \emph{instruction transition}, $\reducesE{}{d_0}{d_0'}$, as specified in \Figref{fig:instruction-transitions};
and (3) places $d_0'$ back at the selected position, indicated in
the evaluation context by the \emph{mark}, $\evalhole$, to obtain $d'$.%
\footnote{
        We say ``mark'', rather than the more conventional ``hole'', to avoid confusion with the (orthogonal) holes of \HazelnutLive.
        %
        %the form $\evalhole$ in the grammar of
        %evaluation contexts is referred to as the \emph{hole}, but
        %this hole is a technical device entirely orthogonal to the
        %holes of this paper, so we use the term ``mark'' instead.
        }
%
%% \matt{Next paragraph can be dropped entirely for space}
%% %
This approach was originally developed in the reduction semantics of \citet{DBLP:journals/tcs/FelleisenH92} and is the predominant style of operational semantics in the literature on gradual typing.
Because we distinguish final forms judgementally, rather than syntactically, we use a judgemental formulation of this approach called a \emph{contextual dynamics} by \citet{pfpl}.
It would be straightforward to construct an equivalent structural operational semantics \cite{DBLP:journals/jlp/Plotkin04a} by using search rules instead of evaluation contexts (\citet{pfpl} relates the two approaches).

The rules maintain the property that final expressions truly cannot take a step.%
\begin{thm}[Finality] There does not exist $d$ such that both $\isFinal{d}$ and $\stepsToD{}{d}{d'}$ for some $d'$.
\end{thm}

\subsubsection{Application and Substitution}
% !TEX root = hazelnut-dynamics.tex
\begin{comment}
\begin{figure}[t]

\begin{comment}
\vsepRule

\judgbox{\isevalctx{\evalctx}}{$\evalctx$ is an evaluation context}
\begin{mathpar}
\inferrule[ECDot]{ }{
  \isevalctx{\evalhole}
}

%% \inferrule[ECLam]{
%%   \isevalctx{\evalctx}
%% }{
%%   \isevalctx{\halam{x}{\htau}{\evalctx}}
%% }

\inferrule[ECAp1]{
  \isevalctx{\evalctx}
}{
  \isevalctx{\hap{\evalctx}{\dexp}}
}

\inferrule[ECAp2]{
  \maybePremise{\isFinal{\dexp}}\\
  \isevalctx{\evalctx}
}{
  \isevalctx{\hap{\dexp}{\evalctx}}
}

\inferrule[ECNEHole]{
  \isevalctx{\evalctx}
}{
  \isevalctx{\dhole{\evalctx}{\mvar}{\subst}{}}
}

\inferrule[ECCast]{
  \isevalctx{\evalctx}
}{
  \isevalctx{\dcasttwo{\evalctx}{\htau_1}{\htau_2}}
}

\inferrule[ECFailedCast]{
  \isevalctx{\evalctx}
}{
  \isevalctx{\dcastfail{\evalctx}{\htau_1}{\htau_2}}
}
\end{mathpar}
% \end{comment}
\vsepRule

\caption{Evaluation Contexts}
\label{fig:eval-contexts}
\end{figure}
\end{comment}

%% \vsepRule

\begin{figure}
\judgbox{\reducesE{}{\dexp}{\dexp'}}{$\dexp$ takes an instruction transition to $\dexp'$}
\begin{mathpar}
\inferrule[ITLam]{
  \maybePremise{\isFinal{\dexp_2}}
}{
  \reducesE{}{\hap{(\halam{x}{\htau}{\dexp_1})}{\dexp_2}}{[\dexp_2/x]\dexp_1}
}

\inferrule[ITApCast]{
  \maybePremise{\isFinal{\dexp_1}}\\
  \maybePremise{\isFinal{\dexp_2}}\\
  \tarr{\htau_1}{\htau_2} \neq \tarr{\htau_1'}{\htau_2'}
}{
  \reducesE{}
    {\hap{\dcasttwo{\dexp_1}{\tarr{\htau_1}{\htau_2}}{\tarr{\htau_1'}{\htau_2'}}}{\dexp_2}}
    {\dcasttwo{(\hap{\dexp_1}{\dcasttwo{\dexp_2}{\htau_1'}{\htau_1}})}{\htau_2}{\htau_2'}}
}

\inferrule[ITCastId]{
  \maybePremise{\isFinal{\dexp}}
}{
  \reducesE{}{\dcasttwo{\dexp}{\htau}{\htau}}{\dexp}
}

\inferrule[ITCastSucceed]{
  \maybePremise{\isFinal{\dexp}}\\
  \isGround{\htau}
}{
  \reducesE{}{\dcastthree{\dexp}{\htau}{\tehole}{\htau}}{\dexp}
}

\inferrule[ITCastFail]{
  \maybePremise{\isFinal{\dexp}}\\
  \htau_1\neq\htau_2\\\\
  \isGround{\htau_1}\\
  \isGround{\htau_2}
}{
  \reducesE{}
    {\dcastthree{\dexp}{\htau_1}{\tehole}{\htau_2}}
    {\dcastfail{\dexp}{\htau_1}{\htau_2}}
}

\inferrule[ITGround]{
  \maybePremise{\isFinal{\dexp}}\\
  \groundmatch{\htau}{\htau'}
}{
  \reducesE{}
    {\dcasttwo{\dexp}{\htau}{\tehole}}
    {\dcastthree{\dexp}{\htau}{\htau'}{\tehole}}
}

\inferrule[ITExpand]{
  \maybePremise{\isFinal{\dexp}}\\
  \groundmatch{\htau}{\htau'}
}{
  \reducesE{}
    {\dcasttwo{\dexp}{\tehole}{\htau}}
    {\dcastthree{\dexp}{\tehole}{\htau'}{\htau}}
}

%% \inferrule[ITCast]{
%%   \isFinal{d}\\
%%   \hasType{\Delta}{\emptyset}{d}{\tau_2}\\
%%   \tconsistent{\tau_1}{\tau_2}
%% }{
%%   \reducesE{\Delta}{\dcast{\htau_1}{d}}{d}
%% }
%%
%% \inferrule[ITEHole]{ }{
%%   \reducesE{\Delta}{\dehole{\mvar}{\subst}{\unevaled}}{\dehole{\mvar}{\subst}{\evaled}}
%% }
%%
%% \inferrule[ITNEHole]{
%%   \isFinal{d}
%% }{
%%   \reducesE{\Delta}{\dhole{d}{\mvar}{\subst}{\unevaled}}{\dhole{d}{\mvar}{\subst}{\evaled}}
%% }
\end{mathpar}
\CaptionLabel{Instruction Transitions}{fig:instruction-transitions}
\end{figure}

\begin{figure}
$\arraycolsep=4pt\begin{array}{rllllll}
\mathsf{EvalCtx} & \evalctx & ::= &
  \evalhole ~\vert~
  \hap{\evalctx}{\dexp} ~\vert~
  \hap{\dexp}{\evalctx} ~\vert~
  \dhole{\evalctx}{\mvar}{\subst}{} ~\vert~
  \dcasttwo{\evalctx}{\htau}{\htau} ~\vert~
  \dcastfail{\evalctx}{\htau}{\htau}
\end{array}$

\vsepRule

\judgbox{\selectEvalCtx{\dexp}{\evalctx}{\dexp'}}{$\dexp$ is obtained by placing $\dexp'$ at the mark in $\evalctx$}
\begin{mathpar}
\vspace{-3px}
\inferrule[FHOuter]{ }{
  \selectEvalCtx{\dexp}{\evalhole}{\dexp}
}

%% \inferrule[FLam]{
%%   \selectEvalCtx{d}{\evalctx}{d'}
%% }{
%%   \selectEvalCtx{\halam{x}{\htau}{d}}{\halam{x}{\htau}{\evalctx}}{d'}
%% }

\inferrule[FHAp1]{
  \selectEvalCtx{\dexp_1}{\evalctx}{\dexp_1'}
}{
  \selectEvalCtx{\hap{\dexp_1}{\dexp_2}}{\hap{\evalctx}{\dexp_2}}{\dexp_1'}
}

\inferrule[FHAp2]{
  \maybePremise{\isFinal{\dexp_1}}\\
  \selectEvalCtx{\dexp_2}{\evalctx}{\dexp_2'}
}{
  \selectEvalCtx{\hap{\dexp_1}{\dexp_2}}{\hap{\dexp_1}{\evalctx}}{\dexp_2'}
}

%% \inferrule[FHEHole]{ }{
%%   \selectEvalCtx{\dehole{\mvar}{\subst}{}}{\evalhole}{\dehole{\mvar}{\subst}{}}
%% }
%%
%% \inferrule[FHNEHoleEvaled]{ }{
%%   \selectEvalCtx{\dhole{d}{\mvar}{\subst}{\evaled}}{\evalhole}{\dhole{d}{\mvar}{\subst}{\evaled}}
%% }

\inferrule[FHNEHoleInside]{
  \selectEvalCtx{\dexp}{\evalctx}{\dexp'}
}{
  \selectEvalCtx{\dhole{\dexp}{\mvar}{\subst}{}}{\dhole{\evalctx}{\mvar}{\subst}{}}{\dexp'}
}

%% \inferrule[FHNEHoleFinal]{
%%   \isFinal{d}
%% }{
%%   \selectEvalCtx{\dhole{d}{\mvar}{\subst}{\unevaled}}{\evalhole}{\dhole{d}{\mvar}{\subst}{\unevaled}}
%% }

\inferrule[FHCastInside]{
  \selectEvalCtx{\dexp}{\evalctx}{\dexp'}
}{
  \selectEvalCtx{\dcasttwo{\dexp}{\htau_1}{\htau_2}}
                {\dcasttwo{\evalctx}{\htau_1}{\htau_2}}
                {\dexp'}
}

\inferrule[FHFailedCast]{
  \selectEvalCtx{\dexp}{\evalctx}{\dexp'}
}{
  \selectEvalCtx{\dcastfail{\dexp}{\htau_1}{\htau_2}}
                {\dcastfail{\evalctx}{\htau_1}{\htau_2}}
                {\dexp'}
}

%% \inferrule[FHCastFinal]{
%%   \isFinal{d}
%% }{
%%   \selectEvalCtx{\dcast{\htau}{d}}{\evalhole}{\dcast{\htau}{d}}
%% }
\end{mathpar}

\vsepRule

\judgbox{\stepsToD{}{\dexp}{\dexp'}}{$\dexp$ steps to $\dexp'$}
\vspace{-20px}
\begin{mathpar}
\inferrule[Step]{
  \selectEvalCtx{d}{\evalctx}{\dexp_0}\\
  \reducesE{}{\dexp_0}{\dexp_0'}\\
  \selectEvalCtx{\dexp'}{\evalctx}{\dexp_0'}
}{
  \stepsToD{}{\dexp}{\dexp'}
}
\end{mathpar}
\vspace{-6px}
\CaptionLabel{Evaluation Contexts and Steps}{fig:step}
\vspace{-4px}
\end{figure}

Rule \rulename{ITLam} in Fig.~\ref{fig:instruction-transitions} defines the
standard beta reduction transition.
The bracketed premises of the form $\maybePremise{\isFinal{\dexp}}$ in
Fig.~\ref{fig:instruction-transitions}-\ref{fig:step} may be \emph{included}
to specify an eager, left-to-right evaluation strategy, or \emph{excluded} to
leave the evaluation strategy and order unspecified.
In our metatheory, we exclude these premises, both for
the sake of generality, and to support the specification of the fill-and-resume operation~(see \Secref{sec:resumption}).

Substitution, written $[d/x]d'$, operates in the standard capture-avoiding manner~\cite{pfpl} (see \ifarxiv Appendix \ref{sec:substitution} \else the \appendixName~\fi for the full definition).
% Appendix~\ref{sec:additional-defns}).
%
The only cases of special interest arise when substitution reaches a hole closure:
\[
\begin{array}{rcl}
  [d/x]\dehole{u}{\sigma}{} & = & \dehole{u}{[d/x]\sigma}{} \\%
  \substitute{d}{x}{\dhole{d'}{u}{\sigma}{}} & = & \dhole{[d/x]d'}{u}{[d/x]\sigma}{}
\end{array}
\]
In both cases, we write~$[d/x]\sigma$ to perform substitution on each expression in the hole environment~$\sigma$, i.e. the environment ``records'' the substitution.
For example, $\stepsToD{}
    {\hap{(\halam{x}{b}{\halam{y}{b}{\dehole{u}{[x/x, y/y]}{}}})}{c}}
    {\halam{y}{b}{\dehole{u}{[c/x, y/y]}{}}}$.
Beta reduction can duplicate hole closures.
Consequently, the environments of different closures with the same hole name may differ,
e.g., when a reduction applies a function with a hole closure body multiple times as in Fig.~\ref{fig:grades-example}.
Hole closures may also appear within the environments of other hole
closures, giving rise to the closure paths described in
Sec.~\ref{sec:paths}.

The \rulename{ITLam} rule is not the only rule we need to handle function
application, because lambdas are not the only final form of arrow type.
Two other situations may also arise.

First, the expression in function position might be a cast between
arrow types, in which case we apply the arrow cast conversion rule,
Rule \rulename{ITApCast}, to rewrite the application form, obtaining an
equivalent application where the expression~$d_1$ under the function
cast is exposed.
We know from inverting the typing rules that~$d_1$ has type
$\tarr{\htau_1}{\htau_2}$, and that~$d_2$ has type~$\htau_1'$, where
$\tconsistent{\htau_1}{\htau_1'}$.
Consequently, we maintain type
safety by placing a cast on~$d_2$ from~$\htau_1'$ to~$\htau_1$.
The result of this application has type $\htau_2$, but the
original cast promised that the result would have consistent type
$\htau_2'$, so we also need a cast on the result from $\htau_2$ to
$\htau_2'$.

Second, the expression in function position may be indeterminate,
where arrow cast conversion is not applicable,
e.g. $\hap{(\dehole{u}{\sigma}{})}{c}$.
In this case, the application is indeterminate (Rule~\rulename{IAp}
in \Figref{fig:isFinal}), and the application reduces no
further.

\subsubsection{Casts}
Rule \rulename{ITCastId} strips identity casts. The remaining instruction
transition rules assign meaning to non-identity casts.
As discussed in Sec.~\ref{sec:cast-insertion}, the structure of a term
cast \emph{to} hole type is statically obscure,
so we must await a \emph{use} of the term at some other type, via a
cast \emph{away} from hole type, to detect the type error dynamically.
Rules \rulename{ITCastSucceed} and \rulename{ITCastFail} handle this situation when the
two types involved are ground types (Fig.~\ref{fig:isGround}).
If the two ground types are equal, then the cast succeeds and the cast
may be dropped.
If they are not equal, then the cast fails and the failed cast form,
$\dcastfail{\dexp}{\htau_1}{\htau_2}$, arises.
Rule \rulename{TAFailedCast} specifies that a failed cast is well-typed exactly
when $d$ has ground type $\tau_1$ and $\tau_2$ is a ground type
not equal to $\tau_1$.
Rule \rulename{IFailedCast} specifies that a failed cast operates as an
indeterminate form (once $d$ is final), i.e. evaluation does not stop. For simplicity, we do not include blame labels as found in some accounts of gradual typing \cite{DBLP:conf/esop/WadlerF09,DBLP:conf/snapl/SiekVCB15}, but it would be straightforward to do so by recording the blame labels from the two constituent casts on the two arrows of the failed cast.

% Rules \rulename{ITCastSucceed} and \rulename{ITCastFail} only operate at ground type.
%
The two remaining instruction transition rules, Rule~\rulename{ITGround}
and~\rulename{ITExpand}, insert intermediate casts from non-ground type to a
consistent ground type, and \emph{vice versa}.
These rules serve as technical devices, permitting us to restrict our
interest exclusively to casts involving ground types and type holes elsewhere.
Here, the only non-ground types are the arrow types, so the grounding
judgement~$\groundmatch{\tau_1}{\htau_2}$ (\Figref{fig:groundmatch}),
produces the ground arrow type~$\tarr{\tehole}{\tehole}$.
More generally, the following invariant governs this judgement.
\begin{lem}[Grounding]
  If $\groundmatch{\htau_1}{\htau_2}$
  then $\isGround{\htau_2}$
  and $\tconsistent{\htau_1}{\htau_2}$
  and $\htau_1\neq\htau_2$.
\end{lem}

In all other cases, casts evaluate either to boxed values or to
indeterminate forms according to the remaining rules
in \Figref{fig:isFinal}.
Of note, Rule \rulename{ICastHoleGround} handles casts from hole to
ground type that are not of
the form $\dcastthree{\dexp}{\htau_1}{\tehole}{\htau_2}$.

\subsubsection{Type Safety}
The purpose of establishing type safety is to ensure that the static and dynamic semantics of a
language cohere.
We follow the approach developed by \citet{wright94:_type_soundness},
now standard \cite{pfpl}, which distinguishes two type safety
properties, preservation and progress.
To permit the evaluation of incomplete programs, we establish these
properties for terms typed under arbitrary hole context $\Delta$.
We assume an empty typing context, $\hGamma$; to run open programs, the
system may treat free variables as empty holes with a corresponding
name.

The preservation theorem establishes that transitions preserve type
assignment, i.e. that the type of an expression accurately predicts
the type of the result of reducing that expression.

\begin{thm}[Preservation]
  If $\hasType{\Delta}{\emptyset}{\dexp}{\htau}$ and
  $\stepsToD{\Delta}{\dexp}{\dexp'}$ then
  $\hasType{\Delta}{\emptyset}{\dexp'}{\htau}$.
\end{thm}
\noindent
The proof relies on an analogous preservation lemma for instruction
transitions and a standard substitution lemma stated in \ifarxiv Appendix \ref{sec:substitution}\else the \appendixName\fi.
% Appendix~\ref{sec:additional-defns}.
%
Hole closures can disappear during evaluation, so we must have structural weakening 
of $\Delta$.

The progress theorem establishes that the dynamic semantics accounts
for every well-typed term, i.e. that we have not forgotten some
necessary rules or premises.
\begin{thm}[Progress]
  If $\hasType{\Delta}{\emptyset}{\dexp}{\htau}$ then either
  (a) there exists $\dexp'$ such that $\stepsToD{}{\dexp}{\dexp'}$ or
  (b) $\isBoxedValue{\dexp}$ or
  (c) $\isIndet{\dexp}$.
\end{thm}
\noindent
The key to establishing the progress theorem under a non-empty hole
context is to explicitly account for indeterminate forms,
i.e. those rooted at either a hole closure or a failed cast.
The proof relies on canonical forms lemmas stated in \ifarxiv Appendix~\ref{sec:canonical-forms}\else the \appendixName\fi.
% Appendix~\ref{sec:additional-defns}.

\subsubsection{Complete Programs}
Although this paper focuses on running \emph{incomplete} programs, it helps
to know that the necessary machinery does not interfere with running
\emph{complete} programs, i.e. those with no type or expression holes.
%
% Appendix~\ref{sec:additional-defns}
\ifarxiv Appendix~\ref{sec:complete-programs} \else The \appendixName{} \fi defines the predicates~$\isComplete{\htau}$,
$\isComplete{\hexp}$, $\isComplete{\dexp}$ and~$\isComplete{\hGamma}$.
Of note, failed casts cannot appear in complete internal expressions.
The following theorem establishes that elaboration preserves program
completeness.

\begin{thm}[Complete Elaboration]
% ~
  % \begin{enumerate}[nolistsep]
      If $\isComplete{\hGamma}$ and $\isComplete{\hexp}$
      and $\elabSyn{\hGamma}{\hexp}{\htau}{\dexp}{\Delta}$
      then $\isComplete{\htau}$ and $\isComplete{\dexp}$ and $\Delta = \emptyset$.
%     \item
%       If $\isComplete{\hGamma}$ and $\isComplete{\hexp}$
% and $\isComplete{\htau}$
%       and $\elabAna{\hGamma}{\hexp}{\htau}{\dexp}{\htau'}{\Delta}$
%       then $\isComplete{\dexp}$ and $\isComplete{\htau'}$ and $\Delta=\emptyset$
  % \end{enumerate}
\end{thm}

The following preservation theorem establishes that stepping preserves
program completeness.
\begin{thm}[Complete Preservation]
  If $\hasType{\hDelta}{\emptyset}{\dexp}{\htau}$
  and $\isComplete{\dexp}$
  and $\stepsToD{}{\dexp}{\dexp'}$
  then $\hasType{\hDelta}{\emptyset}{\dexp'}{\htau}$
  and $\isComplete{\dexp'}$.
\end{thm}

The following progress theorem establishes that evaluating a complete
program always results in classic values, not boxed values nor
indeterminate forms.
\begin{thm}[Complete Progress]
  If $\hasType{\hDelta}{\emptyset}{\dexp}{\htau}$ and $\isComplete{\dexp}$
  then either there exists a $\dexp'$ such that
  $\stepsToD{}{\dexp}{\dexp'}$, or $\isValue{\dexp}$.
\end{thm}

%% \begin{figure}[!ht]
%%   \begin{definition}
%%     $\hasType{\Delta}{\hGamma}{\sigma}{\hGamma'}$ iff for each $\dexp/x \in \sigma$, we have $x : \htau \in \hGamma'$ and $\hasType{\Delta}{\hGamma}{\dexp}{\htau}$.
%%   \end{definition}
%%   \caption{substitution type assignment}
%%   \label{fig:subassign}
%% \end{figure}

%% \begin{figure}[!ht]
%%   \caption{substitution type assignment}
%% \end{figure}

\vspace{-3px}
\subsection{Agda Mechanization}
\label{sec:agda-mechanization}
\vspace{-2px}

The archived artifact includes our Agda
mechanization  \cite{norell2009dependently,norell:thesis,Aydemir:2005fk}
of the semantics and metatheory of \HazelnutLive,
%including proofs of all of the
including all of the theorems stated above and necessary lemmas.
%
%We choose the
%(as did the mechanization of \Hazelnut by \citet{popl-paper}, though only a few definitions are common).
%
%Agda is a good choice because it is designed to explicitly communicate a proof's structure, as is our goal, rather than relying on proof automation.
%
%Agda itself was also an inspiration for this work because it supports holes, albeit in a more limited form than described here (cf. Sec.~\ref{sec:intro}).
%
%
Our approach is standard: we model judgements as
inductive datatypes, and rules as dependently typed constructors of these judgements.
We adopt Barendregt's convention for bound variables \cite{urban,barendregt84:_lambda_calculus} and hole names, and avoid certain other complications related to substitution by enforcing the requirement that all bound variables in a term are unique when convenient (this requirement can always be discharged by alpha-variation). We encode typing
contexts and hole contexts using metafunctions.
To support this encoding choice, we postulate function extensionality (which is independent of Agda's axioms) \cite{awodey2012inductive}. We encode finite substitutions as an inductive datatype with a base case representing the identity substitution and an inductive case that records a single substitution. Every finite substitution can be represented this way. This makes it easier for Agda to see why certain inductions that we perform are well-founded.  
The documentation provided with the mechanization has more details.

\vspace{-3px}
\subsection{Implementation and Continuity}\label{sec:implementation}
\vspace{-2px}

The \Hazel implementation described in Sec.~\ref{sec:examples}
includes an unoptimized interpreter, written in OCaml, that implements the semantics as described
in this section, with some simple extensions. As with many full-scale systems, there is not currently a formal
specification for the full \Hazel language, but \ifarxiv Appendix~\ref{sec:extensions} \else the \appendixName~\fi 
discusses how the standard approach for deriving a ``gradualized'' version of a
language construct provides most of the necessary scaffolding \cite{DBLP:conf/popl/CiminiS16}, and provides some examples (sum types and numbers).

% The supplemental material includes a browser-based implementation
% of \HazelnutLive.
% % The implementation supports the core calculus of this section but with the notional base type $b$ replaced by the $\tnum$ type from Appendix~\ref{sec:extensions}. The implementation also includes the sum types extension from Appendix~\ref{sec:extensions} and a few minor  conveniences, e.g. \li{let} binding.
% All of the live programming features from \Secref{sec:examples} are available in the implementation essentially as shown, but for this more austere language (with some minor conveniences, notably let binding). Appendix~\ref{sec:impl-screenshots} provides screenshots of the full user interface.

%  % which is a functional reactive program \cite{Elliott:1997jh,DBLP:conf/pldi/CzaplickiC13}
% %
% % Functional Reactive Animation (Hudak 1997) is the canonical FRP cite, not Elm, if you want one; we don't need one, I think. --Matt
% %
% %\cite{DBLP:conf/pldi/CzaplickiC13}

% The implementation is written with the Reason toolchain for OCaml \cite{reason-what,leroy03:_ocaml}
% together with the OCaml \lismall{React} library \cite{OcamlReact}
% and the \lismall{js_of_ocaml} compiler and its associated libraries \cite{vouillon2014bytecode}. This follows the implementation of \Hazel, which, although tracking toward an Elm-like semantics, is implemented in OCaml because the Elm compiler is not yet self-hosted. The implementation of the dynamic semantics consists of a simple evaluator that closely follows the rules specified in this section.

The editor component of the \Hazel implementation is derived
from the structure editor calculus of~\Hazelnut, but with support for more natural cursor-based movement and infix operator sequences (the details of which are beyond the scope of this paper). It exposes a language of structured
edit actions that automatically insert empty and non-empty holes as necessary
to guarantee that every edit state has some (possibly incomplete) type. This corresponds to the top-level Sensibility invariant established for the \Hazelnut calculus by \citet{popl-paper}, reproduced below:
\begin{prop}[Sensibility]
  \label{thrm:sensibility}
  If $\hsyn{\Gamma}{\removeSel{\zexp}}{\htau}$ and
    $\performSyn{\Gamma}{\zexp}{\htau}{\alpha}{\zexp'}{\tau'}$ then
    $\hsyn{\hGamma}{\removeSel{\zexp'}}{\htau'}$.
  % \item If $\hana{\hGamma}{\removeSel{\zexp}}{\htau}$ and
  %   $\performAna{\hGamma}{\zexp}{\htau}{\alpha}{\zexp'}$ then
  %   $\hana{\hGamma}{\removeSel{\zexp'}}{\htau}$.
\end{prop}
\noindent
Here, $\zexp$ is an editor state (an expression with a cursor), and $\removeSel{\zexp}$ drops the cursor, producing an expression ($e$ in this paper). So in words, ``if, ignoring the cursor, the editor state, $\removeSel{\zexp}$, initially has type $\htau$ and we perform an edit action $\alpha$ on it, then the resulting editor state, $\removeSel{\zexp'}$, will have type $\htau'$''.

By composing this Sensibility property with the \Property{Elaborability},
\Property{Typed Elaboration}, \Property{Progress} and
\Property{Preservation} properties from this section, we establish a
uniquely powerful Continuity invariant:
\begin{corol}[Continuity]
  \label{thrm:continuity}
  If $\hsyn{\emptyset}{\removeSel{\zexp}}{\htau}$ and
    $\performSyn{\emptyset}{\zexp}{\htau}{\alpha}{\zexp'}{\tau'}$ then
    $\elabSyn{\emptyset}{\removeSel{\zexp'}}{\htau'}{\dexp}{\Delta}$
      for some $\dexp$ and $\Delta$ such that
$\hasType{\Delta}{\emptyset}{\dexp}{\htau'}$
and either
  (a) $\stepsToD{}{\dexp}{\dexp'}$ for some $d'$ such that $\hasType{\Delta}{\emptyset}{\dexp'}{\htau'}$; or
  (b) $\isBoxedValue{\dexp}$ or
  (c) $\isIndet{\dexp}$.
  % \item If $\hana{\hGamma}{\removeSel{\zexp}}{\htau}$ and
  %   $\performAna{\hGamma}{\zexp}{\htau}{\alpha}{\zexp'}$ then
  %   $\hana{\hGamma}{\removeSel{\zexp'}}{\htau}$.
\end{corol}

This addresses the gap problem: \emph{every} editor state has a {static
meaning} (so editor services like the type inspector from Fig.~\ref{fig:qsort-type-inspector} are always available) and a non-trivial {dynamic meaning} (a result is always available, evaluation does not stop when a hole or cast failure is encountered, and editor services that rely on hole closures, like the live context inspector from Fig.~\ref{fig:grades-sidebar}, are always available).

In settings where the editor does not maintain this Sensibility invariant, but where programmers can manually insert holes, our approach still helps to reduce the severity of the gap problem, i.e. \emph{more} editor states are dynamically meaningful, even if not \emph{all} of them are.% We can formally state this end-to-end continuity corollary as follows:
%
% We make no further claims about the usability or practicality of the implementation (and indeed, there remain many important and decidedly unresolved questions along these lines, which we leave beyond the scope of this paper).
% We also reiterate that the essential ideas developed from type-theoretic first principles in this paper do not require that
% the editor component of the programming environment be implemented as a structure editor (see Sec.~\ref{sec:intro}).
 %,  to scale up these ideas to a ``real-world'' language and programming environment.

%, and provides to be of use to researchers studying the calculus as presented in this section.
% Consistent with this goal, it closely follows the theoretical account in this section, rather than including advanced language features.

%%%%%%%%%%%%%%%%%%%%%%%%%%%%%%%%%%%%%%%%%%%%%%%%%%%%%%%%%%%%%%%%%%%%%%%%%%%%%%%%%%%%%%%%%%%%%%%%%%%

%% \matt{I'd drop everything else in this sub-section;
%% it leaves us open to attack from a hostile reviewer
%% who is cranky that the implementation is not ``complete'' yet; also, we'd save the space}

%%   \halam{x}{\htau}{\evalctx} ~\vert~

%%%%%%%%%%%%%%%%%%%%%%%%%%%%%%%%%%%%%%%%%%%%%%%%%%%%%%%%%%%%%%%%%%%%%%%%%%%

\newcommand{\commutativitySec}{A Contextual Modal Interpretation of Fill-and-Resume}
\section{\commutativitySec}
\label{sec:resumption}

%The result of evaluation is a final internal expression with hole closures, each with an associated hole environment, $\sigma$. These hole environments can be reported directly to the programmer, e.g. via the sidebar shown in Fig.~X\todo{fig}, to help them as they think about how to fill in the corresponding hole in the external expression. Hole environments might also be useful indirectly, e.g. by informing an edit action synthesis and suggestion system. In any case,
When the programmer performs one or more edit actions to fill in an expression hole in the program, a new result must be computed, ideally quickly \cite{DBLP:conf/icse/Tanimoto13,DBLP:journals/vlc/Tanimoto90}. Na\"ively, the system would need to compute the result ``from scratch'' on each such edit. For small exploratory programming tasks, recomputation is acceptable, but in cases where a large amount of computation might occur, e.g. in data science tasks, a more efficient approach is to resume evaluation from where it left off after an edit that amounts to hole filling. This section develops a foundational account of this feature, which we call \emph{fill-and-resume}. This approach is complementary to, but distinct from, incremental computing (which is concerned with changes in input, not code insertions)~\cite{Hammer2014}.

%%%%%%%%% New par 
Formally,
the key idea is to interpret hole environments as \emph{delayed substitutions}. This is the same interpretation suggested for metavariable closures in contextual modal
type theory (CMTT) by \citet{Nanevski2008}.
%
%In practice, it may be useful to cache results from several previous
%expansions, e.g., by employing off-the-shelf programming language
%abstractions for incremental computation~\cite{Hammer14,Hammer15}.
%
%%%%%%%%% New par
\Figref{fig:substitution} defines the hole filling operation~$\instantiate{d}{u}{d'}$
based on the contextual substitution operation of~CMTT.
Unlike usual notions of capture-avoiding substitution,
hole filling imposes no condition on the binder when passing into the
body of a lambda expression---the expression that fills a hole can, of
course, refer to variables in scope where the hole appears.
When hole filling encounters an empty closure for the hole being
instantiated, $\instantiate{d}{u}{\dehole{u}{\sigma}{}}$, the result
is $[\instantiate{d}{u}{\sigma}]d$.
That is, we apply the delayed substitution to the fill expression~$d$
after first recursively filling any instances of hole~$u$ in~$\sigma$.
Hole filling for non-empty closures is analogous, where it discards
the previously-enveloped expression.
%
%
% \matt{Will any reader actually wonder this? Does this thought connect to any other statement elsewhere in the paper?}
%
This case shows why we cannot interpret a non-empty hole as an empty
hole of arrow type applied to the enveloped expression---the hole
filling operation would not operate as expected under this
interpretation.

% !TEX root = hazelnut-dynamics.tex

\begin{figure}
\small
%% TODO use instantiate macro
\judgbox
  {\instantiate{\dexp}{u}{\dexp'} = \dexp''}
  {$\dexp''$ is obtained by filling hole $u$ with $\dexp$ in $\dexp'$}

\vsepRule

\judgbox
  {\instantiate{\dexp}{u}{\sigma} = \sigma'}
  {$\sigma'$ is obtained by filling hole $u$ with $\dexp$ in $\sigma$}
  %% {$\dexp''$ is the result of substituting $\dexp$ for $u$ in $\dexp'$}
\[
\begin{array}{lcll}
\instantiate{\dexp}{u}{c}
&=&
c\\
{\instantiate{\dexp}{u}{x}}
&=&
x\\
%% {[\dexp_1 / x] y}
%% &=&
%% y & (y \neq x)\\
{\instantiate{\dexp}{u}{\halam{x}{\htau}{\dexp'}}}
&=&
{\halam{x}{\htau}{\instantiate{\dexp}{u}{\dexp'}}}\\
{\instantiate{\dexp}{u}{\dap{d_1}{d_2}}}
&=&
{\dapP{\instantiate{\dexp}{u}{d_1}}{\instantiate{\dexp}{u}{d_2}}}
\\
% {[\dexp_1 / x] \dinj{i}{e_2}}
% &=&
% {\dinj{i}{[\dexp_1 / x] e_2}}
% \\
% {[\dexp_1 / x] \dcase{e}{x}{e_x}{y}{e_y}}
% &=&
% {\dcase{[\dexp_1 / x]e}{x}{[\dexp_1 / x]e_x}{y}{[\dexp_1 / x]e_y}}
% \\
\instantiate{\dexp}{u}{\dehole{u}{\subst}{}}
&=&
[\instantiate{\dexp}{u}{\subst}]\dexp
\\
\instantiate{\dexp}{u}{\dehole{v}{\subst}{}}
&=&
\dehole{v}{\instantiate{\dexp}{u}{\subst}}{}
& \text{when $u \neq v$}
\\
\instantiate{\dexp}{u}{\dhole{\dexp'}{u}{\subst}{}}
&=&
[\instantiate{\dexp}{u}{\subst}]\dexp
\\
\instantiate{\dexp}{u}{\dhole{\dexp'}{v}{\subst}{}}
&=&
\dhole{\instantiate{\dexp}{u}{\dexp'}}{v}{\instantiate{\dexp}{u}{\subst}}{}
& \text{when $u \neq v$}
\\
%\end{array}
%\]
%\[
%\begin{array}{lcll}
%
\instantiate{\dexp}{u}{\dcasttwo{\dexp'}{\htau}{\htau'}}
&=&
\dcasttwo{(\instantiate{\dexp}{u}{\dexp'})}{\htau}{\htau'}
\\
\instantiate{\dexp}{u}{\dcastfail{\dexp'}{\htau_1}{\htau_2}}
&=&
\dcastfail{(\instantiate{\dexp}{u}{\dexp'})}{\htau_1}{\htau_2}\\
\instantiate{\dexp}{u}{\cdot} 
&=&
\cdot\\
\instantiate{\dexp}{u}{\sigma, \dexp'/x} 
&=&
\instantiate{\dexp}{u}{\sigma}, \instantiate{\dexp}{u}\dexp'/x
%% {[\dexp_1 / x] \dcast{\htau}{\dexp}}
%% &=&
%% \dcast{\htau}{[\dexp_1 / x] \dexp}
\end{array}
\]
\CaptionLabel{Hole Filling}{fig:substitution}
\vspace{-8px}
\end{figure}

%%%%%%%%% New par
The following theorem characterizes the static behavior of hole filling.
\begin{thm}[Filling]
  If $\hasType{\hDelta, \Dbinding{u}{\hGamma'}{\htau'}}{\hGamma}{\dexp}{\tau}$
  and $\hasType{\hDelta}{\hGamma'}{\dexp'}{\htau'}$
  then $\hasType{\hDelta}{\hGamma}{\instantiate{\dexp'}{u}{\dexp}}{\tau}$.
\end{thm}

Dynamically, the correctness of fill-and-resume depends on
the following \emph{commutativity} property: if there is some sequence of
steps that go from $d_1$ to $d_2$, then one can fill a hole in these
terms at \emph{either} the beginning or at the end of that step
sequence.
We write $\multiStepsTo{\dexp_1}{\dexp_2}$ for the reflexive,
transitive closure of stepping (see \ifarxiv Appendix~\ref{sec:multi-step}\else the \appendixName\fi).
% Appendix~\ref{sec:additional-defns}).
%
\begin{thm}[Commutativity]
  If $\hasType{\hDelta, \Dbinding{u}{\hGamma'}{\htau'}}{\emptyset}{\dexp_1}{\tau}$
  and $\hasType{\hDelta}{\hGamma'}{\dexp'}{\htau'}$ and $\multiStepsTo{\dexp_1}{\dexp_2}$
  then $\multiStepsTo{\instantiate{\dexp'}{u}{\dexp_1}}
                     {\instantiate{\dexp'}{u}{\dexp_2}}$.
\end{thm}
The key idea is that we can resume evaluation by replaying the substitutions that were recorded in the closures and then taking the necessary ``catch up'' steps that evaluate the hole fillings that now appear.
The caveat is that resuming from $\instantiate{\dexp'}{u}{d_2}$ will not
reduce sub-expressions in the same order as a ``fresh'' eager left-to-right
reduction sequence starting from $\instantiate{\dexp'}{u}{d_1}$ so 
filling commutes with reduction only for 
languages where evaluation order ``does not matter'', e.g. pure functional languages like \HazelnutLive.\footnote{There are various standard ways to formalize this intuition, e.g. by stating a suitable confluence property. 
For the sake of space, we review confluence in \ifarxiv Appendix \ref{sec:confluence}\else the \appendixName\fi.} Languages with non-commutative effects do not enjoy this property.
% Appendix~\ref{sec:confluence}.)

%% Everyone knows what confluence means, or they should.  We don't need to cite Church :)
%%

We describe the proof, which is straightforward but involves a
number of lemmas and definitions, in \ifarxiv Appendix~\ref{sec:hole-filling}\else the \appendixName\fi.
%Appendix~\ref{sec:hole-filling}.
%
In particular, care is needed to handle the situation where a
now-filled non-empty hole had taken a step in the original evaluation trace.

% \matt{We could (should?) move the next paragraph to the appendix,
% before/after the proofs mentioned above;
% its not really adding much to the discussion above,
% but may be interesting to the very studious reader/prover; also, leaving it here also gives a lazy/cranky reviewer a nice weapon to hit us with}

We do not separately define hole filling in the external language (i.e. we
consider a change to an external expression to be a hole filling if the new
elaboration differs from the previous elaboration up to hole filling).  In
practice, it may be useful to cache more than one recent edit state to take
full advantage of hole filling.  As an example, consider two edits, the
first filling a hole~$u$ with the number~$2$, and the next applying
operator~$+$, resulting in $2 + \dehole{v}{\sigma}{}$.
This second edit is not a hole filling edit with respect to the
immediately preceding edit state, $2$, but it can be understood as filling
hole $u$ from two states back with $2 + \dehole{v}{\sigma}{}$.

%%%%%%%%%%%%%%%%%%%%%%%%%%%%%%%%%%%%%%%%%%%%%%%%%%%%%%%%%%%%%%%%%%%%%%%%%%%5

Hole filling also allows us to give a contextual modal interpretation
to lab notebook cells  like those of
Jupyter/IPython \cite{PER-GRA:2007} (and read-eval-print loops as a restricted case
where edits to previous cells are impossible).
Each cell can be understood as a series of \li{let}
bindings ending implicitly in a hole, which is filled by the next cell.
The live environment in the subsequent cell is exactly the hole environment of this implicit trailing hole.
Hole filling when a subsequent cell changes avoids recomputing the
environment from preceding cells, without relying on mutable state. Commutativity provides a reproducibility guarantee
missing from Jupyter/IPython, where editing and executing previous cells can cause the state to differ substantially from the state that would result when attempting to run the notebook from the top.
% TODO citation

\begin{comment}

\begin{theorem}[Maximum Informativity]
If the elaboration produces $t1$, and there exists another possible type
choice $t2$, then $t1 \sim t2$ and $t1 JOIN t2 = t1$
\end{theorem}\footnote{idea is that special casing the holes in EANEHole gives you ``the
most descriptive hole types'' for some sense of what that means -- they'd
all just be hole other wise. from Matt:
\begin{quote}
It sounds like we need a something akin to an abstract domain (a lattice),
where hole has the least information, and a fully-defined type (without
holes) has the most information.  You can imagine that this lattice really
expands the existing definition we have of type consistency, which is
merely the predicate that says whether two types are comparable
(“join-able”) in this lattice.  lattice join is the operation that goes
through the structure of two (consistent) types, and chooses the structure
that is more defined (i.e., non-hole, if given the choice between hole and
non-hole).

The rule choosenonhole below is the expansion of this consistency rule that
we already have (hole consistent with everything)
\end{quote}}

\begin{verbatim}
t not hole
-------------------- :: choose-non-hole
hole JOIN t  = t
\end{verbatim}
\begin{verbatim}
------------ :: hole-consistent-with-everything
hole ~ t
\end{verbatim}

\end{comment}

% \input{implementation}
% !TEX root = hazelnut-dynamics.tex
\newcommand{\relatedWorkSection}{Related and Future Work}
\section{\relatedWorkSection} % don't like the all-caps thing that the template does, so protecting it from that
\label{sec:relatedWork}

This paper builds directly on the static semantics developed by \citet{popl-paper}. That paper, as well as a subsequent ``vision paper'' introducing \Hazel \cite{HazelnutSNAPL}, suggested as future work  
a corresponding dynamic semantics for the purposes of live programming without gaps. \citet{Bayne:2011:ASD:1985793.1985864} have also extensively argued the value of assigning meaning even to incomplete and erroneous programs. 
This paper delivers conceptual and type-theoretic details necessary to advance these visions.

\parahead{Gradual Type Theory}

The semantics borrows machinery related to type holes from gradual type theory \cite{Siek06a,DBLP:conf/snapl/SiekVCB15}, as discussed
at length in Sec.~\ref{sec:calculus}. The main innovation relative to this
prior work is the treatment of cast failures like holes, rather than errors.
Many of the methods developed to make gradual type systems more expressive and practical are directly relevant to future directions for \Hazel and other implementations of the ideas herein \cite{takikawa_et_al:LIPIcs:2015:5215}. For example, there has been substantial work on the problem of implementing casts efficiently \cite{herman2010space,siek2010threesomes,garcia2013threesomes}. There has also been substantial work on integrating gradual typing with polymorphism \cite{DBLP:conf/esop/XieBO18,DBLP:journals/pacmpl/DevriesePP18,Igarashi:2017:PGT:3136534.3110284}, and with refinement types \cite{DBLP:conf/popl/LehmannT17}.

%% Although our expectation is that \Hazel will remain a pure functional language \emph{a la} \Elm~or Haskell,
%
Another interesting future direction would be to move beyond
pure functional programming and carefully integrate imperative features, e.g. ML-style references.
%
%% These require particular care, as recognized by \citet{Siek06a} and \citet{DBLP:conf/esop/SiekVCTG15}, but the techniques are known and we see no reason why the type safety properties  established in this paper would not be conserved, with suitable modifications to account for a store.
%
\citet{DBLP:conf/esop/SiekVCTG15} show how to incorporate such features into gradual type theory; we expect that this approach would also work for the semantics of Sec.~\ref{sec:calculus}, i.e. it would conserve the type safety properties established there, with suitable modifications to account for a store. However, the commutativity property we establish in Sec.~\ref{sec:resumption} will not hold for a language that supports non-commutative effects. We leave to future work the task of defining more restricted special cases of the fill-and-resume operation that respects a suitable commutativity property in effectful settings (perhaps by checkpointing or by using a type and effect system to granularly determine where re-evaluation is needed \cite{burckhardt2013s}).

Going beyond references to incorporate external effects, e.g. IO effects, raises some additional practical concerns, however---we do not want to continue past a hole or error and in so doing haphazardly trigger an unintended effect. In this setting, it is likely better to explicitly ask the programmer whether to continue evaluation past a hole or cast failure. The small step specification in this paper is suitable as the basis for a  step-based evaluator like this. \citet{ocaml-stepper} discuss some unresolved usability issues relevant to single steppers for functional languages. 

We did not give type holes unique names nor did we define a commutative type hole filling operation, but we conjecture that it is possible to do so by using non-empty expression holes to replace casts that go from the filled hole to an inconsistent type.

\parahead{Metavariables in Type Theory}

The other major pillar of related work is the work on contextual modal type theory (CMTT) by \citet{Nanevski2008}, which we also discussed at length throughout the paper. To reiterate, there is a close relationship between expression holes in this paper and metavariables in CMTT: both are associated with a type and a typing context. 
Empty hole closures 
relate to the concept of a {metavariable closure} (delayed substitution) in CMTT, which consists
of a metavariable paired with a substitution for all of the variables in the
typing context associated with that metavariable. 
Empty expression hole filling relates to contextual substitution. 

%% Although these connections with gradual typing and CMTT are encouraging, our contributions do not neatly fall out from this prior work.
%
These connections with gradual typing and CMTT are encouraging, but our contributions do not immediately fall out from this prior work.
The problem is first that \citet{Nanevski2008} defined only the logical reductions for CMTT, viewing it as a proof system for intuitionistic contextual modal logic via the propositions-as-types (Curry-Howard) principle. 
The paper therefore proved only a subject reduction property (which is closely related to type preservation) and sketched an interpretation of CMTT into the simply-typed lambda calculus with sums under permutation conversion, %, 
which has been studied by \citet{DBLP:journals/iandc/Groote02}. 
% Permutation conversions are necessary to encode the commuting reductions of CMTT, which in turn are necessary to prove a strong normalization property.
\citet{DBLP:conf/ppdp/PientkaD08} more directly defines a dynamic semantics for an extension of CMTT. In both cases, the issue is that these systems can evaluate only closed terms, while we need to consider terms with free metavariables. 

There has been much formal work on reduction of open terms. In particular, the typical approach is to define the \emph{weak head normal forms} \cite{barendregt84:_lambda_calculus,DBLP:journals/corr/abs-1009-2789,Abramsky:1990vv}, and a weak head normalization procedure. 
This is useful when using reduction to perform optimizations throughout a program, e.g. when using supercompilation-by-evaluation \cite{DBLP:conf/haskell/BolingbrokeJ10} or symbolic evaluation \cite{King:1976,SurveySymExec-CSUR18}, or when using evaluation in the service of equational reasoning as in dependently typed proof assistants. Notably, \citet{DBLP:journals/corr/abs-1009-2789} considered weak head normalization for CMTT, which forms the basis for equational reasoning in the Beluga proof assistant \cite{DBLP:conf/flops/Pientka10}. There are two points of note here. First, the addition of type holes allow us to express general recursion \cite{Siek06a}, so we cannot rely on normalization arguments and instead need a more conventional dynamic semantics with a progress theorem. Second, we do not want to evaluate under arbitrary binders but rather only around holes when they would otherwise have been evaluated \cite{DBLP:conf/birthday/BlancLM05}. As such, we do not need to consider terms with free variables, and can restrict our interest to only those with free metavariables. These considerations together lead us to the indeterminate forms and to the progress theorem that we established.

There are various other systems similar in many ways to CMTT in 
that they consider the problem of reasoning about metavariables. 
For example, McBride's OLEG is another system for reasoning contextually about metavariables
\cite{DBLP:phd/ethos/McBride00}, and it is the conceptual basis of certain hole refinement features in Idris \cite{brady2013idris}. \citet{DBLP:conf/csl/GeuversJ02} discuss similar ideas, again in the setting of hole refinement in proof assistants. However, these systems do not account for substitutions around metavariables. The systems underlying TypeLab \cite{Strecker:98a} and Alf \cite{magnusson1994implementation} do record substitutions around metavariables. These can be considered predecessors of CMTT, which is unique in that it has a clear Curry-Howard interpretation. The discussion above applies similarly to these systems.

We use the machinery 
borrowed from CMTT only extralinguistically.
A key feature of CMTT that we have not yet touched on is the
\emph{internalization} of metavariable binding and contextual
substitution via the contextual modal types, $[\hGamma]\htau$, which
are introduced by the operation $\mathsf{box}(\hGamma.d)$ and
eliminated by the operation $\mathsf{letbox}(d_1, u.d_2)$.
A program with expression holes can be interpreted as being bound under a number of these $\mathsf{letbox}$ constructs, i.e. CMTT serves as a ``{development calculus}'' \cite{DBLP:phd/ethos/McBride00}. Live programming corresponds to reduction under the metavariable binders interleaved with elimination steps. 
This suggests the possibility of \emph{computing} hole fillings by specifying a non-trivial expression of the appropriate contextual 
modal type in this development calculus.
This could, in turn, serve as the basis for a computational hole
refinement system that supports efficient live programming, extending the capabilities of purely static hole
refinement systems like those available in many languages,
e.g. the editor-integrated system of Idris
\cite{brady2013idris,Korkut:2018:ETE:3240719.3241791} and the hole refinement system of Beluga
\cite{DBLP:conf/flops/Pientka10,pientka2015inductive}.
Each applied hole filling can be interpreted as inducing a new dynamic
\emph{edit stage}.
This contextual modal interpretation of live hole refinement
nicely mirrors the modal interpretation of staging and partial evaluation
\cite{Davies:2001op}, with indeterminate forms corresponding to the residual programs that arise when performing partial evaluation. 
The difference is that staging and partial evaluation systems evaluate around an input that sits outside of a
function \cite{Jones:1993uq}, whereas holes are contextual, i.e. they are located inside the program. 

There is also more general work on explicit substitutions, which records all  substitutions, not just substitutions around metavariables, explicitly \cite{Abadi:1991fr,levy1999explicit}.  \citet{DBLP:journals/corr/abs-1009-2789} have developed a theory of explicit substitutions for CMTT, which, following other work on explicit substitutions and environmental abstract machines \cite{DBLP:journals/tcs/Curien91}, could be useful in implementing \HazelnutLive more efficiently by delaying both standard and contextual substitutions until needed during evaluation. We leave this and other questions of fast implementation (e.g. using thunks to encapsulate indeterminate sub-expressions) to future work. 

% We also cannot rely on, for example, weak head normalization because \HazelnutLive admits non-termination (due to casts).\todo{citation} 

% Furthermore, we need to integrate casts into the dynamic semantics. 
% Fortunately, the dynamic semantics for the cast calculus from the gradually typed lambda calculus provides most, but not all, of the necessary machinery. 
% The first problem is again with progress: in the cast calculus, the only irreducible terms of hole type are casts, which are accounted for by the progress theorem, but in $\HazelnutLive$, holes induce additional irreducible terms. 
% The second missing piece is that in prior work on casts, evaluation aborts when cast failure occurs. 
% Our goal, as discussed in Sec.~\ref{sec:failed-cast-example}\todo{where?}{}, is for cast failure to instead insert a membrane around the dynamic type error, 
% much like a non-empty hole serves as a membrane around a static type error, 
% allowing in both cases for evaluation to safely and meaningfully continue past the error when desired.

% \begin{itemize}
	%% \item Agda
	%% \item Idris
	%% \item GHC holes
	%% \item Visual Studio (and others) support for edit-and-resume
	% \item Scratch lets you just skip over statement holes
	% \item papers that show up in a search for ``typed holes'' -- \url{https://scholar.google.com/scholar?hl=en&as_sdt=0%2C39&q=%22typed+holes%22&btnG=}
% \end{itemize}

\parahead{Type Error Messages}

A key feature of our semantics is that it permits evaluation not only of terms with empty holes, but also terms with non-empty holes, i.e. reified static type inconsistencies. DuctileJ \cite{Bayne:2011:ASD:1985793.1985864} and GHC \cite{DBLP:conf/icfp/VytiniotisJM12} 
have also considered this problem, but have taken
an ``exceptional approach'' --- these systems can defer static type errors until run-time, but do not continue further once the term containing the error
 has been evaluated.

Understanding and debugging static type errors is notoriously difficult,
particularly for novices.
A variety of approaches have been
proposed to better localize
and explain type errors \cite{Seminal,ChenErwig2014,DBLP:journals/jfp/ChenE18,Pavlinovic2015,sherrloc}.
One of these approaches, by \citet{Seidel2016},  uses symbolic execution and program synthesis to generate a dynamic witness
that demonstrates a run-time failure that could be caused by the static type error.
\HazelnutLive{} has similar motivations in that it can run programs with type errors  and provide concrete feedback about the values that erroneous terms take during evaluation (Sec.~\ref{sec:static-errors}-\ref{sec:dynamic-type-errors}). However, no attempt is made to synthesize examples that do not already appear in the program. 
%
% Combining the strengths of these approaches may be fruitful in the future. 

\parahead{Coroutines} 
The fill-and-resume interaction is reminiscent of the interactions that occur when using coroutines \cite{DBLP:conf/ifip/KahnM77} and related mechanisms (e.g. algebraic effects) --  a coroutine might yield to the caller until a value is sent back in and then continue in the suspended environment. The difference is that 
the hole filling can make use of the context
around the hole. 

\parahead{Dynamic Error Propagation} 
\citet{DBLP:conf/sp/HritcuGKPM13} consider the problem of dynamic violations of information-flow control (IFC) policies. An exceptional approach here is problematic in practice, because it would lead critical systems to shutdown entirely. Instead, the authors develop several mechanisms for 
propagating errors through subsequent computations (in a manner that preserves non-interference properties). The most closely related is the NaV-lax approach, which turns errors into special ``not-a-value'' (NaV) terms that consume other terms that they interact with (like floating point NaN). Our approach differs in that terms are not consumed. Furthermore, we track hole closures and consider static and dynamic type errors, but not exception propagation.  

\parahead{Debuggers}

Our approach is reminiscent of the workflow that debuggers make available using breakpoints \cite{fitzgerald2008debugging,DBLP:journals/jfp/TolmachA95}, visualizers of program state \cite{Nelson2017,Guo13}, and a variety of logging and tracing capabilities.
Debuggers do not directly support incomplete programs, so the programmer first needs to insert suitable dummy exceptions as discussed in Sec.~\ref{sec:intro}.
Beyond that, there are two main distinctions. First, evaluation does not stop at each hole and so it is straightforward to explore the \emph{space} of values that a variable takes. Second, breakpoints, logs and tracing tools convey the values of a variable at a position in the evaluation trace. Hole closures, on the other hand, convey information from a syntactic position in the \emph{result} of evaluation. The result is, by nature, a simpler object than the trace. %  These approaches are fundamentally complementary in that standard debugging facilities are useful even when there are no holes in the program whereas our focus has been on incomplete programs. 

Some debuggers support ``edit-and-continue'', e.g. Visual Studio \cite{VSEditAndContinue} and Flutter \cite{flutter}, based on the dynamic software update (DSU) capabilities of the underlying run-time system \cite{DBLP:journals/toplas/StoyleHBSN07,DBLP:conf/vstte/HaydenMHFF12,DBLP:journals/toplas/HicksN05}. These features do not come with any guarantee that rerunning the program will produce the same result.

\parahead{Structure Editors}

Holes play a prominent role in structure editors, and indeed the prior work on \Hazelnut was primarily motivated by this application \cite{popl-paper}. 
Most work on structure editors has focused on the user interfaces that they
present. This is important work---presenting a fluid user interface involving
exclusively structural edit actions is a non-trivial problem that has not yet
been fully resolved, though recent studies have started to show productivity
gains for blocks-based structure editors like Scratch for novice programmers \cite{Resnick:2009:SP:1592761.1592779,DBLP:conf/chi/WeintropASFLSF18,DBLP:conf/acmidc/WeintropW15}, and for keyboard-driven structure
editors like \li{mbeddr} in professional programming settings~\cite{DBLP:conf/vl/Asenov014,DBLP:conf/sle/VolterSBK14,voelter_mbeddr:_2012}.
Some structure editors support live programming in various ways, e.g. Lamdu is a structure editor for a typed functional language that displays variable values inline based on recorded execution traces (see above) \cite{lamdu}. However, Lamdu takes the exceptional approach to holes. Scratch will execute a program with holes by simply skipping over incomplete commands, but this is a limited protocol because rerunning the program after filling the hole may produce a result unrelated to initial result. Though in some situations, skipping over problematic commands has been observed to work surprisingly well \cite{DBLP:conf/dac/Rinard12}, this paper focuses on a sound approach. 

\parahead{Scaling Challenges}

We make no empirical claims regarding the usability of the particular user interface presented in this paper. The \Hazel user interface as presented is a proof of concept
that demonstrates (1) a comprehensive solution to the gap problem, as described in Sec.~\ref{sec:implementation}; and 
(2) one possible user interface for presenting hole closures, including deeply nested hole closures, to the programmer. We leave detailed empirical evaluation to future work, and note that the purpose of defining a formal semantics is to provide a foundation for a variety of user interface experiments.

For larger programs, there are some limitations that will certainly require further UI refinement. 
Indeterminate results can get to be quite large, particularly when a hole or failed cast
appears in guard position as in Fig.~\ref{fig:cast-errors}. 
One approach is for the pretty printer can elide irrelevant branches. The semantics in Sec.~\ref{sec:calculus} would also allow for evaluation to pause when such
a form is produced, continuing only if requested. 

In larger programs, there can be many dozens of variables in scope, so search and sorting tools in the live context inspector would be useful. It would also be useful
to support the evaluation of arbitrary ``watched'' expressions in the selected 
closure. There are other ways to display the variable values, e.g. intercalated
into the code \cite{DBLP:journals/corr/abs-1806-07449,lamdu}.

Finally, we have not yet investigated nor attempted to optimize the memory and performance overhead of tracking hole closures. We note simply that it is often the case that 
programmers use small example inputs during initial development, providing larger inputs
only after the program is complete. There is no run-time overhead when there are no holes remaining.

\parahead{Program Slicing}

Hole-like constructs also appear in work on program slicing
\cite{DBLP:conf/icfp/PereraACL12,DBLP:journals/pacmpl/RicciottiSPC17}, where empty expression holes arise as a technical device to determine which parts of a complete program do not impact a selected part of the result of evaluating that program. In other words, the focus is on explaining a result that is presumed to exist, whereas the work in this paper is focused on producing results where they would not have otherwise been available. Combining the strengths of these approaches is another fruitful avenue for future research.

\parahead{Program Synthesis}

Expression holes also often appear in the context of program synthesis, serving as
placeholders in \emph{templates}~\cite{srivastava2013template} or
\emph{sketches}~\cite{solar2009sketching} to be filled in by an expression
synthesis engine. We leave to future work the exciting possibility of combining these approaches. For example, one can imagine running an incomplete program as described in this paper and then adding tests or assertions about the value that a hole should take on under each of the different associated hole closures, via the live context inspector. These could serve as  constraints for use by a type-and-example-driven synthesis engine \cite{DBLP:conf/popl/FrankleOWZ16}. Relatedly, \emph{prorogued programming} proposes soliciting an external source, e.g. the programmer, for a suitable value when a hole instance is encountered \cite{DBLP:conf/oopsla/AfshariBS12}.

% \begin{itemize}
% 	\item work on debuggers that allow you to inspect environments
%   \item might be something relevant in the paper ``A Debugger for Standard ML'' 
%   \item "Visualizing the evaluation of functional programs for debugging" by Whitington and Ridge
%   \item "A lightweight interactive debugger for Haskell'' and ``Multiple-View Tracing for Haskell: a New Hat'' might be relevant
%   \item ocamli -- \url{https://github.com/johnwhitington/ocamli}
%   \item Better supporting debugging aids learning a novel programming language. -- Scaffidi at VLHCC 2017
%   \item quote from Wadler in ``Why no one uses functional languages'':
%     \begin{quote}
%     “...there are few debuggers or
% profilers for strict [functional] languages, perhaps because constructing them is not considered
% research. This is a shame, since such tools are sorely needed, and there remains much of
% interest to learn about their construction and use.
%     \end{quote}
%    \item edit-and-resume in visual studio, others
% \end{itemize}

% Visual Studio edit-and-continue \url{https://docs.microsoft.com/en-us/visualstudio/debugger/edit-and-continue}

% \begin{itemize}
% 	% \item simply typed underdeterminism paper
% 	% \item DuctileJ stuff -- \url{https://homes.cs.washington.edu/~mernst/pubs/ductile-icse2011.pdf}
% 	% \item ``Achieving flexibility in direct-manipulation programming environments by relaxing the edit-time grammar'' -- \url{http://ieeexplore.ieee.org/document/1509511/}
% 	\item 
% \end{itemize}

% !TEX root = hazelnut-dynamics.tex
\newcommand{\discussionSection}{Conclusion}
\section{\discussionSection} % don't like the all-caps thing that the template does, so protecting it from that
\label{sec:discussion}

\begin{comment}
To conclude, we quote Weinberg from The Psychology of Computer Programming (1998): ``how truly sad it is that just at the very moment
when the computer has something important to tell us, it starts
speaking gibberish.''
\end{comment}
\vspace{3pt}
\begin{quote}
\textit{``[H]ow truly sad it is that just at the very moment
when the computer has something important to tell us, it starts
speaking gibberish.''}

\vspace{3pt}

\hfill{}--- Gerald Weinberg, The Psychology of Computer Programming \cite{weinberg1971psychology}
\end{quote}
\vspace{3pt}

\noindent
Weinberg's sentiment applies to countless situations that programmers face as they work to develop and critically investigate the mental model of the program they are writing---at the very moment where rich feedback might be most helpful, e.g. when there is an error in the program or when the programmer is unsure how to fill a hole, the computer often has comparatively little feedback to offer (perhaps just a parser error message, or an austere explanation of a type error). 
%
% Efforts to develop tools that help programmers explore the behavior of parts of their program interactively, i.e. live programming tools, have been increasingly common in the literature on programming languages, software engineering, and human-computer interaction. 
%
% Yet much of this work has assumed that the programmer has already written a formally complete program, or relies on \emph{ad hoc} heuristics to handle some incomplete states. 
%
It is our hope that the well-behaved  type-theoretic foundations developed here will enable not only \Hazel but live programming tools of a wide variety of other designs to further narrow the temporal and perceptive gap and provide meaningful feedback to programmers at these important   moments.

\newcommand{\acksSection}{Acknowledgments}
\section*{\acksSection} 
\label{sec:acks}

We thank Brigitte Pientka, Jonathan Aldrich, Joshua Sunshine, Michael Hilton, Claire LeGoues, Conor McBride, 
the participants of the PL reading group at UChicago,
the participants of the LIVE 2017 and Off the Beaten Track (OBT)
2017 workshops, and the anonymous referees for their insights and feedback on various iterations of this work.
This material was supported by a gift from
Facebook, from the National Science Foundation under grant
numbers CCF-1619282, SHF-1814900 and SHF-1817145, and from AFRL and DARPA under agreement \#FA8750-16-2-0042. 
 The U.S. Government is authorized to reproduce and distribute reprints for 
Governmental purposes notwithstanding any copyright notation thereon.
Any opinions, findings, and conclusions or recommendations expressed
in this material are those of the authors and do not necessarily
reflect the views of Facebook, NSF, DARPA, AFRL or the U.S. Government. 
%
% \clearpage
% \todo{discussion}

% \rkc{something about UI choices for displaying large indeterminate expressions}

% \rkc{something about UI choices for displaying hole environment traces}

% \rkc{something about patterns of use for debugging, such as inserting holes
% around all expressions in tail position}

% \rkc{something about visualizing hole environment traces}

% \rkc{\citet{Seidel2016} define a partial evaluation with a form of holes. compare.
% compare ``flattened'' sumList indeterminate result to their trace. could be
% useful to study whether evaluating the type-correct parts of the evaluation
% trace help students focus on the part that has the error.}

% \rkc{quick mention that lazy evaluation with run-time exceptions for holes is
% not enough. this would delay the failure a bit longer, but not indefinitely.}

% live mobile application development systems like Flutter \cite{flutter};

%\clearpage
\bibliography{all.short}


\begin{thebibliography}{124}

%%% the Web address.

\ifx \showCODEN    \undefined \def \showCODEN     #1{\unskip}     \fi
\ifx \showDOI      \undefined \def \showDOI       #1{#1}\fi
\ifx \showISBNx    \undefined \def \showISBNx     #1{\unskip}     \fi
\ifx \showISBNxiii \undefined \def \showISBNxiii  #1{\unskip}     \fi
\ifx \showISSN     \undefined \def \showISSN      #1{\unskip}     \fi
\ifx \showLCCN     \undefined \def \showLCCN      #1{\unskip}     \fi
\ifx \shownote     \undefined \def \shownote      #1{#1}          \fi
\ifx \showarticletitle \undefined \def \showarticletitle #1{#1}   \fi
\ifx \showURL      \undefined \def \showURL       {\relax}        \fi
% invisible to TeX
\providecommand\bibfield[2]{#2}
\providecommand\bibinfo[2]{#2}
\providecommand\natexlab[1]{#1}
\providecommand\showeprint[2][]{arXiv:#2}

\bibitem[\protect\citeauthoryear{Abadi, Cardelli, Curien, and L{\'{e}}vy}{Abadi
  et~al\mbox{.}}{1991}]%
        {Abadi:1991fr}
\bibfield{author}{\bibinfo{person}{Mart{\'{\i}}n Abadi}, \bibinfo{person}{Luca
  Cardelli}, \bibinfo{person}{Pierre{-}Louis Curien}, {and}
  \bibinfo{person}{Jean{-}Jacques L{\'{e}}vy}.}
  \bibinfo{year}{1991}\natexlab{}.
\newblock \showarticletitle{Explicit Substitutions}.
\newblock \bibinfo{journal}{\emph{J. Funct. Program.}} \bibinfo{volume}{1},
  \bibinfo{number}{4} (\bibinfo{year}{1991}), \bibinfo{pages}{375--416}.
\newblock
\urldef\tempurl%
\url{https://doi.org/10.1017/S0956796800000186}
\showDOI{\tempurl}


\bibitem[\protect\citeauthoryear{Abel and Pientka}{Abel and Pientka}{2010}]%
        {DBLP:journals/corr/abs-1009-2789}
\bibfield{author}{\bibinfo{person}{Andreas Abel} {and}
  \bibinfo{person}{Brigitte Pientka}.} \bibinfo{year}{2010}\natexlab{}.
\newblock \showarticletitle{Explicit Substitutions for Contextual Type Theory}.
  In \bibinfo{booktitle}{\emph{International Workshop on Logical Frameworks and
  Meta-languages: Theory and Practice (LFMTP).}}
\newblock
\urldef\tempurl%
\url{https://doi.org/10.4204/EPTCS.34.3}
\showDOI{\tempurl}


\bibitem[\protect\citeauthoryear{Abramsky}{Abramsky}{1990}]%
        {Abramsky:1990vv}
\bibfield{author}{\bibinfo{person}{Samson Abramsky}.}
  \bibinfo{year}{1990}\natexlab{}.
\newblock \showarticletitle{The lazy lambda calculus}. In
  \bibinfo{booktitle}{\emph{Research Topics in Functional Programming}}.
  Addison-Wesley Longman Publishing Co., Inc., \bibinfo{pages}{65--116}.
\newblock
\urldef\tempurl%
\url{http://moscova.inria.fr/~levy/courses/X/M1/lambda/bib/90abramskylazy.pdf}
\showURL{%
\tempurl}


\bibitem[\protect\citeauthoryear{Afshari, Barr, and Su}{Afshari
  et~al\mbox{.}}{2012}]%
        {DBLP:conf/oopsla/AfshariBS12}
\bibfield{author}{\bibinfo{person}{Mehrdad Afshari}, \bibinfo{person}{Earl~T.
  Barr}, {and} \bibinfo{person}{Zhendong Su}.} \bibinfo{year}{2012}\natexlab{}.
\newblock \showarticletitle{Liberating the programmer with prorogued
  programming}. In \bibinfo{booktitle}{\emph{Symposium on New Ideas in
  Programming and Reflections on Software (Onward!)}}.
\newblock
\urldef\tempurl%
\url{https://doi.org/10.1145/2384592.2384595}
\showDOI{\tempurl}


\bibitem[\protect\citeauthoryear{Aho and Peterson}{Aho and Peterson}{1972}]%
        {DBLP:journals/siamcomp/AhoP72}
\bibfield{author}{\bibinfo{person}{Alfred~V. Aho} {and}
  \bibinfo{person}{Thomas~G. Peterson}.} \bibinfo{year}{1972}\natexlab{}.
\newblock \showarticletitle{A Minimum Distance Error-Correcting Parser for
  Context-Free Languages}.
\newblock \bibinfo{journal}{\emph{{SIAM} J. Comput.}} \bibinfo{volume}{1},
  \bibinfo{number}{4} (\bibinfo{year}{1972}), \bibinfo{pages}{305--312}.
\newblock
\urldef\tempurl%
\url{https://doi.org/10.1137/0201022}
\showDOI{\tempurl}


\bibitem[\protect\citeauthoryear{Amorim, Erdweg, Wachsmuth, and Visser}{Amorim
  et~al\mbox{.}}{2016}]%
        {Amorim2016}
\bibfield{author}{\bibinfo{person}{Lu\'{\i}s Eduardo de~Souza Amorim},
  \bibinfo{person}{Sebastian Erdweg}, \bibinfo{person}{Guido Wachsmuth}, {and}
  \bibinfo{person}{Eelco Visser}.} \bibinfo{year}{2016}\natexlab{}.
\newblock \showarticletitle{Principled {S}yntactic {C}ode {C}ompletion {U}sing
  {P}laceholders}. In \bibinfo{booktitle}{\emph{International Conference on
  Software Language Engineering (SLE)}}.
\newblock
\urldef\tempurl%
\url{http://doi.acm.org/10.1145/2997364.2997374}
\showURL{%
\tempurl}


\bibitem[\protect\citeauthoryear{Asenov and M{\"{u}}ller}{Asenov and
  M{\"{u}}ller}{2014}]%
        {DBLP:conf/vl/Asenov014}
\bibfield{author}{\bibinfo{person}{Dimitar Asenov} {and} \bibinfo{person}{Peter
  M{\"{u}}ller}.} \bibinfo{year}{2014}\natexlab{}.
\newblock \showarticletitle{Envision: {A} fast and flexible visual code editor
  with fluid interactions (Overview)}. In \bibinfo{booktitle}{\emph{{IEEE}
  Symposium on Visual Languages and Human-Centric Computing ({VL/HCC})}}.
\newblock
\urldef\tempurl%
\url{https://doi.org/10.1109/VLHCC.2014.6883014}
\showURL{%
\tempurl}


\bibitem[\protect\citeauthoryear{Awodey, Gambino, and Sojakova}{Awodey
  et~al\mbox{.}}{2012}]%
        {awodey2012inductive}
\bibfield{author}{\bibinfo{person}{Steve Awodey}, \bibinfo{person}{Nicola
  Gambino}, {and} \bibinfo{person}{Kristina Sojakova}.}
  \bibinfo{year}{2012}\natexlab{}.
\newblock \showarticletitle{Inductive types in homotopy type theory}. In
  \bibinfo{booktitle}{\emph{Symposium on Logic in Computer Science (LICS)}}.
\newblock
\urldef\tempurl%
\url{https://doi.org/10.1109/LICS.2012.21}
\showURL{%
\tempurl}


\bibitem[\protect\citeauthoryear{Aydemir, Bohannon, Fairbairn, Foster, Pierce,
  Sewell, Vytiniotis, Washburn, Weirich, and Zdancewic}{Aydemir
  et~al\mbox{.}}{2005}]%
        {Aydemir:2005fk}
\bibfield{author}{\bibinfo{person}{Brian~E Aydemir}, \bibinfo{person}{Aaron
  Bohannon}, \bibinfo{person}{Matthew Fairbairn}, \bibinfo{person}{J~Nathan
  Foster}, \bibinfo{person}{Benjamin~C Pierce}, \bibinfo{person}{Peter Sewell},
  \bibinfo{person}{Dimitrios Vytiniotis}, \bibinfo{person}{Geoffrey Washburn},
  \bibinfo{person}{Stephanie Weirich}, {and} \bibinfo{person}{Steve
  Zdancewic}.} \bibinfo{year}{2005}\natexlab{}.
\newblock \showarticletitle{Mechanized metatheory for the masses: the
  {POPLM}ark challenge}. In \bibinfo{booktitle}{\emph{International Conference
  on Theorem Proving in Higher Order Logics}}. Springer,
  \bibinfo{pages}{50--65}.
\newblock
\urldef\tempurl%
\url{https://doi.org/10.1007/11541868\_4}
\showURL{%
\tempurl}


\bibitem[\protect\citeauthoryear{Baldoni, Coppa, D'Elia, Demetrescu, and
  Finocchi}{Baldoni et~al\mbox{.}}{2018}]%
        {SurveySymExec-CSUR18}
\bibfield{author}{\bibinfo{person}{Roberto Baldoni}, \bibinfo{person}{Emilio
  Coppa}, \bibinfo{person}{Daniele~Cono D'Elia}, \bibinfo{person}{Camil
  Demetrescu}, {and} \bibinfo{person}{Irene Finocchi}.}
  \bibinfo{year}{2018}\natexlab{}.
\newblock \showarticletitle{A Survey of Symbolic Execution Techniques}.
\newblock \bibinfo{journal}{\emph{ACM Comput. Surv.}} \bibinfo{volume}{51},
  \bibinfo{number}{3}, Article \bibinfo{articleno}{50} (\bibinfo{year}{2018}).
\newblock
\urldef\tempurl%
\url{http://doi.acm.org/10.1145/3182657}
\showURL{%
\tempurl}


\bibitem[\protect\citeauthoryear{Barendregt}{Barendregt}{1984}]%
        {barendregt84:_lambda_calculus}
\bibfield{author}{\bibinfo{person}{H.P. Barendregt}.}
  \bibinfo{year}{1984}\natexlab{}.
\newblock \bibinfo{booktitle}{\emph{The Lambda Calculus}}.
  \bibinfo{series}{Studies in Logic}, Vol.~\bibinfo{volume}{103}.
\newblock \bibinfo{publisher}{Elsevier}.
\newblock


\bibitem[\protect\citeauthoryear{Bayne, Cook, and Ernst}{Bayne
  et~al\mbox{.}}{2011}]%
        {Bayne:2011:ASD:1985793.1985864}
\bibfield{author}{\bibinfo{person}{Michael Bayne}, \bibinfo{person}{Richard
  Cook}, {and} \bibinfo{person}{Michael~D. Ernst}.}
  \bibinfo{year}{2011}\natexlab{}.
\newblock \showarticletitle{Always-available Static and Dynamic Feedback}. In
  \bibinfo{booktitle}{\emph{International Conference on Software Engineering
  (ICSE)}}.
\newblock
\urldef\tempurl%
\url{https://doi.org/10.1145/1985793.1985864}
\showDOI{\tempurl}


\bibitem[\protect\citeauthoryear{Blanc, L{\'{e}}vy, and Maranget}{Blanc
  et~al\mbox{.}}{2005}]%
        {DBLP:conf/birthday/BlancLM05}
\bibfield{author}{\bibinfo{person}{Tomasz Blanc},
  \bibinfo{person}{Jean{-}Jacques L{\'{e}}vy}, {and} \bibinfo{person}{Luc
  Maranget}.} \bibinfo{year}{2005}\natexlab{}.
\newblock \showarticletitle{Sharing in the Weak Lambda-Calculus}. In
  \bibinfo{booktitle}{\emph{Processes, Terms and Cycles: Steps on the Road to
  Infinity, Essays Dedicated to Jan Willem Klop, on the Occasion of His 60th
  Birthday}} \emph{(\bibinfo{series}{Lecture Notes in Computer Science})},
  Vol.~\bibinfo{volume}{3838}. \bibinfo{publisher}{Springer},
  \bibinfo{pages}{70--87}.
\newblock
\urldef\tempurl%
\url{https://doi.org/10.1007/11601548_7}
\showDOI{\tempurl}


\bibitem[\protect\citeauthoryear{Bolingbroke and {Peyton Jones}}{Bolingbroke
  and {Peyton Jones}}{2010}]%
        {DBLP:conf/haskell/BolingbrokeJ10}
\bibfield{author}{\bibinfo{person}{Maximilian~C. Bolingbroke} {and}
  \bibinfo{person}{Simon~L. {Peyton Jones}}.} \bibinfo{year}{2010}\natexlab{}.
\newblock \showarticletitle{Supercompilation by evaluation}. In
  \bibinfo{booktitle}{\emph{Symposium on Haskell}}.
\newblock
\urldef\tempurl%
\url{https://doi.org/10.1145/1863523.1863540}
\showDOI{\tempurl}


\bibitem[\protect\citeauthoryear{Brady}{Brady}{2013}]%
        {brady2013idris}
\bibfield{author}{\bibinfo{person}{Edwin Brady}.}
  \bibinfo{year}{2013}\natexlab{}.
\newblock \showarticletitle{{Idris, A General-Purpose Dependently Typed
  Programming Language: Design and Implementation}}.
\newblock \bibinfo{journal}{\emph{Journal of Functional Programming}}
  \bibinfo{volume}{23}, \bibinfo{number}{05} (\bibinfo{year}{2013}),
  \bibinfo{pages}{552--593}.
\newblock
\urldef\tempurl%
\url{https://doi.org/10.1017/S095679681300018X}
\showURL{%
\tempurl}


\bibitem[\protect\citeauthoryear{Burckhardt, F{\"{a}}hndrich, de~Halleux,
  McDirmid, Moskal, Tillmann, and Kato}{Burckhardt et~al\mbox{.}}{2013}]%
        {burckhardt2013s}
\bibfield{author}{\bibinfo{person}{Sebastian Burckhardt},
  \bibinfo{person}{Manuel F{\"{a}}hndrich}, \bibinfo{person}{Peli de Halleux},
  \bibinfo{person}{Sean McDirmid}, \bibinfo{person}{Michal Moskal},
  \bibinfo{person}{Nikolai Tillmann}, {and} \bibinfo{person}{Jun Kato}.}
  \bibinfo{year}{2013}\natexlab{}.
\newblock \showarticletitle{{It's Alive! Continuous Feedback in UI
  Programming}}. In \bibinfo{booktitle}{\emph{Programming Language Design and
  Implementation (PLDI)}}.
\newblock
\urldef\tempurl%
\url{http://doi.acm.org/10.1145/2462156.2462170}
\showURL{%
\tempurl}


\bibitem[\protect\citeauthoryear{Burnett, Jr., and Welch}{Burnett
  et~al\mbox{.}}{1998}]%
        {DBLP:conf/vl/BurnettAW98}
\bibfield{author}{\bibinfo{person}{Margaret~M. Burnett}, \bibinfo{person}{John
  W.~Atwood Jr.}, {and} \bibinfo{person}{Zachary~T. Welch}.}
  \bibinfo{year}{1998}\natexlab{}.
\newblock \showarticletitle{Implementing Level 4 Liveness in Declarative Visual
  Programming Languages}. In \bibinfo{booktitle}{\emph{{IEEE} Symposium on
  Visual Languages}}.
\newblock
\urldef\tempurl%
\url{https://doi.org/10.1109/VL.1998.706155}
\showURL{%
\tempurl}


\bibitem[\protect\citeauthoryear{{\c{C}}agman and Hindley}{{\c{C}}agman and
  Hindley}{1998}]%
        {DBLP:journals/tcs/CagmanH98}
\bibfield{author}{\bibinfo{person}{Naim {\c{C}}agman} {and}
  \bibinfo{person}{J.~Roger Hindley}.} \bibinfo{year}{1998}\natexlab{}.
\newblock \showarticletitle{Combinatory Weak Reduction in Lambda Calculus}.
\newblock \bibinfo{journal}{\emph{Theor. Comput. Sci.}} \bibinfo{volume}{198},
  \bibinfo{number}{1-2} (\bibinfo{year}{1998}), \bibinfo{pages}{239--247}.
\newblock
\urldef\tempurl%
\url{https://doi.org/10.1016/S0304-3975(97)00250-8}
\showDOI{\tempurl}


\bibitem[\protect\citeauthoryear{Charles}{Charles}{1991}]%
        {charles1991practical}
\bibfield{author}{\bibinfo{person}{Philippe Charles}.}
  \bibinfo{year}{1991}\natexlab{}.
\newblock \emph{\bibinfo{title}{A Practical Method for Constructing Efficient
  LALR(k) Parsers with Automatic Error Recovery}}.
\newblock \bibinfo{thesistype}{Ph.D. Dissertation}. \bibinfo{address}{New York,
  NY, USA}.
\newblock


\bibitem[\protect\citeauthoryear{Chen and Erwig}{Chen and Erwig}{2014}]%
        {ChenErwig2014}
\bibfield{author}{\bibinfo{person}{Sheng Chen} {and} \bibinfo{person}{Martin
  Erwig}.} \bibinfo{year}{2014}\natexlab{}.
\newblock \showarticletitle{Counter-{F}actual {T}yping for {D}ebugging {T}ype
  {E}rrors}. In \bibinfo{booktitle}{\emph{Principles of Programming Languages
  (POPL)}}.
\newblock
\urldef\tempurl%
\url{https://doi.org/10.1145/2535838.2535863}
\showURL{%
\tempurl}


\bibitem[\protect\citeauthoryear{Chen and Erwig}{Chen and Erwig}{2018}]%
        {DBLP:journals/jfp/ChenE18}
\bibfield{author}{\bibinfo{person}{Sheng Chen} {and} \bibinfo{person}{Martin
  Erwig}.} \bibinfo{year}{2018}\natexlab{}.
\newblock \showarticletitle{Systematic identification and communication of type
  errors}.
\newblock \bibinfo{journal}{\emph{J. Funct. Program.}}  \bibinfo{volume}{28}
  (\bibinfo{year}{2018}).
\newblock
\urldef\tempurl%
\url{https://doi.org/10.1017/S095679681700020X}
\showURL{%
\tempurl}


\bibitem[\protect\citeauthoryear{Chlipala, Petersen, and Harper}{Chlipala
  et~al\mbox{.}}{2005}]%
        {Chlipala:2005da}
\bibfield{author}{\bibinfo{person}{Adam Chlipala}, \bibinfo{person}{Leaf
  Petersen}, {and} \bibinfo{person}{Robert Harper}.}
  \bibinfo{year}{2005}\natexlab{}.
\newblock \showarticletitle{Strict bidirectional type checking}. In
  \bibinfo{booktitle}{\emph{International Workshop on Types in Language Design
  and Implementation (TLDI)}}.
\newblock
\urldef\tempurl%
\url{http://doi.acm.org/10.1145/1040294.1040301}
\showURL{%
\tempurl}


\bibitem[\protect\citeauthoryear{Christiansen}{Christiansen}{2013}]%
        {bidi-tutorial}
\bibfield{author}{\bibinfo{person}{David~Raymond Christiansen}.}
  \bibinfo{year}{2013}\natexlab{}.
\newblock \bibinfo{title}{{Bidirectional Typing Rules: A Tutorial}}.
\newblock
  \bibinfo{howpublished}{\url{http://davidchristiansen.dk/tutorials/bidirectional.pdf}}.
\newblock


\bibitem[\protect\citeauthoryear{Chugh, Hempel, Spradlin, and Albers}{Chugh
  et~al\mbox{.}}{2016}]%
        {sns-pldi}
\bibfield{author}{\bibinfo{person}{Ravi Chugh}, \bibinfo{person}{Brian Hempel},
  \bibinfo{person}{Mitchell Spradlin}, {and} \bibinfo{person}{Jacob Albers}.}
  \bibinfo{year}{2016}\natexlab{}.
\newblock \showarticletitle{Programmatic and {D}irect {M}anipulation,
  {T}ogether at {L}ast}. In \bibinfo{booktitle}{\emph{Programming Language
  Design and Implementation (PLDI)}}.
\newblock
\urldef\tempurl%
\url{https://doi.org/10.1145/2908080.2908103}
\showURL{%
\tempurl}


\bibitem[\protect\citeauthoryear{Church and Rosser}{Church and Rosser}{1936}]%
        {church1936some}
\bibfield{author}{\bibinfo{person}{Alonzo Church} {and}
  \bibinfo{person}{J~Barkley Rosser}.} \bibinfo{year}{1936}\natexlab{}.
\newblock \showarticletitle{Some properties of conversion}.
\newblock \bibinfo{journal}{\emph{Trans. Amer. Math. Soc.}}
  \bibinfo{volume}{39}, \bibinfo{number}{3} (\bibinfo{year}{1936}),
  \bibinfo{pages}{472--482}.
\newblock


\bibitem[\protect\citeauthoryear{Cimini and Siek}{Cimini and Siek}{2016}]%
        {DBLP:conf/popl/CiminiS16}
\bibfield{author}{\bibinfo{person}{Matteo Cimini} {and}
  \bibinfo{person}{Jeremy~G. Siek}.} \bibinfo{year}{2016}\natexlab{}.
\newblock \showarticletitle{The gradualizer: a methodology and algorithm for
  generating gradual type systems}. In \bibinfo{booktitle}{\emph{Principles of
  Programming Languages (POPL)}}.
\newblock
\urldef\tempurl%
\url{http://doi.acm.org/10.1145/2837614.2837632}
\showURL{%
\tempurl}


\bibitem[\protect\citeauthoryear{Cimini and Siek}{Cimini and Siek}{2017}]%
        {DBLP:conf/popl/CiminiS17}
\bibfield{author}{\bibinfo{person}{Matteo Cimini} {and}
  \bibinfo{person}{Jeremy~G. Siek}.} \bibinfo{year}{2017}\natexlab{}.
\newblock \showarticletitle{Automatically generating the dynamic semantics of
  gradually typed languages}. In \bibinfo{booktitle}{\emph{Proceedings of the
  44th {ACM} {SIGPLAN} Symposium on Principles of Programming Languages, {POPL}
  2017, Paris, France, January 18-20, 2017}},
  \bibfield{editor}{\bibinfo{person}{Giuseppe Castagna} {and}
  \bibinfo{person}{Andrew~D. Gordon}} (Eds.). \bibinfo{publisher}{{ACM}},
  \bibinfo{pages}{789--803}.
\newblock
\urldef\tempurl%
\url{http://dl.acm.org/citation.cfm?id=3009863}
\showURL{%
\tempurl}


\bibitem[\protect\citeauthoryear{Curien}{Curien}{1991}]%
        {DBLP:journals/tcs/Curien91}
\bibfield{author}{\bibinfo{person}{Pierre{-}Louis Curien}.}
  \bibinfo{year}{1991}\natexlab{}.
\newblock \showarticletitle{An Abstract Framework for Environment Machines}.
\newblock \bibinfo{journal}{\emph{Theor. Comput. Sci.}} \bibinfo{volume}{82},
  \bibinfo{number}{2} (\bibinfo{year}{1991}), \bibinfo{pages}{389--402}.
\newblock
\urldef\tempurl%
\url{https://doi.org/10.1016/0304-3975(91)90230-Y}
\showDOI{\tempurl}


\bibitem[\protect\citeauthoryear{Czaplicki}{Czaplicki}{2012}]%
        {czaplicki2012elm}
\bibfield{author}{\bibinfo{person}{Evan Czaplicki}.}
  \bibinfo{year}{2012}\natexlab{}.
\newblock \showarticletitle{Elm: Concurrent {FRP} for Functional {GUIs}}.
\newblock \bibinfo{journal}{\emph{Senior thesis, Harvard University}}
  (\bibinfo{year}{2012}).
\newblock


\bibitem[\protect\citeauthoryear{Czaplicki}{Czaplicki}{2018}]%
        {Elm}
\bibfield{author}{\bibinfo{person}{Evan Czaplicki}.}
  \bibinfo{year}{2018}\natexlab{}.
\newblock \bibinfo{title}{An Introduction to Elm}.  (\bibinfo{year}{2018}).
\newblock
\newblock
\shownote{\url{https://guide.elm-lang.org/}. Retrieved Apr. 7, 2018.}


\bibitem[\protect\citeauthoryear{D'Alves, Bouman, Schankula, Hogg, Noronha,
  Horsman, Siddiqui, and Anand}{D'Alves et~al\mbox{.}}{2017}]%
        {DBLP:journals/corr/abs-1805-05125}
\bibfield{author}{\bibinfo{person}{Curtis D'Alves}, \bibinfo{person}{Tanya
  Bouman}, \bibinfo{person}{Christopher Schankula}, \bibinfo{person}{Jenell
  Hogg}, \bibinfo{person}{Levin Noronha}, \bibinfo{person}{Emily Horsman},
  \bibinfo{person}{Rumsha Siddiqui}, {and} \bibinfo{person}{Christopher~Kumar
  Anand}.} \bibinfo{year}{2017}\natexlab{}.
\newblock \showarticletitle{Using Elm to Introduce Algebraic Thinking to {K-8}
  Students}. In \bibinfo{booktitle}{\emph{International Workshop on Trends in
  Functional Programming in Education (TFPIE)}}.
\newblock
\urldef\tempurl%
\url{https://doi.org/10.4204/EPTCS.270.2}
\showURL{%
\tempurl}


\bibitem[\protect\citeauthoryear{Damas and Milner}{Damas and Milner}{1982}]%
        {damas1982principal}
\bibfield{author}{\bibinfo{person}{Luis Damas} {and} \bibinfo{person}{Robin
  Milner}.} \bibinfo{year}{1982}\natexlab{}.
\newblock \showarticletitle{{Principal type-schemes for functional programs}}.
  In \bibinfo{booktitle}{\emph{Principles of Programming Languages (POPL)}}.
\newblock
\urldef\tempurl%
\url{https://doi.org/10.1145/582153.582176}
\showURL{%
\tempurl}


\bibitem[\protect\citeauthoryear{Davies and Pfenning}{Davies and
  Pfenning}{2001}]%
        {Davies:2001op}
\bibfield{author}{\bibinfo{person}{Rowan Davies} {and} \bibinfo{person}{Frank
  Pfenning}.} \bibinfo{year}{2001}\natexlab{}.
\newblock \showarticletitle{A modal analysis of staged computation}.
\newblock \bibinfo{journal}{\emph{J. ACM}} \bibinfo{volume}{48},
  \bibinfo{number}{3} (\bibinfo{year}{2001}), \bibinfo{pages}{555--604}.
\newblock
\showISSN{0004-5411}
\urldef\tempurl%
\url{https://doi.org/10.1145/382780.382785}
\showURL{%
\tempurl}


\bibitem[\protect\citeauthoryear{de~Groote}{de~Groote}{2002}]%
        {DBLP:journals/iandc/Groote02}
\bibfield{author}{\bibinfo{person}{Philippe de Groote}.}
  \bibinfo{year}{2002}\natexlab{}.
\newblock \showarticletitle{On the Strong Normalisation of Intuitionistic
  Natural Deduction with Permutation-Conversions}.
\newblock \bibinfo{journal}{\emph{Inf. Comput.}} \bibinfo{volume}{178},
  \bibinfo{number}{2} (\bibinfo{year}{2002}), \bibinfo{pages}{441--464}.
\newblock
\urldef\tempurl%
\url{https://doi.org/10.1006/inco.2002.3147}
\showDOI{\tempurl}


\bibitem[\protect\citeauthoryear{Devriese, Patrignani, and Piessens}{Devriese
  et~al\mbox{.}}{2018}]%
        {DBLP:journals/pacmpl/DevriesePP18}
\bibfield{author}{\bibinfo{person}{Dominique Devriese}, \bibinfo{person}{Marco
  Patrignani}, {and} \bibinfo{person}{Frank Piessens}.}
  \bibinfo{year}{2018}\natexlab{}.
\newblock \showarticletitle{Parametricity versus the universal type}.
\newblock \bibinfo{journal}{\emph{{PACMPL}}} \bibinfo{volume}{2},
  \bibinfo{number}{{POPL}} (\bibinfo{year}{2018}),
  \bibinfo{pages}{38:1--38:23}.
\newblock
\urldef\tempurl%
\url{https://doi.org/10.1145/3158126}
\showURL{%
\tempurl}


\bibitem[\protect\citeauthoryear{Dunfield and Krishnaswami}{Dunfield and
  Krishnaswami}{2013}]%
        {DBLP:conf/icfp/DunfieldK13}
\bibfield{author}{\bibinfo{person}{Joshua Dunfield} {and}
  \bibinfo{person}{Neelakantan~R. Krishnaswami}.}
  \bibinfo{year}{2013}\natexlab{}.
\newblock \showarticletitle{Complete and easy bidirectional typechecking for
  higher-rank polymorphism}. In \bibinfo{booktitle}{\emph{International
  Conference on Functional Programming (ICFP)}}.
\newblock
\urldef\tempurl%
\url{https://doi.org/10.1145/2500365.2500582}
\showDOI{\tempurl}


\bibitem[\protect\citeauthoryear{Felleisen and Hieb}{Felleisen and
  Hieb}{1992}]%
        {DBLP:journals/tcs/FelleisenH92}
\bibfield{author}{\bibinfo{person}{Matthias Felleisen} {and}
  \bibinfo{person}{Robert Hieb}.} \bibinfo{year}{1992}\natexlab{}.
\newblock \showarticletitle{The Revised Report on the Syntactic Theories of
  Sequential Control and State}.
\newblock \bibinfo{journal}{\emph{Theor. Comput. Sci.}} \bibinfo{volume}{103},
  \bibinfo{number}{2} (\bibinfo{year}{1992}), \bibinfo{pages}{235--271}.
\newblock
\urldef\tempurl%
\url{https://doi.org/10.1016/0304-3975(92)90014-7}
\showDOI{\tempurl}


\bibitem[\protect\citeauthoryear{Ferreira and Pientka}{Ferreira and
  Pientka}{2014}]%
        {DBLP:conf/ppdp/FerreiraP14}
\bibfield{author}{\bibinfo{person}{Francisco Ferreira} {and}
  \bibinfo{person}{Brigitte Pientka}.} \bibinfo{year}{2014}\natexlab{}.
\newblock \showarticletitle{Bidirectional Elaboration of Dependently Typed
  Programs}. In \bibinfo{booktitle}{\emph{Symposium on Principles and Practice
  of Declarative Programming (PPDP)}}.
\newblock
\urldef\tempurl%
\url{https://doi.org/10.1145/2643135.2643153}
\showDOI{\tempurl}


\bibitem[\protect\citeauthoryear{Fitzgerald, Lewandowski, McCauley, Murphy,
  Simon, Thomas, and Zander}{Fitzgerald et~al\mbox{.}}{2008}]%
        {fitzgerald2008debugging}
\bibfield{author}{\bibinfo{person}{Sue Fitzgerald}, \bibinfo{person}{Gary
  Lewandowski}, \bibinfo{person}{Renee McCauley}, \bibinfo{person}{Laurie
  Murphy}, \bibinfo{person}{Beth Simon}, \bibinfo{person}{Lynda Thomas}, {and}
  \bibinfo{person}{Carol Zander}.} \bibinfo{year}{2008}\natexlab{}.
\newblock \showarticletitle{Debugging: finding, fixing and flailing, a
  multi-institutional study of novice debuggers}.
\newblock \bibinfo{journal}{\emph{Computer Science Education}}
  \bibinfo{volume}{18}, \bibinfo{number}{2} (\bibinfo{year}{2008}),
  \bibinfo{pages}{93--116}.
\newblock
\urldef\tempurl%
\url{https://doi.org/10.1080/08993400802114508}
\showURL{%
\tempurl}


\bibitem[\protect\citeauthoryear{{Flutter Developers}}{{Flutter
  Developers}}{2017}]%
        {flutter}
\bibfield{author}{\bibinfo{person}{{Flutter Developers}}.}
  \bibinfo{year}{2017}\natexlab{}.
\newblock \bibinfo{title}{{Technical Overview - Flutter}}.
\newblock
\newblock
\newblock
\shownote{\url{https://flutter.io/technical-overview/}. Retrieved Sep. 21,
  2017.}


\bibitem[\protect\citeauthoryear{Frankle, Osera, Walker, and Zdancewic}{Frankle
  et~al\mbox{.}}{2016}]%
        {DBLP:conf/popl/FrankleOWZ16}
\bibfield{author}{\bibinfo{person}{Jonathan Frankle},
  \bibinfo{person}{Peter{-}Michael Osera}, \bibinfo{person}{David Walker},
  {and} \bibinfo{person}{Steve Zdancewic}.} \bibinfo{year}{2016}\natexlab{}.
\newblock \showarticletitle{Example-directed synthesis: a type-theoretic
  interpretation}. In \bibinfo{booktitle}{\emph{Principles of Programming
  Languages (POPL)}}.
\newblock
\urldef\tempurl%
\url{https://doi.org/10.1145/2837614.2837629}
\showDOI{\tempurl}


\bibitem[\protect\citeauthoryear{Garcia}{Garcia}{2013}]%
        {garcia2013threesomes}
\bibfield{author}{\bibinfo{person}{Ronald Garcia}.}
  \bibinfo{year}{2013}\natexlab{}.
\newblock \showarticletitle{Calculating {T}hreesomes, {W}ith {B}lame}. In
  \bibinfo{booktitle}{\emph{International Conference on Functional Programming
  (ICFP)}}.
\newblock
\showISBNx{978-1-4503-2326-0}
\urldef\tempurl%
\url{https://doi.org/10.1145/2500365.2500603}
\showDOI{\tempurl}


\bibitem[\protect\citeauthoryear{Garcia and Cimini}{Garcia and Cimini}{2015}]%
        {DBLP:conf/popl/GarciaC15}
\bibfield{author}{\bibinfo{person}{Ronald Garcia} {and} \bibinfo{person}{Matteo
  Cimini}.} \bibinfo{year}{2015}\natexlab{}.
\newblock \showarticletitle{{Principal Type Schemes for Gradual Programs}}. In
  \bibinfo{booktitle}{\emph{Principles of Programming Languages (POPL)}}.
\newblock
\urldef\tempurl%
\url{http://doi.acm.org/10.1145/2676726.2676992}
\showURL{%
\tempurl}


\bibitem[\protect\citeauthoryear{Geuvers and Jojgov}{Geuvers and
  Jojgov}{2002}]%
        {DBLP:conf/csl/GeuversJ02}
\bibfield{author}{\bibinfo{person}{Herman Geuvers} {and}
  \bibinfo{person}{Gueorgui~I. Jojgov}.} \bibinfo{year}{2002}\natexlab{}.
\newblock \showarticletitle{Open Proofs and Open Terms: {A} Basis for
  Interactive Logic}. In \bibinfo{booktitle}{\emph{International Workshop on
  Computer Science Logic (CSL)}}.
\newblock
\urldef\tempurl%
\url{https://doi.org/10.1007/3-540-45793-3_36}
\showDOI{\tempurl}


\bibitem[\protect\citeauthoryear{Goldberg and Robson}{Goldberg and
  Robson}{1983}]%
        {Goldberg:1983cn}
\bibfield{author}{\bibinfo{person}{Adele Goldberg} {and} \bibinfo{person}{David
  Robson}.} \bibinfo{year}{1983}\natexlab{}.
\newblock \bibinfo{booktitle}{\emph{Smalltalk-80: the language and its
  implementation}}.
\newblock \bibinfo{publisher}{Addison-Wesley Longman Publishing Co., Inc.}
\newblock


\bibitem[\protect\citeauthoryear{Graham, Haley, and Joy}{Graham
  et~al\mbox{.}}{1979}]%
        {graham1979practical}
\bibfield{author}{\bibinfo{person}{Susan~L. Graham},
  \bibinfo{person}{Charles~B. Haley}, {and} \bibinfo{person}{William~N. Joy}.}
  \bibinfo{year}{1979}\natexlab{}.
\newblock \showarticletitle{Practical {L}{R} {E}rror {R}ecovery}. In
  \bibinfo{booktitle}{\emph{Symposium on Compiler Construction (CC)}}.
\newblock
\urldef\tempurl%
\url{https://doi.org/10.1145/800229.806967}
\showURL{%
\tempurl}


\bibitem[\protect\citeauthoryear{Guo}{Guo}{2013}]%
        {Guo13}
\bibfield{author}{\bibinfo{person}{Philip~J. Guo}.}
  \bibinfo{year}{2013}\natexlab{}.
\newblock \showarticletitle{Online {P}ython tutor: embeddable web-based program
  visualization for {CS} education}. In \bibinfo{booktitle}{\emph{Technical
  Symposium on Computer Science Education (SIGCSE)}}.
\newblock
\urldef\tempurl%
\url{https://doi.org/10.1145/2445196.2445368}
\showURL{%
\tempurl}


\bibitem[\protect\citeauthoryear{Hammer, Khoo, Hicks, and Foster}{Hammer
  et~al\mbox{.}}{2014}]%
        {Hammer2014}
\bibfield{author}{\bibinfo{person}{Matthew~A. Hammer},
  \bibinfo{person}{Yit~Phang Khoo}, \bibinfo{person}{Michael Hicks}, {and}
  \bibinfo{person}{Jeffrey~S. Foster}.} \bibinfo{year}{2014}\natexlab{}.
\newblock \showarticletitle{Adapton: composable, demand-driven incremental
  computation}. In \bibinfo{booktitle}{\emph{Programming Language Design and
  Implementation (PLDI)}}.
\newblock
\urldef\tempurl%
\url{https://doi.org/10.1145/2594291.2594324}
\showURL{%
\tempurl}


\bibitem[\protect\citeauthoryear{Harper}{Harper}{2016}]%
        {pfpl}
\bibfield{author}{\bibinfo{person}{Robert Harper}.}
  \bibinfo{year}{2016}\natexlab{}.
\newblock \bibinfo{booktitle}{\emph{{Practical Foundations for Programming
  Languages}} (\bibinfo{edition}{2nd} ed.)}.
\newblock \bibinfo{publisher}{Cambridge University Press}.
\newblock


\bibitem[\protect\citeauthoryear{Harper and Stone}{Harper and Stone}{2000}]%
        {Harper00atype-theoretic}
\bibfield{author}{\bibinfo{person}{Robert Harper} {and}
  \bibinfo{person}{Christopher Stone}.} \bibinfo{year}{2000}\natexlab{}.
\newblock \showarticletitle{{A Type-Theoretic Interpretation of Standard ML}}.
  In \bibinfo{booktitle}{\emph{Proof, Language and Interaction: Essays in
  Honour of Robin Milner}}. \bibinfo{publisher}{MIT Press}.
\newblock


\bibitem[\protect\citeauthoryear{Hayden, Magill, Hicks, Foster, and
  Foster}{Hayden et~al\mbox{.}}{2012}]%
        {DBLP:conf/vstte/HaydenMHFF12}
\bibfield{author}{\bibinfo{person}{Christopher~M. Hayden},
  \bibinfo{person}{Stephen Magill}, \bibinfo{person}{Michael Hicks},
  \bibinfo{person}{Nate Foster}, {and} \bibinfo{person}{Jeffrey~S. Foster}.}
  \bibinfo{year}{2012}\natexlab{}.
\newblock \showarticletitle{Specifying and Verifying the Correctness of Dynamic
  Software Updates}. In \bibinfo{booktitle}{\emph{International Conference on
  Verified Software: Theories, Tools, Experiments (VSTTE)}}.
\newblock
\urldef\tempurl%
\url{https://doi.org/10.1007/978-3-642-27705-4_22}
\showDOI{\tempurl}


\bibitem[\protect\citeauthoryear{Hempel and Chugh}{Hempel and Chugh}{2016}]%
        {sns-uist}
\bibfield{author}{\bibinfo{person}{Brian Hempel} {and} \bibinfo{person}{Ravi
  Chugh}.} \bibinfo{year}{2016}\natexlab{}.
\newblock \showarticletitle{Semi-{A}utomated {S}{V}{G} {P}rogramming via
  {D}irect {M}anipulation}. In \bibinfo{booktitle}{\emph{Symposium on User
  Interface Software and Technology (UIST)}}.
\newblock
\urldef\tempurl%
\url{https://doi.org/10.1145/2984511.2984575}
\showURL{%
\tempurl}


\bibitem[\protect\citeauthoryear{Herman, Tomb, and Flanagan}{Herman
  et~al\mbox{.}}{2010}]%
        {herman2010space}
\bibfield{author}{\bibinfo{person}{David Herman}, \bibinfo{person}{Aaron Tomb},
  {and} \bibinfo{person}{Cormac Flanagan}.} \bibinfo{year}{2010}\natexlab{}.
\newblock \showarticletitle{Space-efficient gradual typing}.
\newblock \bibinfo{journal}{\emph{Higher-Order and Symbolic Computation}}
  \bibinfo{volume}{23}, \bibinfo{number}{2} (\bibinfo{year}{2010}),
  \bibinfo{pages}{167}.
\newblock
\urldef\tempurl%
\url{https://doi.org/10.1007/s10990-011-9066-z}
\showURL{%
\tempurl}


\bibitem[\protect\citeauthoryear{Hicks and Nettles}{Hicks and Nettles}{2005}]%
        {DBLP:journals/toplas/HicksN05}
\bibfield{author}{\bibinfo{person}{Michael~W. Hicks} {and}
  \bibinfo{person}{Scott Nettles}.} \bibinfo{year}{2005}\natexlab{}.
\newblock \showarticletitle{Dynamic software updating}.
\newblock \bibinfo{journal}{\emph{{ACM} Trans. Program. Lang. Syst.}}
  \bibinfo{volume}{27}, \bibinfo{number}{6} (\bibinfo{year}{2005}),
  \bibinfo{pages}{1049--1096}.
\newblock
\urldef\tempurl%
\url{https://doi.org/10.1145/1108970.1108971}
\showDOI{\tempurl}


\bibitem[\protect\citeauthoryear{Hritcu, Greenberg, Karel, Pierce, and
  Morrisett}{Hritcu et~al\mbox{.}}{2013}]%
        {DBLP:conf/sp/HritcuGKPM13}
\bibfield{author}{\bibinfo{person}{Catalin Hritcu}, \bibinfo{person}{Michael
  Greenberg}, \bibinfo{person}{Ben Karel}, \bibinfo{person}{Benjamin~C.
  Pierce}, {and} \bibinfo{person}{Greg Morrisett}.}
  \bibinfo{year}{2013}\natexlab{}.
\newblock \showarticletitle{All Your IFCException Are Belong to Us}. In
  \bibinfo{booktitle}{\emph{{IEEE} Symposium on Security and Privacy}}.
  \bibinfo{publisher}{{IEEE} Computer Society}, \bibinfo{pages}{3--17}.
\newblock
\urldef\tempurl%
\url{https://doi.org/10.1109/SP.2013.10}
\showURL{%
\tempurl}


\bibitem[\protect\citeauthoryear{Huet}{Huet}{1980}]%
        {Huet:1980ng}
\bibfield{author}{\bibinfo{person}{G\'{e}rard Huet}.}
  \bibinfo{year}{1980}\natexlab{}.
\newblock \showarticletitle{Confluent Reductions: Abstract Properties and
  Applications to Term Rewriting Systems: Abstract Properties and Applications
  to Term Rewriting Systems}.
\newblock \bibinfo{journal}{\emph{J. ACM}} \bibinfo{volume}{27},
  \bibinfo{number}{4} (\bibinfo{year}{1980}), \bibinfo{pages}{797--821}.
\newblock
\showISSN{0004-5411}


\bibitem[\protect\citeauthoryear{Igarashi, Sekiyama, and Igarashi}{Igarashi
  et~al\mbox{.}}{2017}]%
        {Igarashi:2017:PGT:3136534.3110284}
\bibfield{author}{\bibinfo{person}{Yuu Igarashi}, \bibinfo{person}{Taro
  Sekiyama}, {and} \bibinfo{person}{Atsushi Igarashi}.}
  \bibinfo{year}{2017}\natexlab{}.
\newblock \showarticletitle{On Polymorphic Gradual Typing}.
\newblock \bibinfo{journal}{\emph{Proc. ACM Program. Lang.}}
  \bibinfo{volume}{1}, \bibinfo{number}{ICFP}, Article \bibinfo{articleno}{40}
  (\bibinfo{date}{Aug.} \bibinfo{year}{2017}), \bibinfo{numpages}{29}~pages.
\newblock
\showISSN{2475-1421}
\urldef\tempurl%
\url{https://doi.org/10.1145/3110284}
\showDOI{\tempurl}


\bibitem[\protect\citeauthoryear{Jones}{Jones}{2017}]%
        {VSEditAndContinue}
\bibfield{author}{\bibinfo{person}{Mike Jones}.}
  \bibinfo{year}{2017}\natexlab{}.
\newblock \bibinfo{title}{{Edit Code and Continue Debugging in Visual Studio
  (C\#, VB, C++)}}.
\newblock
\newblock
\newblock
\shownote{\url{https://docs.microsoft.com/en-us/visualstudio/debugger/edit-and-continue}.
  Retrieved Apr. 27, 2018.}


\bibitem[\protect\citeauthoryear{Jones, Gomard, and Sestoft}{Jones
  et~al\mbox{.}}{1993}]%
        {Jones:1993uq}
\bibfield{author}{\bibinfo{person}{Neil~D. Jones}, \bibinfo{person}{Carsten~K.
  Gomard}, {and} \bibinfo{person}{Peter Sestoft}.}
  \bibinfo{year}{1993}\natexlab{}.
\newblock \bibinfo{booktitle}{\emph{Partial evaluation and automatic program
  generation}}.
\newblock \bibinfo{publisher}{Prentice-Hall, Inc.}
\newblock


\bibitem[\protect\citeauthoryear{Kahn and MacQueen}{Kahn and MacQueen}{1977}]%
        {DBLP:conf/ifip/KahnM77}
\bibfield{author}{\bibinfo{person}{Gilles Kahn} {and} \bibinfo{person}{David~B.
  MacQueen}.} \bibinfo{year}{1977}\natexlab{}.
\newblock \showarticletitle{Coroutines and Networks of Parallel Processes}. In
  \bibinfo{booktitle}{\emph{{IFIP} Congress}}. \bibinfo{pages}{993--998}.
\newblock


\bibitem[\protect\citeauthoryear{Kats, de~Jonge, Nilsson{-}Nyman, and
  Visser}{Kats et~al\mbox{.}}{2009}]%
        {DBLP:conf/oopsla/KatsJNV09}
\bibfield{author}{\bibinfo{person}{Lennart C.~L. Kats},
  \bibinfo{person}{Maartje de Jonge}, \bibinfo{person}{Emma Nilsson{-}Nyman},
  {and} \bibinfo{person}{Eelco Visser}.} \bibinfo{year}{2009}\natexlab{}.
\newblock \showarticletitle{Providing rapid feedback in generated modular
  language environments: adding error recovery to scannerless generalized-LR
  parsing}. In \bibinfo{booktitle}{\emph{Conference on Object-Oriented
  Programming, Systems, Languages, and Applications (OOPSLA)}}.
\newblock
\urldef\tempurl%
\url{https://doi.org/10.1145/1640089.1640122}
\showURL{%
\tempurl}


\bibitem[\protect\citeauthoryear{King}{King}{1976}]%
        {King:1976}
\bibfield{author}{\bibinfo{person}{James~C. King}.}
  \bibinfo{year}{1976}\natexlab{}.
\newblock \showarticletitle{Symbolic Execution and Program Testing}.
\newblock \bibinfo{journal}{\emph{Commun. ACM}} \bibinfo{volume}{19},
  \bibinfo{number}{7} (\bibinfo{date}{July} \bibinfo{year}{1976}),
  \bibinfo{pages}{385--394}.
\newblock
\showISSN{0001-0782}
\urldef\tempurl%
\url{https://doi.org/10.1145/360248.360252}
\showDOI{\tempurl}


\bibitem[\protect\citeauthoryear{Korkut and Christiansen}{Korkut and
  Christiansen}{2018}]%
        {Korkut:2018:ETE:3240719.3241791}
\bibfield{author}{\bibinfo{person}{Joomy Korkut} {and}
  \bibinfo{person}{David~Thrane Christiansen}.}
  \bibinfo{year}{2018}\natexlab{}.
\newblock \showarticletitle{{Extensible Type-Directed Editing}}. In
  \bibinfo{booktitle}{\emph{International Workshop on Type-Driven Development
  (TyDe)}}.
\newblock
\showISBNx{978-1-4503-5825-5}
\urldef\tempurl%
\url{https://doi.org/10.1145/3240719.3241791}
\showDOI{\tempurl}


\bibitem[\protect\citeauthoryear{Lehmann and Tanter}{Lehmann and
  Tanter}{2017}]%
        {DBLP:conf/popl/LehmannT17}
\bibfield{author}{\bibinfo{person}{Nico Lehmann} {and}
  \bibinfo{person}{{\'{E}}ric Tanter}.} \bibinfo{year}{2017}\natexlab{}.
\newblock \showarticletitle{Gradual refinement types}. In
  \bibinfo{booktitle}{\emph{Principles of Programming Languages (POPL)}}.
\newblock
\urldef\tempurl%
\url{http://dl.acm.org/citation.cfm?id=3009856}
\showURL{%
\tempurl}


\bibitem[\protect\citeauthoryear{Lerner, Grossman, and Chambers}{Lerner
  et~al\mbox{.}}{2006}]%
        {Seminal}
\bibfield{author}{\bibinfo{person}{Benjamin Lerner}, \bibinfo{person}{Dan
  Grossman}, {and} \bibinfo{person}{Craig Chambers}.}
  \bibinfo{year}{2006}\natexlab{}.
\newblock \showarticletitle{Seminal: Searching for {M}{L} {T}ype-error
  {M}essages}. In \bibinfo{booktitle}{\emph{Workshop on ML}}.
\newblock
\urldef\tempurl%
\url{https://doi.org/10.1145/1159876.1159887}
\showURL{%
\tempurl}


\bibitem[\protect\citeauthoryear{L{\'e}vy and Maranget}{L{\'e}vy and
  Maranget}{1999}]%
        {levy1999explicit}
\bibfield{author}{\bibinfo{person}{Jean-Jacques L{\'e}vy} {and}
  \bibinfo{person}{Luc Maranget}.} \bibinfo{year}{1999}\natexlab{}.
\newblock \showarticletitle{Explicit substitutions and programming languages}.
  In \bibinfo{booktitle}{\emph{International Conference on Foundations of
  Software Technology and Theoretical Computer Science}}. Springer,
  \bibinfo{pages}{181--200}.
\newblock
\urldef\tempurl%
\url{https://doi.org/10.1007/3-540-46691-6\_14}
\showURL{%
\tempurl}


\bibitem[\protect\citeauthoryear{Lotem and Chuchem}{Lotem and Chuchem}{2016}]%
        {lamdu}
\bibfield{author}{\bibinfo{person}{Eyal Lotem} {and} \bibinfo{person}{Yair
  Chuchem}.} \bibinfo{year}{2016}\natexlab{}.
\newblock \bibinfo{title}{{Project Lamdu}}.
\newblock \bibinfo{howpublished}{\url{http://www.lamdu.org/}}.
\newblock
\newblock
\shownote{Accessed: 2016-04-08.}


\bibitem[\protect\citeauthoryear{Magnusson}{Magnusson}{1995}]%
        {magnusson1994implementation}
\bibfield{author}{\bibinfo{person}{Lena Magnusson}.}
  \bibinfo{year}{1995}\natexlab{}.
\newblock \emph{\bibinfo{title}{The implementation of ALF: a proof editor based
  on Martin-L{\"o}f's monomorphic type theory with explicit substitution}}.
\newblock \bibinfo{thesistype}{Ph.D. Dissertation}. \bibinfo{school}{Chalmers
  Institute of Technology}.
\newblock


\bibitem[\protect\citeauthoryear{McBride}{McBride}{2000}]%
        {DBLP:phd/ethos/McBride00}
\bibfield{author}{\bibinfo{person}{Conor McBride}.}
  \bibinfo{year}{2000}\natexlab{}.
\newblock \emph{\bibinfo{title}{Dependently typed functional programs and their
  proofs}}.
\newblock \bibinfo{thesistype}{Ph.D. Dissertation}. \bibinfo{school}{University
  of Edinburgh, {UK}}.
\newblock
\urldef\tempurl%
\url{http://hdl.handle.net/1842/374}
\showURL{%
\tempurl}


\bibitem[\protect\citeauthoryear{McCauley, Fitzgerald, Lewandowski, Murphy,
  Simon, Thomas, and Zander}{McCauley et~al\mbox{.}}{2008}]%
        {mccauley2008debugging}
\bibfield{author}{\bibinfo{person}{Renee McCauley}, \bibinfo{person}{Sue
  Fitzgerald}, \bibinfo{person}{Gary Lewandowski}, \bibinfo{person}{Laurie
  Murphy}, \bibinfo{person}{Beth Simon}, \bibinfo{person}{Lynda Thomas}, {and}
  \bibinfo{person}{Carol Zander}.} \bibinfo{year}{2008}\natexlab{}.
\newblock \showarticletitle{Debugging: a review of the literature from an
  educational perspective}.
\newblock \bibinfo{journal}{\emph{Computer Science Education}}
  \bibinfo{volume}{18}, \bibinfo{number}{2} (\bibinfo{year}{2008}),
  \bibinfo{pages}{67--92}.
\newblock
\urldef\tempurl%
\url{https://doi.org/10.1080/08993400802114581}
\showURL{%
\tempurl}


\bibitem[\protect\citeauthoryear{McDirmid}{McDirmid}{2007}]%
        {McDirmid:2007}
\bibfield{author}{\bibinfo{person}{Sean McDirmid}.}
  \bibinfo{year}{2007}\natexlab{}.
\newblock \showarticletitle{Living It Up with a Live Programming Language}. In
  \bibinfo{booktitle}{\emph{Conference on Object-Oriented Programming, Systems,
  Languages, and Applications (OOPSLA)}}.
\newblock
\urldef\tempurl%
\url{https://doi.org/10.1145/1297027.1297073}
\showURL{%
\tempurl}


\bibitem[\protect\citeauthoryear{Nanevski, Pfenning, and Pientka}{Nanevski
  et~al\mbox{.}}{2008}]%
        {Nanevski2008}
\bibfield{author}{\bibinfo{person}{Aleksandar Nanevski}, \bibinfo{person}{Frank
  Pfenning}, {and} \bibinfo{person}{Brigitte Pientka}.}
  \bibinfo{year}{2008}\natexlab{}.
\newblock \showarticletitle{Contextual modal type theory}.
\newblock \bibinfo{journal}{\emph{{ACM} Trans. Comput. Log.}}
  \bibinfo{volume}{9}, \bibinfo{number}{3} (\bibinfo{year}{2008}).
\newblock
\urldef\tempurl%
\url{https://doi.org/10.1145/1352582.1352591}
\showDOI{\tempurl}


\bibitem[\protect\citeauthoryear{Nelson, Xie, and Ko}{Nelson
  et~al\mbox{.}}{2017}]%
        {Nelson2017}
\bibfield{author}{\bibinfo{person}{Greg~L. Nelson}, \bibinfo{person}{Benjamin
  Xie}, {and} \bibinfo{person}{Andrew~J. Ko}.} \bibinfo{year}{2017}\natexlab{}.
\newblock \showarticletitle{Comprehension First: Evaluating a Novel Pedagogy
  and Tutoring System for Program Tracing in {CS1}}. In
  \bibinfo{booktitle}{\emph{Conference on International Computing Education
  Research (ICER)}}.
\newblock
\urldef\tempurl%
\url{https://doi.org/10.1145/3105726.3106178}
\showURL{%
\tempurl}


\bibitem[\protect\citeauthoryear{Norell}{Norell}{2007}]%
        {norell:thesis}
\bibfield{author}{\bibinfo{person}{Ulf Norell}.}
  \bibinfo{year}{2007}\natexlab{}.
\newblock \emph{\bibinfo{title}{Towards a practical programming language based
  on dependent type theory}}.
\newblock \bibinfo{thesistype}{Ph.D. Dissertation}. \bibinfo{school}{Department
  of Computer Science and Engineering, Chalmers University of Technology},
  \bibinfo{address}{SE-412 96 G\"{o}teborg, Sweden}.
\newblock


\bibitem[\protect\citeauthoryear{Norell}{Norell}{2009}]%
        {norell2009dependently}
\bibfield{author}{\bibinfo{person}{Ulf Norell}.}
  \bibinfo{year}{2009}\natexlab{}.
\newblock \showarticletitle{Dependently typed programming in Agda}. In
  \bibinfo{booktitle}{\emph{International Workshop on Types in Languages Design
  and Implementation (TLDI)}}.
\newblock
\urldef\tempurl%
\url{https://doi.org/10.1145/1481861.1481862}
\showDOI{\tempurl}


\bibitem[\protect\citeauthoryear{Odersky, Zenger, and Zenger}{Odersky
  et~al\mbox{.}}{2001}]%
        {Odersky:2001lb}
\bibfield{author}{\bibinfo{person}{Martin Odersky}, \bibinfo{person}{Christoph
  Zenger}, {and} \bibinfo{person}{Matthias Zenger}.}
  \bibinfo{year}{2001}\natexlab{}.
\newblock \showarticletitle{Colored local type inference}. In
  \bibinfo{booktitle}{\emph{Principles of Programming Languages (POPL)}}.
\newblock
\urldef\tempurl%
\url{https://doi.org/10.1145/360204.360207}
\showURL{%
\tempurl}


\bibitem[\protect\citeauthoryear{Omar, Voysey, Chugh, and Hammer}{Omar
  et~al\mbox{.}}{2019}]%
        {HazelnutLive}
\bibfield{author}{\bibinfo{person}{Cyrus Omar}, \bibinfo{person}{Ian Voysey},
  \bibinfo{person}{Ravi Chugh}, {and} \bibinfo{person}{Matthew~A. Hammer}.}
  \bibinfo{year}{2019}\natexlab{}.
\newblock \showarticletitle{Live Functional Programming with Typed Holes}.
\newblock \bibinfo{journal}{\emph{{PACMPL}}} \bibinfo{volume}{3},
  \bibinfo{number}{{POPL}} (\bibinfo{year}{2019}).
\newblock
\urldef\tempurl%
\url{https://doi.org/10.1145/3290327}
\showDOI{\tempurl}


\bibitem[\protect\citeauthoryear{Omar, Voysey, Hilton, Aldrich, and
  Hammer}{Omar et~al\mbox{.}}{2017a}]%
        {popl-paper}
\bibfield{author}{\bibinfo{person}{Cyrus Omar}, \bibinfo{person}{Ian Voysey},
  \bibinfo{person}{Michael Hilton}, \bibinfo{person}{Jonathan Aldrich}, {and}
  \bibinfo{person}{Matthew~A. Hammer}.} \bibinfo{year}{2017}\natexlab{a}.
\newblock \showarticletitle{Hazelnut: A Bidirectionally Typed Structure Editor
  Calculus}. In \bibinfo{booktitle}{\emph{Principles of Programming Languages
  (POPL)}}.
\newblock
\urldef\tempurl%
\url{http://dl.acm.org/citation.cfm?id=3009900}
\showURL{%
\tempurl}


\bibitem[\protect\citeauthoryear{Omar, {Voysey}, {Hilton}, {Sunshine}, {Le
  Goues}, {Aldrich}, and {Hammer}}{Omar et~al\mbox{.}}{2017b}]%
        {HazelnutSNAPL}
\bibfield{author}{\bibinfo{person}{Cyrus Omar}, \bibinfo{person}{Ian {Voysey}},
  \bibinfo{person}{Michael {Hilton}}, \bibinfo{person}{Joshua {Sunshine}},
  \bibinfo{person}{Claire {Le Goues}}, \bibinfo{person}{Jonathan {Aldrich}},
  {and} \bibinfo{person}{Matthew~A. {Hammer}}.}
  \bibinfo{year}{2017}\natexlab{b}.
\newblock \showarticletitle{Toward Semantic Foundations for Program Editors}.
  In \bibinfo{booktitle}{\emph{Summit on Advances in Programming Languages
  (SNAPL)}}.
\newblock
\urldef\tempurl%
\url{https://doi.org/10.4230/LIPIcs.SNAPL.2017.11}
\showURL{%
\tempurl}


\bibitem[\protect\citeauthoryear{Pavlinovic, King, and Wies}{Pavlinovic
  et~al\mbox{.}}{2015}]%
        {Pavlinovic2015}
\bibfield{author}{\bibinfo{person}{Zvonimir Pavlinovic}, \bibinfo{person}{Tim
  King}, {and} \bibinfo{person}{Thomas Wies}.} \bibinfo{year}{2015}\natexlab{}.
\newblock \showarticletitle{Practical {S}{M}{T}-{B}ased {T}ype {E}rror
  {L}ocalization}. In \bibinfo{booktitle}{\emph{International Conference on
  Functional Programming (ICFP)}}.
\newblock
\urldef\tempurl%
\url{https://doi.org/10.1145/2784731.2784765}
\showURL{%
\tempurl}


\bibitem[\protect\citeauthoryear{Perera, Acar, Cheney, and Levy}{Perera
  et~al\mbox{.}}{2012}]%
        {DBLP:conf/icfp/PereraACL12}
\bibfield{author}{\bibinfo{person}{Roly Perera}, \bibinfo{person}{Umut~A.
  Acar}, \bibinfo{person}{James Cheney}, {and} \bibinfo{person}{Paul~Blain
  Levy}.} \bibinfo{year}{2012}\natexlab{}.
\newblock \showarticletitle{Functional programs that explain their work}. In
  \bibinfo{booktitle}{\emph{International Conference on Functional Programming
  (ICFP)}}.
\newblock
\urldef\tempurl%
\url{https://doi.org/10.1145/2364527.2364579}
\showURL{%
\tempurl}


\bibitem[\protect\citeauthoryear{P\'erez and Granger}{P\'erez and
  Granger}{2007}]%
        {PER-GRA:2007}
\bibfield{author}{\bibinfo{person}{Fernando P\'erez} {and}
  \bibinfo{person}{Brian~E. Granger}.} \bibinfo{year}{2007}\natexlab{}.
\newblock \showarticletitle{{IP}ython: a System for Interactive Scientific
  Computing}.
\newblock \bibinfo{journal}{\emph{Computing in Science and Engineering}}
  \bibinfo{volume}{9}, \bibinfo{number}{3} (\bibinfo{date}{May}
  \bibinfo{year}{2007}), \bibinfo{pages}{21--29}.
\newblock
\showISSN{1521-9615}
\urldef\tempurl%
\url{http://ipython.org}
\showURL{%
\tempurl}


\bibitem[\protect\citeauthoryear{{Peyton Jones}, Leather, and Alkemade}{{Peyton
  Jones} et~al\mbox{.}}{2014}]%
        {GHCHoles}
\bibfield{author}{\bibinfo{person}{Simon {Peyton Jones}}, \bibinfo{person}{Sean
  Leather}, {and} \bibinfo{person}{Thijs Alkemade}.}
  \bibinfo{year}{2014}\natexlab{}.
\newblock \bibinfo{title}{Language options --- Glasgow Haskell Compiler 8.4.1
  User's Guide (Typed Holes)}.
\newblock
  \bibinfo{howpublished}{\url{http://downloads.haskell.org/~ghc/latest/docs/html/users_guide/glasgow_exts.html}.
  Retrieved Apr 16, 2018}.
\newblock


\bibitem[\protect\citeauthoryear{Pientka}{Pientka}{2010}]%
        {DBLP:conf/flops/Pientka10}
\bibfield{author}{\bibinfo{person}{Brigitte Pientka}.}
  \bibinfo{year}{2010}\natexlab{}.
\newblock \showarticletitle{Beluga: Programming with Dependent Types,
  Contextual Data, and Contexts}. In \bibinfo{booktitle}{\emph{International
  Symposium on Functional and Logic Programming ({FLOPS})}}.
\newblock
\urldef\tempurl%
\url{http://dx._doi.org/10.1007/978-3-642-12251-4_1}
\showURL{%
\tempurl}


\bibitem[\protect\citeauthoryear{Pientka and Cave}{Pientka and Cave}{2015}]%
        {pientka2015inductive}
\bibfield{author}{\bibinfo{person}{Brigitte Pientka} {and}
  \bibinfo{person}{Andrew Cave}.} \bibinfo{year}{2015}\natexlab{}.
\newblock \showarticletitle{Inductive Beluga: Programming Proofs}. In
  \bibinfo{booktitle}{\emph{International Conference on Automated Deduction}}.
\newblock
\urldef\tempurl%
\url{https://doi.org/10.1007/978-3-319-21401-6\_18}
\showURL{%
\tempurl}


\bibitem[\protect\citeauthoryear{Pientka and Dunfield}{Pientka and
  Dunfield}{2008}]%
        {DBLP:conf/ppdp/PientkaD08}
\bibfield{author}{\bibinfo{person}{Brigitte Pientka} {and}
  \bibinfo{person}{Joshua Dunfield}.} \bibinfo{year}{2008}\natexlab{}.
\newblock \showarticletitle{Programming with proofs and explicit contexts}. In
  \bibinfo{booktitle}{\emph{Conference on Principles and Practice of
  Declarative Programming (PPDP)}}.
\newblock
\urldef\tempurl%
\url{https://doi.org/10.1145/1389449.1389469}
\showDOI{\tempurl}


\bibitem[\protect\citeauthoryear{Pierce}{Pierce}{2002}]%
        {tapl}
\bibfield{author}{\bibinfo{person}{Benjamin~C. Pierce}.}
  \bibinfo{year}{2002}\natexlab{}.
\newblock \bibinfo{booktitle}{\emph{Types and Programming Languages}}.
\newblock \bibinfo{publisher}{MIT Press}.
\newblock


\bibitem[\protect\citeauthoryear{Pierce and Turner}{Pierce and Turner}{2000}]%
        {Pierce:2000ve}
\bibfield{author}{\bibinfo{person}{Benjamin~C. Pierce} {and}
  \bibinfo{person}{David~N. Turner}.} \bibinfo{year}{2000}\natexlab{}.
\newblock \showarticletitle{Local type inference}.
\newblock \bibinfo{journal}{\emph{ACM Trans. Program. Lang. Syst.}}
  \bibinfo{volume}{22}, \bibinfo{number}{1} (\bibinfo{year}{2000}),
  \bibinfo{pages}{1--44}.
\newblock
\showISSN{0164-0925}
\urldef\tempurl%
\url{https://doi.org/10.1145/345099.345100}
\showURL{%
\tempurl}


\bibitem[\protect\citeauthoryear{Plotkin}{Plotkin}{2004}]%
        {DBLP:journals/jlp/Plotkin04a}
\bibfield{author}{\bibinfo{person}{Gordon~D. Plotkin}.}
  \bibinfo{year}{2004}\natexlab{}.
\newblock \showarticletitle{A structural approach to operational semantics}.
\newblock \bibinfo{journal}{\emph{J. Log. Algebr. Program.}}
  \bibinfo{volume}{60-61} (\bibinfo{year}{2004}), \bibinfo{pages}{17--139}.
\newblock


\bibitem[\protect\citeauthoryear{Rein, Ramson, Lincke, Hirschfeld, and
  Pape}{Rein et~al\mbox{.}}{2019}]%
        {DBLP:journals/programming/ReinRLHP19}
\bibfield{author}{\bibinfo{person}{Patrick Rein}, \bibinfo{person}{Stefan
  Ramson}, \bibinfo{person}{Jens Lincke}, \bibinfo{person}{Robert Hirschfeld},
  {and} \bibinfo{person}{Tobias Pape}.} \bibinfo{year}{2019}\natexlab{}.
\newblock \showarticletitle{Exploratory and Live, Programming and Coding - {A}
  Literature Study Comparing Perspectives on Liveness}.
\newblock \bibinfo{journal}{\emph{Programming Journal}} \bibinfo{volume}{3},
  \bibinfo{number}{1} (\bibinfo{year}{2019}).
\newblock
\urldef\tempurl%
\url{https://doi.org/10.22152/programming-journal.org/2019/3/1}
\showDOI{\tempurl}


\bibitem[\protect\citeauthoryear{Resnick, Maloney, Monroy-Hern\'{a}ndez, Rusk,
  Eastmond, Brennan, Millner, Rosenbaum, Silver, Silverman, and Kafai}{Resnick
  et~al\mbox{.}}{2009}]%
        {Resnick:2009:SP:1592761.1592779}
\bibfield{author}{\bibinfo{person}{Mitchel Resnick}, \bibinfo{person}{John
  Maloney}, \bibinfo{person}{Andr{\'e}s Monroy-Hern\'{a}ndez},
  \bibinfo{person}{Natalie Rusk}, \bibinfo{person}{Evelyn Eastmond},
  \bibinfo{person}{Karen Brennan}, \bibinfo{person}{Amon Millner},
  \bibinfo{person}{Eric Rosenbaum}, \bibinfo{person}{Jay Silver},
  \bibinfo{person}{Brian Silverman}, {and} \bibinfo{person}{Yasmin Kafai}.}
  \bibinfo{year}{2009}\natexlab{}.
\newblock \showarticletitle{{Scratch: Programming for All}}.
\newblock \bibinfo{journal}{\emph{Commun. ACM}} \bibinfo{volume}{52},
  \bibinfo{number}{11} (\bibinfo{date}{Nov.} \bibinfo{year}{2009}),
  \bibinfo{pages}{60--67}.
\newblock
\urldef\tempurl%
\url{http://doi.acm.org/10.1145/1592761.1592779}
\showURL{%
\tempurl}


\bibitem[\protect\citeauthoryear{Ricciotti, Stolarek, Perera, and
  Cheney}{Ricciotti et~al\mbox{.}}{2017}]%
        {DBLP:journals/pacmpl/RicciottiSPC17}
\bibfield{author}{\bibinfo{person}{Wilmer Ricciotti}, \bibinfo{person}{Jan
  Stolarek}, \bibinfo{person}{Roly Perera}, {and} \bibinfo{person}{James
  Cheney}.} \bibinfo{year}{2017}\natexlab{}.
\newblock \showarticletitle{Imperative functional programs that explain their
  work}.
\newblock \bibinfo{journal}{\emph{{PACMPL}}} \bibinfo{volume}{1},
  \bibinfo{number}{{ICFP}} (\bibinfo{year}{2017}).
\newblock
\urldef\tempurl%
\url{https://doi.org/10.1145/3110258}
\showDOI{\tempurl}


\bibitem[\protect\citeauthoryear{Rinard}{Rinard}{2012}]%
        {DBLP:conf/dac/Rinard12}
\bibfield{author}{\bibinfo{person}{Martin~C. Rinard}.}
  \bibinfo{year}{2012}\natexlab{}.
\newblock \showarticletitle{Obtaining and reasoning about good enough
  software}. In \bibinfo{booktitle}{\emph{Design Automation Conference (DAC)}}.
\newblock
\urldef\tempurl%
\url{https://doi.org/10.1145/2228360.2228526}
\showDOI{\tempurl}


\bibitem[\protect\citeauthoryear{Seidel, Jhala, and Weimer}{Seidel
  et~al\mbox{.}}{2016}]%
        {Seidel2016}
\bibfield{author}{\bibinfo{person}{Eric~L. Seidel}, \bibinfo{person}{Ranjit
  Jhala}, {and} \bibinfo{person}{Westley Weimer}.}
  \bibinfo{year}{2016}\natexlab{}.
\newblock \showarticletitle{Dynamic {W}itnesses for {S}tatic {T}ype {E}rrors
  (or, {I}ll-typed {P}rograms {U}sually {G}o {W}rong)}. In
  \bibinfo{booktitle}{\emph{International Conference on Functional Programming
  (ICFP)}}.
\newblock
\urldef\tempurl%
\url{https://doi.org/10.1145/2951913.2951915}
\showURL{%
\tempurl}


\bibitem[\protect\citeauthoryear{Siek and Taha}{Siek and Taha}{2006}]%
        {Siek06a}
\bibfield{author}{\bibinfo{person}{Jeremy~G. Siek} {and} \bibinfo{person}{Walid
  Taha}.} \bibinfo{year}{2006}\natexlab{}.
\newblock \showarticletitle{Gradual Typing for Functional Languages}. In
  \bibinfo{booktitle}{\emph{Scheme and Functional Programming Workshop}}.
\newblock
\urldef\tempurl%
\url{http://scheme2006.cs.uchicago.edu/13-siek.pdf}
\showURL{%
\tempurl}


\bibitem[\protect\citeauthoryear{Siek and Taha}{Siek and Taha}{2007}]%
        {Siek:2007qy}
\bibfield{author}{\bibinfo{person}{Jeremy~G. Siek} {and} \bibinfo{person}{Walid
  Taha}.} \bibinfo{year}{2007}\natexlab{}.
\newblock \showarticletitle{Gradual Typing for Objects}. In
  \bibinfo{booktitle}{\emph{European Conference on Object-Oriented Programming
  (ECOOP)}}.
\newblock
\urldef\tempurl%
\url{https://doi.org/10.1007/978-3-540-73589-2\_2}
\showURL{%
\tempurl}


\bibitem[\protect\citeauthoryear{Siek, Vitousek, Cimini, and Boyland}{Siek
  et~al\mbox{.}}{2015a}]%
        {DBLP:conf/snapl/SiekVCB15}
\bibfield{author}{\bibinfo{person}{Jeremy~G. Siek}, \bibinfo{person}{Michael~M.
  Vitousek}, \bibinfo{person}{Matteo Cimini}, {and} \bibinfo{person}{John~Tang
  Boyland}.} \bibinfo{year}{2015}\natexlab{a}.
\newblock \showarticletitle{Refined Criteria for Gradual Typing}. In
  \bibinfo{booktitle}{\emph{Summit on Advances in Programming Languages
  (SNAPL)}}.
\newblock
\urldef\tempurl%
\url{https://doi.org/10.4230/LIPIcs.SNAPL.2015.274}
\showDOI{\tempurl}


\bibitem[\protect\citeauthoryear{Siek, Vitousek, Cimini, Tobin{-}Hochstadt, and
  Garcia}{Siek et~al\mbox{.}}{2015b}]%
        {DBLP:conf/esop/SiekVCTG15}
\bibfield{author}{\bibinfo{person}{Jeremy~G. Siek}, \bibinfo{person}{Michael~M.
  Vitousek}, \bibinfo{person}{Matteo Cimini}, \bibinfo{person}{Sam
  Tobin{-}Hochstadt}, {and} \bibinfo{person}{Ronald Garcia}.}
  \bibinfo{year}{2015}\natexlab{b}.
\newblock \showarticletitle{Monotonic References for Efficient Gradual Typing}.
  In \bibinfo{booktitle}{\emph{European Symposium on Programming (ESOP)}}.
\newblock
\urldef\tempurl%
\url{https://doi.org/10.1007/978-3-662-46669-8_18}
\showDOI{\tempurl}


\bibitem[\protect\citeauthoryear{Siek and Wadler}{Siek and Wadler}{2010}]%
        {siek2010threesomes}
\bibfield{author}{\bibinfo{person}{Jeremy~G. Siek} {and}
  \bibinfo{person}{Philip Wadler}.} \bibinfo{year}{2010}\natexlab{}.
\newblock \showarticletitle{Threesomes, {W}ith and {W}ithout {B}lame}. In
  \bibinfo{booktitle}{\emph{Principles of Programming Languages (POPL)}}.
\newblock
\urldef\tempurl%
\url{https://doi.org/10.1145/1706299.1706342}
\showDOI{\tempurl}


\bibitem[\protect\citeauthoryear{Solar-Lezama}{Solar-Lezama}{2009}]%
        {solar2009sketching}
\bibfield{author}{\bibinfo{person}{Armando Solar-Lezama}.}
  \bibinfo{year}{2009}\natexlab{}.
\newblock \showarticletitle{The Sketching Approach to Program Synthesis}. In
  \bibinfo{booktitle}{\emph{Asian Symposium on Programming Languages and
  Systems (APLAS)}}.
\newblock
\urldef\tempurl%
\url{https://doi.org/10.1007/978-3-642-10672-9\_3}
\showURL{%
\tempurl}


\bibitem[\protect\citeauthoryear{Srivastava, Gulwani, and Foster}{Srivastava
  et~al\mbox{.}}{2013}]%
        {srivastava2013template}
\bibfield{author}{\bibinfo{person}{Saurabh Srivastava}, \bibinfo{person}{Sumit
  Gulwani}, {and} \bibinfo{person}{Jeffrey~S Foster}.}
  \bibinfo{year}{2013}\natexlab{}.
\newblock \showarticletitle{Template-based program verification and program
  synthesis}.
\newblock \bibinfo{journal}{\emph{International Journal on Software Tools for
  Technology Transfer}} \bibinfo{volume}{15}, \bibinfo{number}{5-6}
  (\bibinfo{year}{2013}), \bibinfo{pages}{497--518}.
\newblock
\urldef\tempurl%
\url{https://doi.org/10.1007/s10009-012-0223-4}
\showURL{%
\tempurl}


\bibitem[\protect\citeauthoryear{Stoyle, Hicks, Bierman, Sewell, and
  Neamtiu}{Stoyle et~al\mbox{.}}{2007}]%
        {DBLP:journals/toplas/StoyleHBSN07}
\bibfield{author}{\bibinfo{person}{Gareth~Paul Stoyle},
  \bibinfo{person}{Michael~W. Hicks}, \bibinfo{person}{Gavin~M. Bierman},
  \bibinfo{person}{Peter Sewell}, {and} \bibinfo{person}{Iulian Neamtiu}.}
  \bibinfo{year}{2007}\natexlab{}.
\newblock \showarticletitle{\emph{Mutatis Mutandis}: Safe and predictable
  dynamic software updating}.
\newblock \bibinfo{journal}{\emph{{ACM} Trans. Program. Lang. Syst.}}
  \bibinfo{volume}{29}, \bibinfo{number}{4} (\bibinfo{year}{2007}),
  \bibinfo{pages}{22}.
\newblock
\urldef\tempurl%
\url{https://doi.org/10.1145/1255450.1255455}
\showDOI{\tempurl}


\bibitem[\protect\citeauthoryear{Strecker, Luther, and von Henke}{Strecker
  et~al\mbox{.}}{1998}]%
        {Strecker:98a}
\bibfield{author}{\bibinfo{person}{M. Strecker}, \bibinfo{person}{M. Luther},
  {and} \bibinfo{person}{F. von Henke}.} \bibinfo{year}{1998}\natexlab{}.
\newblock \showarticletitle{{Interactive and automated proof construction in
  type theory}}.
\newblock In \bibinfo{booktitle}{\emph{Automated Deduction --- A Basis for
  Applications}}. Vol.~\bibinfo{volume}{I: Foundations}.
  \bibinfo{publisher}{Kluwer Academic Publishers}, Chapter 3: Interactive
  Theorem Proving.
\newblock


\bibitem[\protect\citeauthoryear{Sul\'{\i}r and Porub\"{a}n}{Sul\'{\i}r and
  Porub\"{a}n}{2018}]%
        {DBLP:journals/corr/abs-1806-07449}
\bibfield{author}{\bibinfo{person}{Mat\'{u}\v{s} Sul\'{\i}r} {and}
  \bibinfo{person}{Jaroslav Porub\"{a}n}.} \bibinfo{year}{2018}\natexlab{}.
\newblock \showarticletitle{Augmenting Source Code Lines with Sample Variable
  Values}. In \bibinfo{booktitle}{\emph{Conference on Program Comprehension
  (ICPC)}}.
\newblock
\showISBNx{978-1-4503-5714-2}
\urldef\tempurl%
\url{https://doi.org/10.1145/3196321.3196364}
\showDOI{\tempurl}


\bibitem[\protect\citeauthoryear{Takikawa, Feltey, Dean, Flatt, Findler,
  Tobin-Hochstadt, and Felleisen}{Takikawa et~al\mbox{.}}{2015}]%
        {takikawa_et_al:LIPIcs:2015:5215}
\bibfield{author}{\bibinfo{person}{Asumu Takikawa}, \bibinfo{person}{Daniel
  Feltey}, \bibinfo{person}{Earl Dean}, \bibinfo{person}{Matthew Flatt},
  \bibinfo{person}{Robert~Bruce Findler}, \bibinfo{person}{Sam
  Tobin-Hochstadt}, {and} \bibinfo{person}{Matthias Felleisen}.}
  \bibinfo{year}{2015}\natexlab{}.
\newblock \showarticletitle{{Towards Practical Gradual Typing}}. In
  \bibinfo{booktitle}{\emph{European Conference on Object-Oriented Programming
  (ECOOP)}}.
\newblock
\urldef\tempurl%
\url{https://doi.org/10.4230/LIPIcs.ECOOP.2015.4}
\showDOI{\tempurl}


\bibitem[\protect\citeauthoryear{Tanimoto}{Tanimoto}{1990}]%
        {DBLP:journals/vlc/Tanimoto90}
\bibfield{author}{\bibinfo{person}{Steven~L. Tanimoto}.}
  \bibinfo{year}{1990}\natexlab{}.
\newblock \showarticletitle{{VIVA:} {A} visual language for image processing}.
\newblock \bibinfo{journal}{\emph{J. Vis. Lang. Comput.}} \bibinfo{volume}{1},
  \bibinfo{number}{2} (\bibinfo{year}{1990}), \bibinfo{pages}{127--139}.
\newblock
\urldef\tempurl%
\url{https://doi.org/10.1016/S1045-926X(05)80012-6}
\showDOI{\tempurl}


\bibitem[\protect\citeauthoryear{Tanimoto}{Tanimoto}{2013}]%
        {DBLP:conf/icse/Tanimoto13}
\bibfield{author}{\bibinfo{person}{Steven~L. Tanimoto}.}
  \bibinfo{year}{2013}\natexlab{}.
\newblock \showarticletitle{A perspective on the evolution of live
  programming}. In \bibinfo{booktitle}{\emph{International Workshop on Live
  Programming (LIVE)}}.
\newblock
\urldef\tempurl%
\url{https://doi.org/10.1109/LIVE.2013.6617346}
\showURL{%
\tempurl}


\bibitem[\protect\citeauthoryear{Tolmach and Appel}{Tolmach and Appel}{1995}]%
        {DBLP:journals/jfp/TolmachA95}
\bibfield{author}{\bibinfo{person}{Andrew~P. Tolmach} {and}
  \bibinfo{person}{Andrew~W. Appel}.} \bibinfo{year}{1995}\natexlab{}.
\newblock \showarticletitle{A Debugger for Standard {ML}}.
\newblock \bibinfo{journal}{\emph{J. Funct. Program.}} \bibinfo{volume}{5},
  \bibinfo{number}{2} (\bibinfo{year}{1995}), \bibinfo{pages}{155--200}.
\newblock
\urldef\tempurl%
\url{https://doi.org/10.1017/S0956796800001313}
\showDOI{\tempurl}


\bibitem[\protect\citeauthoryear{Urban, Berghofer, and Norrish}{Urban
  et~al\mbox{.}}{2007}]%
        {urban}
\bibfield{author}{\bibinfo{person}{Christian Urban}, \bibinfo{person}{Stefan
  Berghofer}, {and} \bibinfo{person}{Michael Norrish}.}
  \bibinfo{year}{2007}\natexlab{}.
\newblock \showarticletitle{Barendregt's Variable Convention in Rule
  Inductions}. In \bibinfo{booktitle}{\emph{Conference on Automated Deduction
  (CADE)}}.
\newblock
\urldef\tempurl%
\url{https://doi.org/10.1007/978-3-540-73595-3\_4}
\showURL{%
\tempurl}


\bibitem[\protect\citeauthoryear{Victor}{Victor}{2012}]%
        {victor2012inventing}
\bibfield{author}{\bibinfo{person}{B Victor}.} \bibinfo{year}{2012}\natexlab{}.
\newblock \bibinfo{title}{Inventing on principle, Invited talk at the Canadian
  University Software Engineering Conference (CUSEC)}.
\newblock
\newblock
\urldef\tempurl%
\url{http://worrydream.com/#!/InventingOnPrinciple}
\showURL{%
\tempurl}


\bibitem[\protect\citeauthoryear{Voelter, Ratiu, Schaetz, and Kolb}{Voelter
  et~al\mbox{.}}{2012}]%
        {voelter_mbeddr:_2012}
\bibfield{author}{\bibinfo{person}{Markus Voelter}, \bibinfo{person}{Daniel
  Ratiu}, \bibinfo{person}{Bernhard Schaetz}, {and} \bibinfo{person}{Bernd
  Kolb}.} \bibinfo{year}{2012}\natexlab{}.
\newblock \showarticletitle{{mbeddr: An Extensible C-based Programming Language
  and {IDE} for Embedded Systems}}. In \bibinfo{booktitle}{\emph{SPLASH}}.
\newblock
\urldef\tempurl%
\url{http://doi.acm.org/10.1145/2384716.2384767}
\showURL{%
\tempurl}


\bibitem[\protect\citeauthoryear{Voelter, Siegmund, Berger, and Kolb}{Voelter
  et~al\mbox{.}}{2014}]%
        {DBLP:conf/sle/VolterSBK14}
\bibfield{author}{\bibinfo{person}{Markus Voelter}, \bibinfo{person}{Janet
  Siegmund}, \bibinfo{person}{Thorsten Berger}, {and} \bibinfo{person}{Bernd
  Kolb}.} \bibinfo{year}{2014}\natexlab{}.
\newblock \showarticletitle{{Towards User-Friendly Projectional Editors}}. In
  \bibinfo{booktitle}{\emph{International Conference on Software Language
  Engineering ({SLE})}}.
\newblock
\urldef\tempurl%
\url{http://dx.doi.org/10.1007/978-3-319-11245-9\_3}
\showURL{%
\tempurl}


\bibitem[\protect\citeauthoryear{Vytiniotis, {Peyton Jones}, and
  Magalh{\~{a}}es}{Vytiniotis et~al\mbox{.}}{2012}]%
        {DBLP:conf/icfp/VytiniotisJM12}
\bibfield{author}{\bibinfo{person}{Dimitrios Vytiniotis},
  \bibinfo{person}{Simon~L. {Peyton Jones}}, {and}
  \bibinfo{person}{Jos{\'{e}}~Pedro Magalh{\~{a}}es}.}
  \bibinfo{year}{2012}\natexlab{}.
\newblock \showarticletitle{Equality proofs and deferred type errors: a
  compiler pearl}. In \bibinfo{booktitle}{\emph{{International Conference on
  Functional Programming (ICFP)}}}.
\newblock
\urldef\tempurl%
\url{https://doi.org/10.1145/2364527.2364554}
\showURL{%
\tempurl}


\bibitem[\protect\citeauthoryear{Wadler and Findler}{Wadler and
  Findler}{2009}]%
        {DBLP:conf/esop/WadlerF09}
\bibfield{author}{\bibinfo{person}{Philip Wadler} {and}
  \bibinfo{person}{Robert~Bruce Findler}.} \bibinfo{year}{2009}\natexlab{}.
\newblock \showarticletitle{Well-Typed Programs Can't Be Blamed}. In
  \bibinfo{booktitle}{\emph{European Symposium on Programming (ESOP)}}.
\newblock
\urldef\tempurl%
\url{https://doi.org/10.1007/978-3-642-00590-9_1}
\showDOI{\tempurl}


\bibitem[\protect\citeauthoryear{Wakeling}{Wakeling}{2007}]%
        {DBLP:journals/jfp/Wakeling07}
\bibfield{author}{\bibinfo{person}{David Wakeling}.}
  \bibinfo{year}{2007}\natexlab{}.
\newblock \showarticletitle{Spreadsheet functional programming}.
\newblock \bibinfo{journal}{\emph{J. Funct. Program.}} \bibinfo{volume}{17},
  \bibinfo{number}{1} (\bibinfo{year}{2007}), \bibinfo{pages}{131--143}.
\newblock
\urldef\tempurl%
\url{https://doi.org/10.1017/S0956796806006186}
\showDOI{\tempurl}


\bibitem[\protect\citeauthoryear{Weinberg}{Weinberg}{1971}]%
        {weinberg1971psychology}
\bibfield{author}{\bibinfo{person}{Gerald~M Weinberg}.}
  \bibinfo{year}{1971}\natexlab{}.
\newblock \bibinfo{booktitle}{\emph{{The Psychology of Computer Programming}}}.
\newblock \bibinfo{publisher}{Van Nostrand Reinhold New York}.
\newblock


\bibitem[\protect\citeauthoryear{Weintrop, Afzal, Salac, Francis, Li, Shepherd,
  and Franklin}{Weintrop et~al\mbox{.}}{2018}]%
        {DBLP:conf/chi/WeintropASFLSF18}
\bibfield{author}{\bibinfo{person}{David Weintrop}, \bibinfo{person}{Afsoon
  Afzal}, \bibinfo{person}{Jean Salac}, \bibinfo{person}{Patrick Francis},
  \bibinfo{person}{Boyang Li}, \bibinfo{person}{David~C. Shepherd}, {and}
  \bibinfo{person}{Diana Franklin}.} \bibinfo{year}{2018}\natexlab{}.
\newblock \showarticletitle{Evaluating CoBlox: {A} Comparative Study of
  Robotics Programming Environments for Adult Novices}. In
  \bibinfo{booktitle}{\emph{Conference on Human Factors in Computing (CHI)}}.
\newblock
\urldef\tempurl%
\url{https://doi.org/10.1145/3173574.3173940}
\showDOI{\tempurl}


\bibitem[\protect\citeauthoryear{Weintrop and Wilensky}{Weintrop and
  Wilensky}{2015}]%
        {DBLP:conf/acmidc/WeintropW15}
\bibfield{author}{\bibinfo{person}{David Weintrop} {and} \bibinfo{person}{Uri
  Wilensky}.} \bibinfo{year}{2015}\natexlab{}.
\newblock \showarticletitle{To block or not to block, that is the question:
  students' perceptions of blocks-based programming}. In
  \bibinfo{booktitle}{\emph{International Conference on Interaction Design and
  Children (IDC)}}.
\newblock
\urldef\tempurl%
\url{https://doi.org/10.1145/2771839.2771860}
\showDOI{\tempurl}


\bibitem[\protect\citeauthoryear{Whitington and Ridge}{Whitington and
  Ridge}{2017}]%
        {ocaml-stepper}
\bibfield{author}{\bibinfo{person}{John Whitington} {and} \bibinfo{person}{Tom
  Ridge}.} \bibinfo{year}{2017}\natexlab{}.
\newblock \showarticletitle{Visualizing the Evaluation of Functional Programs
  for Debugging}. In \bibinfo{booktitle}{\emph{Symposium on Languages,
  Applications and Technologies (SLATE)}}.
\newblock
\urldef\tempurl%
\url{https://doi.org/10.4230/OASIcs.SLATE.2017.7}
\showURL{%
\tempurl}


\bibitem[\protect\citeauthoryear{Wright and Felleisen}{Wright and
  Felleisen}{1994}]%
        {wright94:_type_soundness}
\bibfield{author}{\bibinfo{person}{Andrew~K. Wright} {and}
  \bibinfo{person}{Matthias Felleisen}.} \bibinfo{year}{1994}\natexlab{}.
\newblock \showarticletitle{A syntactic approach to type soundness}.
\newblock \bibinfo{journal}{\emph{Information and Computation}}
  \bibinfo{volume}{115}, \bibinfo{number}{1} (\bibinfo{year}{1994}),
  \bibinfo{pages}{38--94}.
\newblock
\showISSN{0890-5401}
\urldef\tempurl%
\url{https://doi.org/10.1006/inco.1994.1093}
\showURL{%
\tempurl}


\bibitem[\protect\citeauthoryear{Xie, Bi, and d.~S.~Oliveira}{Xie
  et~al\mbox{.}}{2018}]%
        {DBLP:conf/esop/XieBO18}
\bibfield{author}{\bibinfo{person}{Ningning Xie}, \bibinfo{person}{Xuan Bi},
  {and} \bibinfo{person}{Bruno~C. d. S.~Oliveira}.}
  \bibinfo{year}{2018}\natexlab{}.
\newblock \showarticletitle{Consistent Subtyping for All}. In
  \bibinfo{booktitle}{\emph{European Symposium on Programming (ESOP)}}.
\newblock
\urldef\tempurl%
\url{https://doi.org/10.1007/978-3-319-89884-1_1}
\showDOI{\tempurl}


\bibitem[\protect\citeauthoryear{Yoon and Myers}{Yoon and Myers}{2014}]%
        {6883030}
\bibfield{author}{\bibinfo{person}{Y.~S. Yoon} {and} \bibinfo{person}{B.~A.
  Myers}.} \bibinfo{year}{2014}\natexlab{}.
\newblock \showarticletitle{A longitudinal study of programmers' backtracking}.
  In \bibinfo{booktitle}{\emph{IEEE Symposium on Visual Languages and
  Human-Centric Computing (VL/HCC)}}.
\newblock
\urldef\tempurl%
\url{https://doi.org/10.1109/VLHCC.2014.6883030}
\showURL{%
\tempurl}


\bibitem[\protect\citeauthoryear{Zhang, Myers, Vytiniotis, and {Peyton
  Jones}}{Zhang et~al\mbox{.}}{2017}]%
        {sherrloc}
\bibfield{author}{\bibinfo{person}{Danfeng Zhang}, \bibinfo{person}{Andrew~C.
  Myers}, \bibinfo{person}{Dimitrios Vytiniotis}, {and}
  \bibinfo{person}{Simon~L. {Peyton Jones}}.} \bibinfo{year}{2017}\natexlab{}.
\newblock \showarticletitle{SHErrLoc: {A} Static Holistic Error Locator}.
\newblock \bibinfo{journal}{\emph{{ACM} Trans. Program. Lang. Syst.}}
  \bibinfo{volume}{39}, \bibinfo{number}{4} (\bibinfo{year}{2017}),
  \bibinfo{pages}{18:1--18:47}.
\newblock
\urldef\tempurl%
\url{https://doi.org/10.1145/3121137}
\showDOI{\tempurl}


\bibitem[\protect\citeauthoryear{Zhang, Verma, Sheth, Schankula, Koehl, Kelly,
  Irfan, and Anand}{Zhang et~al\mbox{.}}{2018}]%
        {zhang2018graphics}
\bibfield{author}{\bibinfo{person}{John Zhang}, \bibinfo{person}{Anirudh
  Verma}, \bibinfo{person}{Chinmay Sheth}, \bibinfo{person}{Christopher~W
  Schankula}, \bibinfo{person}{Stephanie Koehl}, \bibinfo{person}{Andrew
  Kelly}, \bibinfo{person}{Yumna Irfan}, {and} \bibinfo{person}{Christopher~K
  Anand}.} \bibinfo{year}{2018}\natexlab{}.
\newblock \showarticletitle{Graphics Programming in Elm Develops Math Knowledge
  \& Social Cohesion}. In \bibinfo{booktitle}{\emph{International Conference on
  Computer Science and Software Engineering (CASCON)}}.
\newblock


\end{thebibliography}

\ifarxiv
\clearpage
\appendix
% !TEX root = hazelnut-dynamics.tex

\newcommand{\additionalDefnsSec}{Additional Definitions for Hazelnut Live}
\section{\protect\additionalDefnsSec} % don't like the all-caps thing that the template does, so protecting it from that
\label{sec:additional-defns}

\subsection{Substitution}
\label{sec:substitution}
\judgbox
  {[\dexp/x]\dexp' = \dexp''}
  {$\dexp''$ is obtained by substituting $\dexp$ for $x$ in $\dexp'$}
  %% {$\dexp''$ is the result of substituting $\dexp$ for $u$ in $\dexp'$}

\vspace{5pt}
\judgbox
  {[\dexp/x]\sigma = \sigma'}
  {$\sigma'$ is obtained by substituting $\dexp$ for $x$ in $\sigma$}
  %% {$\dexp''$ is the result of substituting $\dexp$ for $u$ in $\dexp'$}
\[
\begin{array}{lcll}
\substitute{\dexp}{x}{c}
&=&
c\\
\substitute{\dexp}{x}{x}
&=&
\dexp
\\%
\substitute{\dexp}{x}{y}
&=&
y
& \text{when $x \neq y$}
\\
\substitute{\dexp}{x}{\halam{x}{\htau}{\dexp'}}
&=&
\halam{x}{\htau}{\dexp'}\\
\substitute{\dexp}{x}{\halam{y}{\htau}{\dexp'}}
&=&
\halam{y}{\htau}{\substitute{\dexp}{x}{\dexp'}}
& \text{when $x \neq y$ and $y \notin \fvof{d}$}
\\
\substitute{\dexp}{x}{\dap{\dexp_1}{\dexp_2}}
&=&
\dap{(\substitute{\dexp}{x}{\dexp_1})}{\substitute{\dexp}{x}{\dexp_2}}
\\
\substitute{\dexp}{x}{\dehole{u}{\subst}{}}
&=&
\dehole{u}{\substitute{\dexp}{x}{\subst}}{}
\\
\substitute{\dexp}{x}{\dhole{\dexp'}{u}{\subst}{}}
&=&
\dhole{\substitute{\dexp}{x}{\dexp'}}{u}{\substitute{\dexp}{x}{\subst}}{}
\\
\substitute{\dexp}{x}{\dcasttwo{\dexp'}{\htau_1}{\htau_2}}
&=&
\dcasttwo{(\substitute{\dexp}{x}{\dexp'})}{\htau_1}{\htau_2}
\\
\substitute{\dexp}{x}{\dcastfail{\dexp'}{\htau_1}{\htau_2}}
&=&
\dcastfail{(\substitute{\dexp}{x}{\dexp'})}{\htau_1}{\htau_2}
\\[6px]
\substitute{\dexp}{x}{\cdot}
&=&
\cdot\\
\substitute{\dexp}{x}{\sigma, d/y}
&=&
\substitute{\dexp}{x}{\sigma}, \substitute{\dexp}{x}{d}/y
%% {[\dexp_1 / x] \dcast{\htau}{\dexp}}
%% &=&
%% \dcast{\htau}{[\dexp_1 / x] \dexp}
\end{array}
\]

\vspace{5pt}
\begin{lem}[Substitution] \label{thm:substitution}~
\begin{enumerate}[nolistsep]
\item If $\hasType{\hDelta}{\hGamma, x : \htau'}{d}{\htau}$ and $\hasType{\hDelta}{\hGamma}{d'}{\htau'}$ then $\hasType{\hDelta}{\hGamma}{[d'/x]d}{\htau}$.
\item If $\hasType{\hDelta}{\hGamma, x : \htau'}{\sigma}{\hGamma'}$ and $\hasType{\hDelta}{\hGamma}{d'}{\htau'}$ then $\hasType{\hDelta}{\hGamma}{[d'/x]\sigma}{\hGamma'}$.
\end{enumerate}
\end{lem}
% \begin{proof} By rule induction on the first assumption in each case. The conclusion follows from the definition of substitution in each case. \end{proof}

\subsection{Canonical Forms}
\label{sec:canonical-forms}

\begin{lem}[Canonical Value Forms]\label{thm:canonincal-value-forms}
  If $\hasType{\hDelta}{\emptyset}{\dexp}{\htau}$ and $\isValue{\dexp}$
  then $\htau\neq\tehole$ and
  \begin{enumerate}[nolistsep]
    \item If $\htau=b$ then $\dexp=c$.
    \item If $\htau=\tarr{\htau_1}{\htau_2}$
          then $\dexp=\halam{x}{\htau_1}{\dexp'}$
          where $\hasType{\hDelta}{x : \htau_1}{\dexp'}{\htau_2}$.
  \end{enumerate}
\end{lem}

\begin{lem}[Canonical Boxed Forms]\label{thm:canonical-boxed-forms}
  If $\hasType{\hDelta}{\emptyset}{\dexp}{\htau}$ and $\isBoxedValue{\dexp}$
  then
  \begin{enumerate}[nolistsep]
    \item If $\htau=b$ then $\dexp=c$.
    \item If $\htau=\tarr{\htau_1}{\htau_2}$ then either
      \begin{enumerate}
        \item
          $\dexp=\halam{x}{\htau_1}{\dexp'}$
          where $\hasType{\hDelta}{x : \htau_1}{\dexp'}{\htau_2}$, or
        \item
          $\dexp=\dcasttwo{\dexp'}{\tarr{\htau_1'}{\htau_2'}}{\tarr{\htau_1}{\htau_2}}$
          where $\tarr{\htau_1'}{\htau_2'}\neq\tarr{\htau_1}{\htau_2}$
          and $\hasType{\hDelta}{\emptyset}{\dexp'}{\tarr{\htau_1'}{\htau_2'}}$.
      \end{enumerate}
    \item If $\htau=\tehole$
          then $\dexp=\dcasttwo{\dexp'}{\htau'}{\tehole}$
          where $\isGround{\htau'}$
          and $\hasType{\hDelta}{\emptyset}{\dexp'}{\htau'}$.
  \end{enumerate}
\end{lem}

\begin{lem}[Canonical Indeterminate Forms]
  If $\hasType{\hDelta}{\emptyset}{\dexp}{\htau}$
  and $\isIndet{\dexp}$
  then
  \begin{enumerate}[nolistsep]
    \item 
      If $\htau = b$ then either
        \begin{enumerate}
          \item $\dexp = \dehole{u}{\subst}{}$ where $\Dbinding{u}{\Gamma}{b} \in \hDelta$ and $\hasType{\hDelta}{\emptyset}{\subst}{\hGamma}$, or
          \item $\dexp = \dhole{\dexp'}{u}{\subst}{}$ where $\hasType{\hDelta}{\emptyset}{\dexp'}{\htau'}$ and $\isFinal{\dexp'}$ and $\Dbinding{u}{\Gamma}{b} \in \hDelta$ and $\hasType{\hDelta}{\emptyset}{\subst}{\hGamma}$, or
          \item $\dexp = \dap{\dexp_1}{\dexp_2}$ where $\hasType{\hDelta}{\emptyset}{\dexp_1}{\tarr{\htau_2}{b}}$ and $\hasType{\hDelta}{\emptyset}{\dexp_2}{\htau_2}$ and $\isIndet{\dexp_1}$ and $\isFinal{\dexp_2}$ and $\dexp_1 \neq \dcasttwo{\dexp_1'}{\tarr{\htau_3}{\htau_4}}{\tarr{\htau_3'}{\htau_4'}}$ for any $\dexp_1', \htau_3, \htau_4, \htau_3', \htau_4'$, or 
          \item $\dexp = \dcasttwo{\dexp'}{\tehole}{b}$ where $\hasType{\hDelta}{\emptyset}{\dexp'}{\tehole}$ and $\isIndet{\dexp'}$ and $\dexp' \neq \dcasttwo{\dexp''}{\htau'}{\tehole}$ for any $\dexp'', \htau'$, or 
          \item $\dexp = \dcastfail{\dexp'}{\htau'}{b}$ where $\hasType{\hDelta}{\emptyset}{\dexp'}{\htau'}$ and $\isGround{\htau'}$ and $\htau' \neq b$.
        \end{enumerate}
    \item 
      If $\htau = \tarr{\htau_1}{\htau_2}$ then either 
        \begin{enumerate}
          \item $\dexp = \dehole{u}{\subst}{}$ where $\Dbinding{u}{\Gamma}{\tarr{\htau_1}{\htau_2}} \in \hDelta$ and $\hasType{\hDelta}{\emptyset}{\subst}{\hGamma}$, or
          \item $\dexp = \dhole{\dexp'}{u}{\subst}{}$ where $\hasType{\hDelta}{\emptyset}{\dexp'}{\htau'}$ and $\isFinal{\dexp'}$ and $\Dbinding{u}{\Gamma}{\tarr{\htau_1}{\htau_2}} \in \hDelta$ and $\hasType{\hDelta}{\emptyset}{\subst}{\hGamma}$, or
          \item $\dexp = \dap{\dexp_1}{\dexp_2}$ where $\hasType{\hDelta}{\emptyset}{\dexp_1}{\tarr{\htau_2'}{\tarr{\htau_1}{\htau_2}}}$ and $\hasType{\hDelta}{\emptyset}{\dexp_2}{\htau_2'}$ and $\isIndet{\dexp_1}$ and $\isFinal{\dexp_2}$ and $\dexp_1 \neq \dcasttwo{\dexp_1'}{\tarr{\htau_3}{\htau_4}}{\tarr{\htau_3'}{\htau_4'}}$ for any $\dexp_1', \htau_3, \htau_4, \htau_3', \htau_4'$, or 
          \item $\dexp = \dcasttwo{\dexp'}{\tarr{\htau_1'}{\htau_2'}}{\tarr{\htau_1}{\htau_2}}$ where $\hasType{\hDelta}{\emptyset}{\dexp'}{\tarr{\htau_1'}{\htau_2'}}$ and $\isIndet{\dexp'}$ and $\tarr{\htau_1'}{\htau_2'} \neq \tarr{\htau_1}{\htau_2}$, or 
          \item $\dexp = \dcasttwo{\dexp'}{\tehole}{\tarr{\tehole}{\tehole}}$ and $\htau_1 = \tehole$ and $\htau_2 = \tehole$ where $\hasType{\hDelta}{\emptyset}{\dexp'}{\tehole}$ and $\isIndet{\dexp'}$ and $\dexp' \neq \dcasttwo{\dexp''}{\htau'}{\tehole}$ for any $\dexp'', \htau'$, or 
          \item $\dexp = \dcastfail{\dexp'}{\htau'}{\tarr{\tehole}{\tehole}}$ and $\htau_1 = \tehole$ and $\htau_2 = \tehole$ where $\hasType{\hDelta}{\emptyset}{\dexp'}{\htau'}$ and $\isGround{\htau'}$ and $\htau' \neq \tarr{\tehole}{\tehole}$.
        \end{enumerate}
    \item 
      If $\htau = \tehole$ then either 
        \begin{enumerate}
          \item $\dexp = \dehole{u}{\subst}{}$ where $\Dbinding{u}{\Gamma}{\tehole} \in \hDelta$ and $\hasType{\hDelta}{\emptyset}{\subst}{\hGamma}$, or
          \item $\dexp = \dhole{\dexp'}{u}{\subst}{}$ where $\hasType{\hDelta}{\emptyset}{\dexp'}{\htau'}$ and $\isFinal{\dexp'}$ and $\Dbinding{u}{\Gamma}{\tehole} \in \hDelta$ and $\hasType{\hDelta}{\emptyset}{\subst}{\hGamma}$, or
          \item $\dexp = \dap{\dexp_1}{\dexp_2}$ and $\hasType{\hDelta}{\emptyset}{\dexp_1}{\tarr{\htau_2}{\tehole}}$ and $\hasType{\hDelta}{\emptyset}{\dexp_2}{\htau_2}$ and $\isIndet{\dexp_1}$ and $\isFinal{\dexp_2}$ and $\dexp_1 \neq \dcasttwo{\dexp_1'}{\tarr{\htau_3}{\htau_4}}{\tarr{\htau_3'}{\htau_4'}}$ for any $\dexp_1', \htau_3, \htau_4, \htau_3', \htau_4'$, or 
          \item $\dexp = \dcasttwo{\dexp'}{\htau'}{\tehole}$ where $\hasType{\hDelta}{\emptyset}{\dexp'}{\htau'}$ and $\isGround{\htau'}$ and $\isIndet{\dexp'}$.
        \end{enumerate}
  \end{enumerate}
\end{lem}

% The proofs for all three of these theorems follow by straightforward rule induction.

% No weakening for Gammas in Delta:
% If $\hasType{\Delta, \Dbinding{u}{\hGamma}{\tau}}{\hGamma'}{d}{\tau'}$ then $\hasType{\Delta, \Dbinding{u}{\hGamma, x : \tau''}{\tau}}{\hGamma'}{d}{\tau'}$.

\subsection{Complete Programs}
\label{sec:complete-programs}

% !TEX root = hazelnut-dynamics.tex
\begin{figure}[h]
\judgbox{\isComplete{\htau}}{$\htau$ is complete}
\begin{mathpar}
\inferrule[TCBase]{ }{
  \isComplete{b}
}

\inferrule[TCArr]{
  \isComplete{\htau_1}\\
  \isComplete{\htau_2}
}{
  \isComplete{\tarr{\htau_1}{\htau_2}}
}
\end{mathpar}

\vsepRule

\judgbox{\isComplete{\hexp}}{$\hexp$ is complete}
\begin{mathpar}
\inferrule[ECVar]{ }{
  \isComplete{x}
}

\inferrule[ECConst]{ }{
  \isComplete{c}
}

\inferrule[ECLam1]{
  \isComplete{\htau}\\
  \isComplete{\hexp}
}{
  \isComplete{\halam{x}{\htau}{\hexp}}
}

\inferrule[ECLam2]{
  \isComplete{\hexp}
}{
  \isComplete{\hlam{x}{\hexp}}
}

\inferrule[ECAp]{
  \isComplete{\hexp_1}\\
  \isComplete{\hexp_2}
}{
  \isComplete{\hap{\hexp_1}{\hexp_2}}
}

\inferrule[ECAsc]{
  \isComplete{\hexp}\\
  \isComplete{\htau}
}{
  \isComplete{\hexp : \htau}
}
\end{mathpar}

\vsepRule

\judgbox{\isComplete{\dexp}}{$\dexp$ is complete}
\begin{mathpar}
\inferrule[DCVar]{ }{
  \isComplete{x}
}

\inferrule[DCConst]{ }{
  \isComplete{c}
}

\inferrule[DCLam]{
  \isComplete{\htau}\\
  \isComplete{\dexp}
}{
  \isComplete{\dlam{x}{\htau}{\dexp}}
}

\inferrule[DCAp]{
  \isComplete{\dexp_1}\\
  \isComplete{\dexp_2}
}{
  \isComplete{\dap{\dexp_1}{\dexp_2}}
}

\inferrule[DCCast]{
  \isComplete{\dexp}\\
  \isComplete{\htau_1}\\
  \isComplete{\htau_2}
}{
  \isComplete{\dcasttwo{\dexp}{\htau_1}{\htau_2}}
}
\end{mathpar}

\caption{Complete types, external expressions, and internal expressions}
\label{fig:complete}
\end{figure}

We define $\isComplete{\hGamma}$ as follows.
\begin{defn}[Typing Context Completeness]
$\isComplete{\hGamma}$ iff for each $x : \htau \in \hGamma$, we have $\isComplete{\htau}$.
\end{defn}

When two types are complete and consistent, they are equal.

\begin{lem}[Complete Consistency]\label{lem:complete-consistency} If $\tconsistent{\htau_1}{\htau_2}$ and $\isComplete{\htau_1}$ and $\isComplete{\htau_2}$ then $\htau_1 = \htau_2$.
\end{lem}
\begin{proof} By straightforward rule induction. \end{proof}

This implies that in a well-typed and complete internal expression, every cast is 
an identity cast.

\begin{lem}[Complete Casts] If $\hasType{\hGamma}{\hDelta}{\dcasttwo{\dexp}{\htau_1}{\htau_2}}{\htau_2}$ and $\isComplete{\dcasttwo{\dexp}{\htau_1}{\htau_2}}$ then $\htau_1 = \htau_2$. \end{lem}
\begin{proof} By straightforward rule induction and Lemma~\ref{lem:complete-consistency}. \end{proof}

\subsection{Multiple Steps}
\label{sec:multi-step}

\begin{figure}[h]
\judgbox{\multiStepsTo{\dexp}{\dexp'}}{$\dexp$ multi-steps to $\dexp'$}
\begin{mathpar}
\inferrule[MultiStepRefl]{~}{
  \multiStepsTo{\dexp}{\dexp}
}

\inferrule[MultiStepSteps]{
  \stepsToD{}{\dexp}{\dexp'}\\
  \multiStepsTo{\dexp'}{\dexp''}
}{
  \multiStepsTo{\dexp}{\dexp''}
}
\end{mathpar}
\CaptionLabel{Multi-Step Transitions}{fig:multi-step}
\end{figure}

\subsection{Hole Filling}\label{sec:hole-filling}
\begin{lem}[Filling] ~
  \begin{enumerate}[nolistsep]
  \item If $\hasType{\hDelta, \Dbinding{u}{\hGamma'}{\htau'}}{\hGamma}{\dexp}{\tau}$
  and $\hasType{\hDelta}{\hGamma'}{\dexp'}{\htau'}$
  then $\hasType{\hDelta}{\hGamma}{\instantiate{\dexp'}{u}{\dexp}}{\tau}$.
  \item If $\hasType{\hDelta, \Dbinding{u}{\hGamma'}{\htau'}}{\hGamma}{\sigma}{\hGamma''}$
  and $\hasType{\hDelta}{\hGamma'}{\dexp'}{\htau'}$
  then $\hasType{\hDelta}{\hGamma}{\instantiate{\dexp'}{u}{\sigma}}{\hGamma''}$.
  \end{enumerate}
\end{lem}
\begin{proof}
In each case, we proceed by rule induction on the first assumption, appealing to the Substitution Lemma as necessary.
\end{proof}

To prove the Commutativity theorem, we need the auxiliary definitions in Fig.~\ref{fig:evalctx-filling}, which lift hole filling to evaluation contexts taking care to consider the special situation where the mark is inside the hole that is being filled.
% !TEX root = hazelnut-dynamics.tex

\begin{figure}[h]
\judgbox{\inhole{u}{\evalctx}}{The mark in $\evalctx$ is inside non-empty hole closure $u$}
\begin{mathpar}
\inferrule[InHoleNEHole]{~}{
  \inhole{u}{\dhole{\evalctx}{u}{\subst}{}}
}

\inferrule[InHoleAp1]{
  \inhole{u}{\evalctx}
}{
  \inhole{u}{\hap{\evalctx}{\dexp}}
}

\inferrule[InHoleAp2]{
  \inhole{u}{\evalctx}
}{
  \inhole{u}{\hap{\dexp}{\evalctx}}
}
\\
\inferrule[InHoleCast]{
  \inhole{u}{\evalctx}
}{
  \inhole{u}{\dcasttwo{\evalctx}{\htau_1}{\htau_2}}
}

\inferrule[InHoleFailedCast]{
  \inhole{u}{\evalctx}
}{
  \inhole{u}{\dcastfail{\evalctx}{\htau_1}{\htau_2}}
}
\end{mathpar}

\vsepRule

\judgbox{\instantiate{d}{u}{\evalctx} = \evalctx'}{$\evalctx'$ is obtained by filling hole $u$ in $\evalctx$ with $d$}
\begin{mathpar}
\inferrule[EFillMark]{~}{
  \instantiate{d}{u}{\evalhole} = \evalhole
}

\inferrule[EFillAp1]{
  \instantiate{d}{u}{\evalctx} = \evalctx'
}{
  \instantiate{d}{u}{\hap{\evalctx}{d_2}} = \hap{\evalctx'}{
  	\instantiate{d}{u}{d_2}}
}

\inferrule[EFillAp2]{
  \instantiate{d}{u}{\evalctx} = \evalctx'
}{
  \instantiate{d}{u}{\hap{d_1}{\evalctx}} = \hap{(
  	\instantiate{d}{u}{d_1})}{\evalctx'}
}

\inferrule[EFillNEHole]{
  u \neq v\\
  \instantiate{d}{u}{\evalctx}={\evalctx'}
}{
  \instantiate{d}{u}{\dhole{\evalctx}{v}{\sigma}{}} = \dhole{\evalctx'}{v}{\instantiate{d}{u}{\sigma}{}}{}
}

\inferrule[EFillCast]{
	\instantiate{d}{u}{\evalctx} = \evalctx'
}{
	\instantiate{d}{u}{\dcasttwo{\evalctx}{\htau_1}{\htau_2}} = 
	\dcasttwo{\evalctx'}{\htau_1}{\htau_2}
}

\inferrule[EFillFailedCast]{
	\instantiate{d}{u}{\evalctx} = \evalctx'
}{
	\instantiate{d}{u}{\dcastfail{\evalctx}{\htau_1}{\htau_2}} = 
	\dcastfail{\evalctx'}{\htau_1}{\htau_2}
}
\end{mathpar}
\CaptionLabel{Evaluation Context Filling}{fig:evalctx-filling}
\end{figure}

We also need the following lemmas, which characterize how hole filling interacts with substitution and instruction transitions. 
\begin{lem}[Substitution Commutativity]
  If
  \begin{enumerate}[nolistsep]
  	\item $\hasType{\hDelta, \Dbinding{u}{\hGamma'}{\htau'}}{x : \htau_2}{\dexp_1}{\tau}$ and
  	\item $\hasType{\hDelta, \Dbinding{u}{\hGamma'}{\htau'}}{\emptyset}{\dexp_2}{\htau_2}$ and
  	\item $\hasType{\hDelta}{\hGamma'}{\dexp'}{\htau'}$
  \end{enumerate}

  then  $\instantiate{d'}{u}{\substitute{d_2}{x}{d_1}} = \substitute{\instantiate{d'}{u}{d_2}}{x}{\instantiate{d'}{u}{d_1}}$.
\end{lem}
\begin{proof}
We proceed by structural induction on $d_1$ and rule induction on the typing premises, which serve to ensure that the free
variables in $d'$ are accounted for by every closure for $u$.
\end{proof}

\begin{lem}[Instruction Commutativity]
  If
  \begin{enumerate}[nolistsep]
  	\item $\hasType{\hDelta, \Dbinding{u}{\hGamma'}{\htau'}}{\emptyset}{\dexp_1}{\tau}$ and
  	\item $\hasType{\hDelta}{\hGamma'}{\dexp'}{\htau'}$ and
  	\item $\reducesE{}{\dexp_1}{\dexp_2}$
  \end{enumerate}

  then $\reducesE{}{\instantiate{\dexp'}{u}{\dexp_1}}
                     {\instantiate{\dexp'}{u}{\dexp_2}}$.
\end{lem}
\begin{proof}
We proceed by cases on the instruction transition assumption (no induction is needed). For Rule \rulename{ITLam}, we defer to the Substitution Commutativity lemma above. For the remaining cases, the conclusion follows from the definition of hole filling.
\end{proof}

\begin{lem}[Filling Totality]
Either $\inhole{u}{\evalctx}$ or $\instantiate{d}{u}{\evalctx}=\evalctx'$ for some $\evalctx'$.
\end{lem}
\begin{proof} We proceed by structural induction on $\evalctx$. Every case is handled by one of the two judgements. \end{proof}

\begin{lem}[Discarding] If
	\begin{enumerate}[nolistsep]
	\item $\selectEvalCtx{d_1}{\evalctx}{\dexp_1'}$ and
	\item $\selectEvalCtx{d_2}{\evalctx}{\dexp_2'}$ and
	\item $\inhole{u}{\evalctx}$
	\end{enumerate}

	then $\instantiate{d}{u}{d_1} = \instantiate{d}{u}{d_2}$.
\end{lem}
\begin{proof} We proceed by structural induction on $\evalctx$ and rule induction on all three assumptions. Each case follows from the definition of instruction selection and hole filling. \end{proof}

\begin{lem}[Filling Distribution] If
	$\selectEvalCtx{d_1}{\evalctx}{d_1'}$ and $\instantiate{d}{u}{\evalctx}=\evalctx'$ then $\selectEvalCtx{\instantiate{d}{u}{d_1}}{\evalctx'}{\instantiate{d}{u}{d_1'}}$.
\end{lem}
\begin{proof} We proceed by rule induction on both assumptions. Each case follows from the definition of instruction selection and hole filling. \end{proof}

\begin{thm}[Commutativity]
  If
  \begin{enumerate}[nolistsep]
  \item $\hasType{\hDelta, \Dbinding{u}{\hGamma'}{\htau'}}{\emptyset}{\dexp_1}{\tau}$ and
  \item $\hasType{\hDelta}{\hGamma'}{\dexp'}{\htau'}$ and
  \item $\multiStepsTo{\dexp_1}{\dexp_2}$
\end{enumerate}

  then $\multiStepsTo{\instantiate{\dexp'}{u}{\dexp_1}}
                     {\instantiate{\dexp'}{u}{\dexp_2}}$.
\end{thm}
\begin{proof}
By rule induction on assumption (3). The reflexive case is immediate. In the inductive case, we proceed by rule induction on the stepping premise. There is one rule, Rule~\rulename{Step}. By Filling Totality, either $\inhole{u}{\evalctx}$ or $\instantiate{d}{u}{\evalctx} = \evalctx'$. In the former case, by Discarding, we can conclude by \rulename{MultiStepRefl}. In the latter case, by Instruction Commutativity and Filling Distribution we can take a \rulename{Step}, and we can conclude via \rulename{MultiStepSteps} by applying Filling, Preservation and then the induction hypothesis.
\end{proof}

% We excluded these proofs and definitions from the Agda mechanization
% for two reasons.
% %
% First, fill-and-resume is merely an optimization, and unlike the meta
% theory of \Secref{sec:calculus}, these properties are generally not
% conserved by certain reasonable extensions of the core
% calculus~(e.g., reference cells and other non-commuting effects).
% %
% Second, to properly encode the hole filling operation, such a
% mechanization requires a more complex representation of
% hole environments; unfortunately, Agda cannot be easily convinced that
% the definition of hole filling is well-founded (\citet{Nanevski2008}
% establish that it is in fact well-founded).
% %
% By contrast, the developments in \Secref{sec:calculus} do not require
% these more complex representations. 

\subsection{Confluence and Resumption}\label{sec:confluence}
There are various ways to encode the intuition that ``evaluation order does not matter''. One way to do so is
by establishing a confluence property (which is closely related to
the Church-Rosser property \cite{church1936some}).

The most general confluence property does not hold for the dynamic
semantics in Sec.~\ref{sec:calculus} for the usual reason: we do not
reduce under binders (\citet{DBLP:conf/birthday/BlancLM05} discuss the
standard counterexample).
We could recover confluence by specifying reduction under binders,
either generally or in a more restricted form where only closed
sub-expressions are
reduced \cite{DBLP:journals/tcs/CagmanH98,DBLP:conf/birthday/BlancLM05,levy1999explicit}.
However, reduction under binders conflicts with the standard implementation approaches
for most programming languages \cite{DBLP:conf/birthday/BlancLM05}.
A more satisfying approach considers confluence modulo equality \cite{Huet:1980ng}.
The simplest such approach restricts our interest to top-level expressions
of base type that result in values, in which case the following
special case of confluence does hold (trivially when the only base
type has a single value, but also more generally for other base
types).
\begin{lem}[Base Confluence]
  If $\hasType{\Delta}{\emptyset}{\dexp}{b}$ and
  $\multiStepsTo{\dexp}{\dexp_1}$
  and $\isValue{\dexp_1}$
  and $\multiStepsTo{\dexp}{\dexp_2}$
  then $\multiStepsTo{\dexp_2}{\dexp_1}$.
\end{lem}
We can then prove the following property, which establishes that fill-and-resume is sound.
\begin{thm}[Resumption]
  If $\hasType{\hDelta, \Dbinding{u}{\hGamma'}{\htau'}}{\emptyset}{\dexp}{b}$
  and $\hasType{\hDelta}{\hGamma'}{\dexp'}{\htau'}$
  and $\multiStepsTo{\dexp}{\dexp_1}$
  and $\multiStepsTo{\instantiate{\dexp'}{u}{\dexp}}{\dexp_2}$
  and $\isValue{\dexp_2}$
  then $\multiStepsTo{\instantiate{\dexp'}{u}{\dexp_1}}{\dexp_2}$.
  \begin{proof}
    By Commutativity,
    $\multiStepsTo{\instantiate{\dexp'}{u}{\dexp}}
                  {\instantiate{\dexp'}{u}{\dexp_1}}$.
    By Base Confluence, we can conclude.
  \end{proof}
\end{thm}

% !TEX root = hazelnut-dynamics.tex
\newcommand{\extensionsSec}{Extensions to Hazelnut Live}
\section{\protect\extensionsSec} % don't like the all-caps thing that the template does, so protecting it from that
\label{sec:extensions}

We give two extensions here, numbers and sum types, to maintain parity with \Hazelnut as specified by \citet{popl-paper}.

It is worth observing that these extensions do not make explicit mention of expression holes. The ``non-obvious'' machinery is almost entirely related to casts. Fortunately, there has been progress on the problem of generating ``gradualized'' specifications from standard specifications \cite{DBLP:conf/popl/CiminiS16,DBLP:conf/popl/CiminiS17}. The extensions below, and other extensions of interest, closely follow the output of the gradualizer: \url{http://cimini.info/gradualizerDynamicSemantics/} (which, like our work, is based on the refined account of gradual typing by \cite{DBLP:conf/snapl/SiekVCB15}). The rules for indeterminate forms are mainly where the gradualizer is not sufficient. We leave to future work the related question of automatically generating a hole-aware static and dynamic semantics from a standard language specification.

% !TEX root = hazelnut-dynamics.tex

%%%%%%%% %%%%%%%%% %%%%%%%%% %%%%%%%%% %%%%%%%%% %%%%%%%%% %%%%%%%%% %%%%%%%%% %%%%%%%%% %%%%%%%%%
%%%%%%% Syntax
\subsection{Numbers}
We extend the syntax as follows, assuming $n$ ranges over mathematical numbers of suitable sort:
\[
\begin{array}{rllllll}
\mathsf{HTyp} & \htau & ::= \cdots ~\vert~ \tnum
\\
\mathsf{HExp} & \hexp & ::= \cdots
~\vert~ \hnum{n}
~\vert~ \hadd{\hexp}{\hexp}
\\
\mathsf{IHExp} & \dexp & ::= \cdots
~\vert~ \dnum{n}
~\vert~ \dadd{\dexp}{\dexp}
\end{array}
\]

%%%%%%%% %%%%%%%%% %%%%%%%%% %%%%%%%%% %%%%%%%%% %%%%%%%%% %%%%%%%%% %%%%%%%%% %%%%%%%%%

\vsepRule

\judgbox{
  \hsyn{\hGamma}{\hexp}{\htau}
}{$\hexp$ synthesizes type $\htau$}
\begin{mathpar}
\inferrule[]{~}{\hsyn{\hGamma}{\hnum{n}}{\tnum}}

\inferrule[]{
  \hana{\hGamma}{\hexp_1}{\tnum}\\
  \hana{\hGamma}{\hexp_2}{\tnum}
}{
  \hsyn{\hGamma}{\hadd{\hexp_1}{\hexp_2}}{\tnum}
}
\end{mathpar}

\vsepRule

\judgbox
  {\elabSyn{\hGamma}{\hexp}{\htau}{\dexp}{\Delta}}
  {$\hexp$ synthesizes type $\htau$ and elaborates to $\dexp$}
\begin{mathpar}
\inferrule[]{ ~ }{
  \elabSyn{\hGamma}{\hnum{n}}{\tnum}{\hnum{n}}{\EmptyDelta}
}

%% \inferrule[]{
%%   \elabSyn{\hGamma}{\hexp_1}{\tnum}{\dexp_1}{\Delta_1}
%%   \\
%%   \elabSyn{\hGamma}{\hexp_2}{\tnum}{\dexp_2}{\Delta_1}
%% }{
%%   \elabSyn{\hGamma}
%%             {\dadd{\hexp_1}{\hexp_2}}
%%             {\tnum}
%%             {\dadd{\dexp_1}{\dexp_2}}
%%             {\Delta_1 \cup \Delta_2}
%% }

\inferrule[]{
  \elabAna{\hGamma}{\hexp_1}{\tnum}{\dexp_1}{\tnum}{\Delta_1}
  \\
  \elabAna{\hGamma}{\hexp_2}{\tnum}{\dexp_2}{\tnum}{\Delta_2}
}{
  \elabSyn{\hGamma}
            {\dadd{\hexp_1}{\hexp_2}}
            {\tnum}
            {\dadd{\dexp_1}{\dexp_2}}
            {\Delta_1 \cup \Delta_2}
}
\end{mathpar}

%%%%%%%%% %%%%%%%%% %%%%%%%%% %%%%%%%%% %%%%%%%%% %%%%%%%%% %%%%%%%%% %%%%%%%%% %%%%%%%%% %%%%%%%%% %%%%%%%%%

\judgbox{\hasType{\Delta}{\hGamma}{\dexp}{\htau}}{$\dexp$ is assigned type $\htau$}
\begin{mathpar}
\inferrule[]{
  ~
}{
  \hasType{\Delta}{\hGamma}{\hnum{n}}{\tnum}
}

\inferrule[]{
  \hasType{\Delta}{\hGamma}{\dexp_1}{\tnum}
  \\
  \hasType{\Delta}{\hGamma}{\dexp_2}{\tnum}
}{
  \hasType{\Delta}{\hGamma}{\dadd{\dexp_1}{\dexp_2}}{\tnum}
}

\end{mathpar}

%%%%%%%%% %%%%%%%%% %%%%%%%%% %%%%%%%%% %%%%%%%%% %%%%%%%%% %%%%%%%%% %%%%%%%%% %%%%%%%%% %%%%%%%%%

\judgbox{\isValue{\dexp}}{$\dexp$ is a value}

\begin{mathpar}
\inferrule[]
{~}
{\isValue{\dnum{n}}}
\end{mathpar}

\vsepRule

\judgbox{\isGround{\htau}}{$\htau$ is a ground type}
\begin{mathpar}
\inferrule[]{~}{
  \isGround{\tnum}
}
\end{mathpar}

\vsepRule

\judgbox{\isIndet{\dexp}}{$\dexp$ is indeterminate}
\begin{mathpar}
\inferrule[]
{
 \isIndet{\dexp_1}
 \\
 \isFinal{\dexp_2}
}
{\isIndet{\dadd{\dexp_1}{\dexp_2}}}

\inferrule[]
{
 \isFinal{\dexp_1}
 \\
 \isIndet{\dexp_2}
}
{\isIndet{\dadd{\dexp_1}{\dexp_2}}}
\end{mathpar}

%%%%%%%% %%%%%%%%% %%%%%%%%% %%%%%%%%% %%%%%%%%% %%%%%%%%% %%%%%%%%% %%%%%%%%% %%%%%%%%% %%%%%%%%%

\begin{mathpar}
\arraycolsep=4pt\begin{array}{rllllll}
\mathsf{EvalCtx} & \evalctx & ::= & \cdots ~\vert~
  \dadd{\evalctx}{\dexp_2}
  ~\vert~
  \dadd{\dexp_1}{\evalctx}
\end{array}
\end{mathpar}

%%%%%%%% %%%%%%%%% %%%%%%%%% %%%%%%%%% %%%%%%%%% %%%%%%%%% %%%%%%%%% %%%%%%%%% %%%%%%%%% %%%%%%%%%

\judgbox{\selectEvalCtx{\dexp}{\evalctx}{\dexp'}}{$\dexp$ is obtained by placing $\dexp'$ at the mark in $\evalctx$}
\begin{mathpar}
\inferrule[]
{\selectEvalCtx{\dexp_1}{\evalctx}{\dexp_1'}}
{\selectEvalCtx{\dadd{\dexp_1}{\dexp_2}}
               {(\dadd{\evalctx}{\dexp_2})}
               {\dexp_1'}}

\inferrule[]
{
\maybePremise{\isFinal{\dexp_1}}\\
\selectEvalCtx{\dexp_2}{\evalctx}{\dexp_2'}}
{\selectEvalCtx{\dadd{\dexp_1}{\dexp_2}}
               {(\dadd{\dexp_1}{\evalctx})}
               {\dexp_2'}}
\end{mathpar}

%%%%%%%% %%%%%%%%% %%%%%%%%% %%%%%%%%% %%%%%%%%% %%%%%%%%% %%%%%%%%% %%%%%%%%% %%%%%%%%% %%%%%%%%%

\judgbox{\reducesE{}{\dexp_1}{\dexp_2}}{$\dexp_1$ transitions to $\dexp_2$}
\begin{mathpar}
\inferrule[]
{ n_1 + n_2 = n_3 }
{\reducesE{\Delta}
  {\dadd{\dnum{n_1}}{\dnum{n_2}}}
  {\dnum{n_3}}
}
\end{mathpar}
 %(from Hazelnut)

%\subsection{Product Types}
%\include{prodtypes}

% !TEX root = hazelnut-dynamics.tex

\subsection{Sum Types}

%%%%%%%% %%%%%%%%% %%%%%%%%% %%%%%%%%% %%%%%%%%% %%%%%%%%% %%%%%%%%% %%%%%%%%% %%%%%%%%% %%%%%%%%%
%%%%%%% Syntax
We extend the syntax for sum types as follows:
\[
\begin{array}{rllllll}
\mathsf{HTyp} & \htau & ::= \cdots ~\vert~ {\htau + \htau} &
\\
\mathsf{HExp} & \hexp & ::= \cdots
~\vert~ \hinL{\hexp}
~\vert~ \hinR{\hexp}
%~\vert~ \hcase{\hexp}{\hinL{x}}{\hexp}{\hinR{x}}{\hexp}
~\vert~ \hcase{\hexp}{x}{\hexp}{y}{\hexp}
\\
\mathsf{IHExp} & \dexp & ::= \cdots
~\vert~ \dinL{\htau}{\dexp}
~\vert~ \dinR{\htau}{\dexp}
%~\vert~ \dcase{\dexp}{\hinL{x}}{\dexp}{\hinR{x}}{\dexp}
~\vert~ \dcase{\dexp}{x}{\dexp}{y}{\dexp}
\end{array}
\]

%%%%%%%% %%%%%%%%% %%%%%%%%% %%%%%%%%% %%%%%%%%% %%%%%%%%% %%%%%%%%% %%%%%%%%% %%%%%%%%% %%%%%%%%%
%%%%%%% Definition of Join: a meta-level, binary function over types.

%\newcommand{\JoinTypes}[2]{\textsf{join}~~#1~~#2}
%\newcommand{\JoinTypes}[2]{\textsf{join}(#1,#2)}

%\fbox{$\JoinTypes{\htau_1}{\htau_2} = \htau_3$}
\judgbox
 {\JoinTypes{\htau_1}{\htau_2} = \htau_3}
 {Types~$\htau_1$ and $\htau_2$ join consistently, forming type~$\htau_3$}
\[
\begin{array}{lcl}
\JoinTypes{\htau}{\htau} &=&  \htau
\\
\JoinTypes{\tehole}{\htau} &=&  \htau
\\
\JoinTypes{\htau}{\tehole} &=&  \htau
\\
\JoinTypes
{\tarr{\htau_1}{\htau_2}}
{\tarr{\htau_1}{\htau_2}}
&=&
\tarr{\JoinTypes{\htau_1}{\htau_2}}
     {\JoinTypes{\htau_1}{\htau_2}}
\\
\JoinTypes
{\tsum{\htau_1}{\htau_2}}
{\tsum{\htau_1}{\htau_2}}
&=&
\tsum{\JoinTypes{\htau_1}{\htau_2}}
     {\JoinTypes{\htau_1}{\htau_2}}
\end{array}
\]

\begin{thm}[Joins]
If $\JoinTypes{\tau_1}{\tau_2} = \tau$
then types $\tau_1$, $\tau_2$ and $\tau$ are pair-wise consistent, i.e.,
$\tconsistent{\tau_1}{\tau_2}$,
$\tconsistent{\tau_1}{\tau}$ and
$\tconsistent{\tau_2}{\tau}$.
\begin{proof}
By induction on the derivation of $\JoinTypes{\tau_1}{\tau_2} = \tau$.
\end{proof}
\end{thm}

%%%%%%%% %%%%%%%%% %%%%%%%%% %%%%%%%%% %%%%%%%%% %%%%%%%%% %%%%%%%%% %%%%%%%%% %%%%%%%%% %%%%%%%%%
\vsepRule

\judgbox
 {\summatch{\htau}{\tsum{\htau_1}{\htau_2}}}
 {Type~$\htau$ has matched sum type~$\tsum{\htau_1}{\htau_2}$}
\begin{mathpar}
\inferrule[]{~}{\summatch{\tsum{\tau_1}{\tau_2}}{\tsum{\tau_1}{\tau_2}}}
\and
\inferrule[]{~}{\summatch{\tehole}{\tsum{\tehole}{\tehole}}}
\end{mathpar}

%%%%%%%%% %%%%%%%%% %%%%%%%%% %%%%%%%%% %%%%%%%%% %%%%%%%%% %%%%%%%%% %%%%%%%%% %%%%%%%%% %%%%%%%%%
\vsepRule

\judgbox
  {\hana{\hGamma}{\hexp}{\htau}}
  {$\hexp$ analyzes against type $\htau$}
\begin{mathpar}
\inferrule[]{
  \summatch{\htau}{\tsum{\htau_1}{\htau_2}}\\
  \hana{\hGamma}{\hexp}{\htau_1}
}{
  \hana{\hGamma}{\hinL{\hexp}}{\htau}
}

\inferrule[]{
  \summatch{\htau}{\tsum{\htau_1}{\htau_2}}\\
  \hana{\hGamma}{\hexp}{\htau_2}
}{
  \hana{\hGamma}{\hinR{\hexp}}{\htau}
}

\inferrule[]{
  \hsyn{\hGamma}{\hexp_1}{\htau_1}\\
  \summatch{\htau_1}{\tsum{\htau_{11}}{\htau_{12}}}\\\\
  \hana{\hGamma, x : \htau_{11}}{\hexp_2}{\htau}\\\\
  \hana{\hGamma, y : \htau_{12}}{\hexp_3}{\htau}
}{
  \hana{\hGamma}{
    \hcase{\hexp_1}{x}{\hexp_2}{y}{\hexp_3}
  }{
    \htau
  }
}
\end{mathpar}

%%%%%%%%% %%%%%%%%% %%%%%%%%% %%%%%%%%% %%%%%%%%% %%%%%%%%% %%%%%%%%% %%%%%%%%% %%%%%%%%% %%%%%%%%%
\vsepRule

\judgbox
  {\elabAna{\hGamma}{\hexp}{\htau_1}{\dexp}{\htau_2}{\Delta}}
  {$\hexp$ analyzes against type $\htau_1$ and
   elaborates to $\dexp$ of consistent type $\htau_2$}
\begin{mathpar}
\inferrule[]{
  \summatch{\htau}{\tsum{\htau_1}{\htau_2}}
  \\
  \elabAna{\hGamma}{\hexp}{\htau_1}{\dexp}{\htau'_1}{\Delta}
}{
  \elabAna{\hGamma}{\hinL{\hexp}}{\htau}{\dinL{\tau_2}{\dexp}}{\tsum{\htau_1'}{\htau_2}}{\Delta}
}

\inferrule[]{
  \summatch{\htau}{\tsum{\htau_1}{\htau_2}}
  \\
  \elabAna{\hGamma}{\hexp}{\htau_2}{\dexp}{\htau'_2}{\Delta}
}{
  \elabAna{\hGamma}{\hinR{\hexp}}{\htau}{\dinL{\tau_1}{\dexp}}{\tsum{\htau_1}{\htau'_2}}{\Delta}
}

\inferrule[]{
  \elabSyn{\hGamma}{\hexp_1}{\htau_1}{\dexp_1}{\Delta_1}
  \\
  \summatch{\htau_1}{\tsum{\htau_{11}}{\htau_{12}}}
  \\\\
  \elabAna{\hGamma, x:\htau_{11}}{\hexp_2}{\htau}{\dexp_2}{\htau_2}{\Delta_2}
  \\
  \elabAna{\hGamma, y:\htau_{12}}{\hexp_3}{\htau}{\dexp_3}{\htau_3}{\Delta_3}
  \\\\
  \JoinTypes{\htau_2}{\htau_3} = {\htau'}
  \\
  \Delta = \Delta_1 \cup \Delta_2 \cup \Delta_3
}{
  \elabAna{\hGamma}
            {\hcase{e_1}{x}{e_2}{y}{e_3}}
            {\htau}
            {\dcase
                {\dcasttwo{d_1}{\htau_1}{\tsum{\htau_{11}}{\htau_{12}}}}
                {x}{\dcasttwo{d_2}{\htau_2}{\htau'}}
                {y}{\dcasttwo{d_3}{\htau_3}{\htau'}}
            }
            {\htau'}
            {\Delta}
            }
\end{mathpar}

%%%%%%%%% %%%%%%%%% %%%%%%%%% %%%%%%%%% %%%%%%%%% %%%%%%%%% %%%%%%%%% %%%%%%%%% %%%%%%%%% %%%%%%%%% %%%%%%%%%

\judgbox{\hasType{\Delta}{\hGamma}{\dexp}{\htau}}{$\dexp$ is assigned type $\htau$}
\begin{mathpar}
\inferrule[]{
  \hasType{\Delta}{\hGamma}{\dexp}{\htau_1}
}{
  \hasType{\Delta}{\hGamma}{\dinL{\tau_2}{\dexp}}{\tsum{\htau_1}{\htau_2}}
}

\inferrule[]{
  \hasType{\Delta}{\hGamma}{\dexp}{\htau_2}
}{
  \hasType{\Delta}{\hGamma}{\dinR{\tau_1}{\dexp}}{\tsum{\htau_1}{\htau_2}}
}

\inferrule[]{
  \hasType{\Delta}{\hGamma}{\dexp_1}{\tsum{\htau_1}{\htau_2}}
  \\
  \hasType{\Delta}{\hGamma,x:\htau_1}{\dexp_2}{\htau}
  \\
  \hasType{\Delta}{\hGamma,y:\htau_2}{\dexp_3}{\htau}
}{
  \hasType{\Delta}{\hGamma}{\dcase{\dexp_1}{x}{\dexp_2}{y}{\dexp_3}}{\htau}
}
\end{mathpar}

%%%%%%%%% %%%%%%%%% %%%%%%%%% %%%%%%%%% %%%%%%%%% %%%%%%%%% %%%%%%%%% %%%%%%%%% %%%%%%%%% %%%%%%%%%

\judgbox{\isValue{\dexp}}{$\dexp$ is a value}
\begin{mathpar}
\inferrule[]
{\isValue{\dexp}}
{\isValue{\dinL{\htau}{\dexp}}}

\inferrule[]
{\isValue{\dexp}}
{\isValue{\dinR{\htau}{\dexp}}}
\end{mathpar}

\vsepRule

\judgbox{\isGround{\htau}}{$\htau$ is a ground type}
\begin{mathpar}
\inferrule[]{~}{
  \isGround{\tsum{\tehole}{\tehole}}
}
\end{mathpar}

\judgbox{\groundmatch{\htau}{\htau'}}{$\htau$ has matched ground type $\htau'$}
\begin{mathpar}
\inferrule[]{
  \tsum{\htau_1}{\htau_2}\neq\tsum{\tehole}{\tehole}
}{
  \groundmatch{\tsum{\htau_1}{\htau_2}}{\tsum{\tehole}{\tehole}}
}
\end{mathpar}

\vsepRule

\judgbox{\isBoxedValue{\dexp}}{$\dexp$ is a boxed value}
\begin{mathpar}
\inferrule[]
{\isBoxedValue{\dexp}}
{\isBoxedValue{\dinL{\htau}{\dexp}}}

\inferrule[]
{\isBoxedValue{\dexp}}
{\isBoxedValue{\dinR{\htau}{\dexp}}}

\inferrule[]
{\tsum{\htau_1}{\htau_2} \ne
 \tsum{\htau_1'}{\htau_2'}
 \\
 \isBoxedValue{\dexp}
}
{\isBoxedValue{\dcasttwo{\dexp}
    {\tsum{\htau_1}{\htau_2}}
    {\tsum{\htau_1'}{\htau_2'}}
}}
\end{mathpar}

\vsepRule

\judgbox{\isIndet{\dexp}}{$\dexp$ is indeterminate}
\begin{mathpar}
\inferrule[]
{\isIndet{\dexp}}
{\isIndet{\dinL{\htau}{\dexp}}}

\inferrule[]
{\isIndet{\dexp}}
{\isIndet{\dinR{\htau}{\dexp}}}

\inferrule[]
{\tsum{\htau_1}{\htau_2} \ne
 \tsum{\htau_1'}{\htau_2'}
 \\
 \isIndet{\dexp}
}
{\isIndet{\dcasttwo{\dexp}
    {\tsum{\htau_1}{\htau_2}}
    {\tsum{\htau_1'}{\htau_2'}}
}}

\inferrule[]
{
  \dexp_1 \ne \dinL{\tau}{\dexp_1'}
  \\
  \dexp_1 \ne \dinR{\tau}{\dexp_1'}
  \\
  \dexp_1 \ne \dcasttwo{\dexp_1'}{\tsum{\htau_1}{\htau_2}}{\tsum{\htau_1'}{\htau_2'}}
  \\
  \isIndet{\dexp_1}
}
{
  \dcase{\dexp_1}{x}{\dexp_2}{y}{\dexp_3}
}
\end{mathpar}

%%%%%%%% %%%%%%%%% %%%%%%%%% %%%%%%%%% %%%%%%%%% %%%%%%%%% %%%%%%%%% %%%%%%%%% %%%%%%%%% %%%%%%%%%

\begin{mathpar}
\arraycolsep=4pt\begin{array}{rllllll}
\mathsf{EvalCtx} & \evalctx & ::= & \cdots ~\vert~
  \dinL{\htau}{\evalctx}
  ~\vert~
  \dinR{\htau}{\evalctx}
  ~\vert~
  \dcase{\evalctx}{x}{\dexp_1}{y}{\dexp_2}
  %% \evalhole ~\vert~
  %% \hap{\evalctx}{\dexp} ~\vert~
  %% \hap{\dexp}{\evalctx} ~\vert~
  %% \dhole{\evalctx}{\mvar}{\subst}{} ~\vert~
  %% \dcasttwo{\evalctx}{\htau}{\htau} ~\vert~
  %% \dcastfail{\evalctx}{\htau}{\htau}
\end{array}
\end{mathpar}

%% \judgbox{\isevalctx{\evalctx}}{$\evalctx$ is an evaluation context}
%% \begin{mathpar}
%% \inferrule[]{\isevalctx{\evalctx}}{\isevalctx{\dinL{\htau}{\evalctx}}}
%% \and
%% \inferrule[]{\isevalctx{\evalctx}}{\isevalctx{\dinR{\htau}{\evalctx}}}
%% \and
%% \inferrule[]{\isevalctx{\evalctx}}{\isevalctx{\dcase{\evalctx}{x}{\dexp_1}{y}{\dexp_2}}}
%% \end{mathpar}

\judgbox{\selectEvalCtx{\dexp}{\evalctx}{\dexp'}}{$\dexp$ is obtained by placing $\dexp'$ at the mark in $\evalctx$}
\begin{mathpar}
\inferrule[]
{\selectEvalCtx{\dexp_1}{\evalctx}{\dexp_1'}}
{\selectEvalCtx{\dcase{\dexp_1}{x}{\dexp_2}{y}{\dexp_3}}
               {\dcase{\evalctx}{x}{\dexp_2}{y}{\dexp_3}}{\dexp_1'}}
\end{mathpar}

\judgbox{\reducesE{\Delta}{\dexp_1}{\dexp_2}}{$\dexp_1$ transitions to $\dexp_2$}
\begin{mathpar}
\inferrule[]
{\maybePremise{\isFinal{\dexp_1}}}
{\reducesE{\Delta}
  {\dcase{\dinL{\htau}{\dexp_1}}{x}{\dexp_2}{y}{\dexp_3}}
  {\DoSubst{\dexp_1}{x}{\dexp_2}}}

\inferrule[]
{\maybePremise{\isFinal{\dexp_1}}}
{\reducesE{\Delta}
  {\dcase{\dinR{\htau}{\dexp_1}}{x}{\dexp_2}{y}{\dexp_3}}
  {\DoSubst{\dexp_1}{y}{\dexp_3}}}

\inferrule[]
{\maybePremise{\isFinal{\dexp_1}}}
{\reducesE{\Delta}
  {\dcase
    {\dcasttwo{\dexp_1}{\tsum{\tau_1}{\tau_2}}{\tsum{\tau_1'}{\tau_2'}}}
    {x}{\dexp_2}{y}{\dexp_3}}
  {\dcase
    {\dexp_1}
    {x}{\DoSubst{\dcasttwo{x}{\tau_1}{\tau_1'}}{x}{\dexp_2}}
    {y}{\DoSubst{\dcasttwo{y}{\tau_2}{\tau_2'}}{y}{\dexp_3}}}
}
\end{mathpar}

%% \begin{thm}[Canonical value forms--Sums]
%% If $\hasType{\Delta}{\EmptyhGamma}{\dexp}{\tsum{\htau}{\htau}}$
%% and $\isValue{\dexp}$ then either
%% \begin{enumerate}
%% \item
%% $\dexp = \dinL{\tau_2}{\dexp'}$
%% where $\isValue{\dexp'}$
%% and $\hasType{\Delta}{\EmptyhGamma}{\dexp'}{\htau_1}$
%% \item
%% $\dexp = \dinR{\tau_1}{\dexp'}$
%% where $\isValue{\dexp'}$
%% and $\hasType{\Delta}{\EmptyhGamma}{\dexp'}{\htau_1}$
%% \end{enumerate}
%% \end{thm}

%% \begin{thm}[Canonical boxed value forms--Sums]
%% If $\hasType{\Delta}{\EmptyhGamma}{\dexp}{\tsum{\htau}{\htau}}$
%% and $\isValue{\dexp}$ then either
%% \begin{enumerate}
%% \item
%% $\dexp = \dinL{\tau_2}{\dexp'}$
%% where $\isValue{\dexp'}$
%% and $\hasType{\Delta}{\EmptyhGamma}{\dexp'}{\htau_1}$
%% \item
%% $\dexp = \dinR{\tau_1}{\dexp'}$
%% where $\isValue{\dexp'}$
%% and $\hasType{\Delta}{\EmptyhGamma}{\dexp'}{\htau_1}$
%% \end{enumerate}
%% \end{thm}

%% %%%%%%%

%\subsection{Recursive Types}
%\subsection{Polymorphism}
%\subsection{Mutable References}
%\cy{if we add references, would have to talk about how commutativity isn't conserved}

\else

\fi

\end{document}